\documentclass[onecolumn,a4paper]{IEEEtran}

\usepackage{pgfplots}
\pgfplotsset{compat=newest}
\usepackage{tikz}
\usetikzlibrary{arrows,arrows.meta,matrix,positioning,positioning,calc}

\usepackage[utf8]{inputenc}

\usepackage[english]{babel}
\usepackage[T1]{fontenc}
\usepackage{epsfig}
\usepackage{amsmath, amssymb, amsbsy}
\usepackage{mathtools}
\usepackage{mathdots}
\usepackage{xspace}
\usepackage[noend]{algpseudocode}
\usepackage[linesnumbered,ruled,vlined,titlenumbered]{algorithm2e}
\usepackage{algorithmicx}
\usepackage{color}
\usepackage{cite}
\usepackage{booktabs}
\usepackage{verbatim}
\usepackage{lipsum}
\usepackage{enumitem}
\usepackage{colortbl}
\usepackage{amsthm}
\usepackage{nicefrac}
\usepackage{flushend}

\makeatletter
\newcommand\footnoteref[1]{\protected@xdef\@thefnmark{\ref{#1}}\@footnotemark}
\makeatother

\newtheorem{theorem}{Theorem}
\newtheorem{lemma}[theorem]{Lemma}
\newtheorem{corollary}[theorem]{Corollary}

\newtheorem{example}[theorem]{Example}
\newtheorem{remark}[theorem]{Remark}

\newtheorem{definition}{Definition}

\newenvironment{mymatrix}{\begin{bmatrix}} {\end{bmatrix} }

\def\ve#1{{\mathchoice{\mbox{\boldmath$\displaystyle #1$}}%
              {\mbox{\boldmath$\textstyle #1$}}%
              {\mbox{\boldmath$\scriptstyle #1$}}%
              {\mbox{\boldmath$\scriptscriptstyle #1$}}}}

\newcommand{\Fq}{\ensuremath{\mathbb{F}_q}}

\DeclareMathOperator{\rank}{rk}

\newcommand{\Code}{\mathcal{C}}
\newcommand{\code}{\mathcal{C}}

\newcommand{\0}{\ve{0}}
\renewcommand{\S}{\ve{S}}
\renewcommand{\H}{\ve{H}}

\newcommand{\e}{\ve{e}}

\renewcommand{\c}{\ve{c}}
\renewcommand{\e}{\ve{e}}

\newcommand{\G}{\ve{G}}
\newcommand{\Q}{\ve{Q}}
\renewcommand{\P}{\ve{P}}

\newcommand{\E}{\ve{E}}

\newcommand{\A}{\ve{A}}
\newcommand{\B}{\ve{B}}

\newcommand{\dH}{\mathrm{d}_\mathrm{H}}

\newcommand{\Eset}{\mathcal{E}}
\newcommand{\C}{\ve{C}}
\newcommand{\R}{\ve{R}}
\newcommand{\Sset}{\mathbb{S}}
\newcommand{\Scal}{\mathcal{S}}
\newcommand{\cS}{\mathcal{S}}
\newcommand{\Rset}{\mathcal{R}}

\newcommand{\numgroups}{\mu}

\DeclareMathOperator{\supp}{supp}

\newcommand{\Pfailure}{P_\mathrm{failure}}
\newcommand{\Pmiscorrection}{P_\mathrm{miscorrection}}

\newcommand{\NN}{\mathbb{N}}
\newcommand{\IC}{\mathcal{IC}}

\newcolumntype{C}{>{$}c<{$}} 

\newcommand{\algoref}[1]{Algorithm~\ref{#1}}

\newcommand{\figref}[1]{Figure~\ref{#1}}
\newcommand{\stepref}[1]{Step~\ref{#1}}
\newcommand{\tabref}[1]{Table~\ref{#1}}

\newcommand{\ceil}[1]{\ensuremath{\left\lceil #1 \right\rceil}}
\newcommand{\floor}[1]{\ensuremath{\left\lfloor #1 \right\rfloor}}

\newcommand{\F}[1]{\ensuremath{\mathbb{F}_{#1}}}

\newcommand{\dt}[2]{\ensuremath{\text{d}_{\text{H}} (#1 , #2) }}

\DeclareMathOperator{\wt}{wt}

\newcommand{\omegaVec}{\ve{\omega}}
\definecolor{mygreen}{rgb}{0,0.7,0}

\newcommand{\Psucc}{\mathrm{Pr}_{\mathrm{UDS}}}
\newcommand{\tauIRS}{\tau_{\mathrm{IRS,max}}}
\newcommand{\res}{\mathrm{res}}

\definecolor{darkgreen}{rgb}{0,0.7,0}

\makeatletter
\newcommand{\removelatexerror}{\let\@latex@error\@gobble}
\makeatother

\newcommand{\printalgoIEEE}[1]
{
\vspace{0.3cm}
{\centering
\scalebox{0.97}{
\removelatexerror
\begin{tabular}{p{\columnwidth}}
\begin{algorithm}[H]
 \begin{small}
 #1
 \end{small}
\end{algorithm}
\end{tabular}
}
}
\vspace{-0.3cm}
}

\IEEEoverridecommandlockouts

\begin{document}

\title{Error Decoding of Locally Repairable and\\ Partial MDS Codes}
\author{\IEEEauthorblockN{Lukas Holzbaur, Sven Puchinger, Antonia Wachter-Zeh}
	\thanks{The work of L. Holzbaur was supported by the German Research Foundation (Deutsche Forschungsgemeinschaft, DFG) under Grant No. WA3907/1-1.

	S. Puchinger has received funding from the European Union's Horizon 2020 research and innovation program under the Marie Sklodowska-Curie grant agreement no.~713683. This work was partly done while S.~Puchinger was at Technical University of Munich, where he was supported by the German Israeli Project Cooperation (DIP) grant no.~KR3517/9-1.

	Parts of this paper have been presented at the \emph{2018 IEEE International Symposium on Information Theory (ISIT)} \cite{Holzbaur2018} and \emph{2019 IEEE Information Theory Workshop (ITW)} \cite{Holzbaur2019}.

		L. Holzbaur and A. Wachter-Zeh are with the Institute for Communications Engineering, Technical University of Munich, Germany. S. Puchinger is with the Department of Applied Mathematics and Computer Science, Technical University of Denmark (DTU), Denmark.

Emails: \{lukas.holzbaur, antonia.wachter-zeh\}@tum.de, svepu@dtu.dk}}

\maketitle

\begin{abstract}
  In this work it is shown that locally repairable codes (LRCs) can be list-decoded efficiently beyond the Johnson radius for a large range of parameters by utilizing the local error-correction capabilities. The corresponding decoding radius is derived and the asymptotic behavior is analyzed. A general list-decoding algorithm for LRCs that achieves this radius is proposed along with an explicit realization for LRCs that are subcodes of Reed--Solomon codes (such as, e.g., Tamo--Barg LRCs). Further, a probabilistic algorithm of low complexity for unique decoding of LRCs is given and its success probability is analyzed.
  The second part of this work considers error decoding of LRCs and partial maximum distance separable (PMDS) codes through interleaved decoding. For a specific class of LRCs the success probability of interleaved decoding is investigated. For PMDS codes, it is shown that there is a wide range of parameters for which interleaved decoding can increase their decoding radius beyond the minimum distance such that the probability of successful decoding approaches $1$ when the code length goes to infinity.
\end{abstract}


\section{Introduction}

Vast growth in the popularity of cloud storage and other cloud services in recent years has led to an increased interest in coding solutions for distributed data storage. Traditionally, protection against the failure of nodes/servers was achieved by replicating the contents of a node multiple times. While this results in very efficient repair in the case of a node failure, i.e., a failed node can be replaced by simply replicating any of the remaining nodes, the required storage overhead makes this method unattractive for large scale data centers. Instead, some system operators have made the transition to MDS-coded storage. While these codes offer the optimal trade-off between the number of recoverable nodes and storage overhead, the cost of recovering/replacing a failed node in terms of total network traffic and number of nodes involved is in general far from optimal. Different approaches have been considered to reduce this cost of node recovery, with special attention to the more likely event of a single or very few node failures, as efficient repair for those cases is especially important. The most prominent approaches addressing this issue are \emph{regenerating codes} \cite{Dimakis2010,Dimakis2011,Rashmi2012,Tamo2017}, which aim to decrease the network traffic required for repair, and \emph{locally repairable codes} (LRC) \cite{Huang2007,Huang2012,Gopalan2012,Sathiamoorthy2013,Kamath2014,Papailiopoulos2014,Tamo2014,Silberstein2015}, which limit the number of nodes involved in the repair. In \cite{Tamo2014}, a family of LRCs, popular for the small required field size and for fulfilling the Singleton-like bound on the distance \cite{Gopalan2012,Kamath2014}, was constructed as subcodes of (generalized) Reed-Solomon (GRS) codes. In addition, \emph{partial maximum distance separable} (PMDS) codes\footnote{In \cite{Huang2007, Chen2007, Gopalan2014,gopalan2017,martinez2019universal} these codes are referred to as \emph{maximally recoverable codes}, the definitions are equivalent.} \cite{Huang2007, Chen2007, Gopalan2014,Blaum2013,gabrys2018constructions,Blaum2016,calis2016general,gopalan2017,Horlemann-Trautmann2017,martinez2019universal} fulfill an even stronger notion of locality by requiring that any information-theoretically solvable erasure pattern can be recovered and have been proposed for use in distributed storage systems to further decrease the probability of data loss.

The main motivation of storage codes such as LRCs and regenerating codes is erasure correction, as these occur naturally in distributed storage systems whenever nodes fail, e.g., due to hardware failures, power outages, or maintenance. It is often assumed that errors are detected  \cite{Blaum2013}, e.g., by a cyclic redundancy check (CRC), thereby turning errors into erasures. Such a storage system can also be viewed as a concatenated coding scheme where the inner code is used solely for error detection, and the outer code for recovery of erasures \cite{roth2014coding}. While this declaration of erasures is likely to be successful for some causes of errors, such as, e.g., faulty sectors on a hard-drive or solid-state drive, errors caused by faulty synchronization or bad links between the nodes cannot be detected on these lower levels. These events result in \emph{errors}, i.e., events where the position of occurrence is unknown, which is what we consider in this work. More specifically, such an error event is likely to corrupt a large number of symbols, which results in a burst of errors, a fact that we will exploit in the second part of this work to increase the error correction capability of some LRCs and PMDS codes. As erasures are a far bigger concern than errors in distributed storage systems, where these code classes are of interest, we would like to emphasize that the proposed methods focus either on LRCs/PMDS codes in general or on popular classes of LRCs, i.e., subcodes of GRS codes such as Tamo--Barg codes \cite{Tamo2014}, and \emph{do not require a change in the structure of the codes}. Hence, they can be viewed as a worst-case measure that can be employed as a last resort in the case of error events, without any increase in costs, e.g., storage overhead, for the system.

There have been several previous works that consider error correction from storage codes such as LRCs and regenerating codes. In \cite{Silberstein2015,Pawar2011,Han2012} the authors consider a concatenated structure with LRCs or regenerating codes as inner codes and rank metric codes as outer codes to protect against adversaries of different types. Regenerating codes with an error tolerance in the repair process were considered in \cite{Rashmi2012,Pawar2011,Han2012}.
In \cite{Dikaliotis2010} a hashing scheme is proposed to detect errors and protect against adversarial nodes. Efficient repair of nodes by error correction from parts of the received word was considered in \cite{Tamo2017}. 
To the best of our knowledge, this is the first work which utilizes the additional redundancy accounting for the locality in order to increase the \emph{error} decoding radius beyond the unique decoding radius of the code. Besides the practical implications, the presented results are especially of theoretical interest, as very few classes of codes are known to be decodable beyond the Johnson radius or even up to the Singleton bound.
In particular, it is known that Reed--Solomon (RS) codes, like all linear codes, can be list-decoded up to the Johnson radius \cite{Johnson1962} and an explicit algorithm exists \cite{Guruswami1999}. Although it has been shown that some Reed-Solomon codes can be list-decoded beyond this radius~\cite{Rudra2013}, there are no known algorithms to achieve this in general.
It is therefore an interesting observation distance-optimal LRCs \emph{can be list-decoded beyond the Johnson radius} while the complexity and list size grow polynomially in the code length, when the number of local repair sets is constant.

\textbf{Our Contribution:}
We consider different approaches that apply to different settings and rate regimes. First, we study \emph{list decoding} of LRCs. A list decoder returns all codewords within a specified distance around the received word and it is known that every $q$-ary code can be list-decoded up to the $q$-ary Johnson radius~\cite{Johnson1962,bassalygo1965} with a list size polynomial in the code length. For specific code classes, such as GRS codes, there exists a full body of work on list decoding, including explicit decoding algorithms \cite{sudan1997,Guruswami1999,kotter1996,roth2000,Guruswami2012} and analysis of the (average) list size for random errors \cite{Cheung1988,cheung1989,McEliece1986,McEliece2003}. Further, for some classes of codes based on GRS codes it has been shown that they can be decoded beyond the Johnson radius~\cite{Guruswami2012,guruswami2008,parvaresh2005}.

In the second part of this work, we show that interleaved codes, i.e., the direct sum of codes with errors occurring in the same positions, can increase the tolerance against errors even further.
In fact, in distributed data storage the assumption of burst errors, i.e., errors that corrupt the same positions in many codewords, which is required for a possible increase of the decoding radius through interleaved decoding, is very natural (cf.~Figure~\ref{fig:illustration}). Typically, a distributed storage system stores many codewords of the storage code, where each node stores one symbol of each codeword. Hence, if, e.g., one of the nodes is not synchronized correctly or an inner decoder fails, all codewords will be corrupted in the same position and in this case interleaved decoding can correct more errors compared to bounded minimum distance (BMD) decoding.

\begin{figure}[htb]
  \begin{center}
   \resizebox{0.6\linewidth}{!}{\def\x{0.5}

\begin{tikzpicture}

\node (S1) at (0,0) [draw,thick,minimum width=\x*1cm,minimum height=\x*5cm,] {};
\node (S2)  [right=\x*0.3cm of S1, draw,thick,minimum width=\x*1cm,minimum height=\x*5cm] {};
\node (S3)  [right=\x*0.3cm of S2, draw,thick,minimum width=\x*1cm,minimum height=\x*5cm] {};
\node (S4)  [right=\x*0.3cm of S3, draw,thick,minimum width=\x*1cm,minimum height=\x*5cm] {};


\node (S5)  [right=\x*0.7cm of S4, draw,thick,minimum width=\x*1cm,minimum height=\x*5cm] {};
\node (S6)  [right=\x*0.3cm of S5, draw,thick,minimum width=\x*1cm,minimum height=\x*5cm] {};
\node (S7)  [right=\x*0.3cm of S6, draw,thick,minimum width=\x*1cm,minimum height=\x*5cm] {};
\node (S8)  [right=\x*0.3cm of S7, draw,thick,minimum width=\x*1cm,minimum height=\x*5cm] {};


\node (S9)  [right=\x*0.7cm of S8, draw,thick,minimum width=\x*1cm,minimum height=\x*5cm] {};
\node (S10)  [right=\x*0.3cm of S9, draw,thick,minimum width=\x*1cm,minimum height=\x*5cm] {};
\node (S11)  [right=\x*0.3cm of S10, draw,thick,minimum width=\x*1cm,minimum height=\x*5cm] {};
\node (S12)  [right=\x*0.3cm of S11, draw,thick,minimum width=\x*1cm,minimum height=\x*5cm] {};


\draw[dotted,thick] (\x*4.9, \x*3) -- (\x*4.9,-\x*2.7);
\draw[dotted,thick] (\x*10.8, \x*3) -- (\x*10.8,-\x*2.7);


\foreach \i in {1,...,12}{
  \node (C\i) at ($(S\i)+(0,\x*2)$) [draw,minimum width=\x*0.7cm,minimum height=\x*0.5cm,rounded corners=1pt] {};
  \node (C\i) at ($(S\i)+(0,\x*1.4)$) [draw,minimum width=\x*0.7cm,minimum height=\x*0.5cm,rounded corners=1pt] {};
  \node (C\i) at ($(S\i)+(0,\x*0.8)$) [draw,minimum width=\x*0.7cm,minimum height=\x*0.5cm,rounded corners=1pt] {};
  \node (C\i) at ($(S\i)+(0,\x*0.2)$) [draw,minimum width=\x*0.7cm,minimum height=\x*0.5cm,rounded corners=1pt] {};
  \node (C\i) at ($(S\i)+(0,\x*-0.4)$) [minimum width=\x*0.7cm,minimum height=\x*0.5cm,rounded corners=1pt] {$\vdots$};
  \node (C\i) at ($(S\i)+(0,\x*-1.4)$) [draw,minimum width=\x*0.7cm,minimum height=\x*0.5cm,rounded corners=1pt] {};
  \node (C\i) at ($(S\i)+(0,\x*-2)$) [draw,minimum width=\x*0.7cm,minimum height=\x*0.5cm,rounded corners=1pt] {};

}


\draw[dashed, rounded corners = 1pt, blue, thick] ($(S1)+(-\x*.7,\x*2.4)$) rectangle ($(S12)+(\x*.7,\x*1.6)$);


\node[anchor = south east] (L1) at ($(S12)+(\x*1,\x*4)$) {\footnotesize LRC codeword};
\path (L1) edge[bend left, -{Latex[length=2mm,width=1mm]}]  ($(S12)+(\x*.7,\x*2.4)$) ;

\draw [decorate,decoration={brace,amplitude=5pt}]  ($(S1)+(-\x*1,-\x*2.4)$) -- ($(S1)+(-\x*1,\x*2.4)$) node [black,midway,xshift=-0.4cm,rotate=90] {\footnotesize $\ell$ independent codewords};


\draw[thick,-{Latex[length=2mm,width=1mm]},red] ($(S2)+(0,\x*2)$) -- ($(S2)+(-\x*0.3,-\x*0.4)$) -- ($(S2)+ (\x*0.3,\x*0.4)$) -- ($(S2)+(0,-\x*2)$);
\draw[thick,-{Latex[length=2mm,width=1mm]},red] ($(S4)+(0,\x*2)$) -- ($(S4)+(-\x*0.3,-\x*0.4)$) -- ($(S4)+ (\x*0.3,\x*0.4)$) -- ($(S4)+(0,-\x*2)$);
\draw[thick,-{Latex[length=2mm,width=1mm]},red] ($(S7)+(0,\x*2)$) -- ($(S7)+(-\x*0.3,-\x*0.4)$) -- ($(S7)+ (\x*0.3,\x*0.4)$) -- ($(S7)+(0,-\x*2)$);
\draw[thick,-{Latex[length=2mm,width=1mm]},red] ($(S12)+(0,\x*2)$) -- ($(S12)+(-\x*0.3,-\x*0.4)$) -- ($(S12)+ (\x*0.3,\x*0.4)$) -- ($(S12)+(0,-\x*2)$);


\node at ($(S6)+(\x*0.6,\x*4.2)$) {Servers};
\node at ($(S2)+(\x*0.6,\x*3.2)$) {Local group $1$};
\node at ($(S6)+(\x*0.6,\x*3.2)$) {Local group $2$};
\node at ($(S10)+(\x*0.6,\x*3.2)$) {Local group $3$};


\node (El) at ($(S6)+(\x*0.6,-\x*4.2)$) {\color{red} Burst Errors};
\path ($(El.north west)+(\x*0,\x*-0.2)$) edge[bend left=5, -{Latex[length=2mm,width=1mm]}]  ($(S2)+(\x*0.1,-\x*2.6)$) ;
\path ($(El.north)+ (-\x*0.5,0)$) edge[bend left=5, -{Latex[length=2mm,width=1mm]}]  ($(S4)+(\x*0.1,-\x*2.6)$) ;
\path ($(El.north) + (\x*0.5,0)$) edge[bend right=5, -{Latex[length=2mm,width=1mm]}]  ($(S7)+(-\x*0,-\x*2.6)$) ;
\path ($(El.north east)+(\x*0,\x*-0.2)$) edge[bend right=5, -{Latex[length=2mm,width=1mm]}]  ($(S12)+(-\x*0.1,-\x*2.6)$) ;

\end{tikzpicture}

}
\end{center}
\caption{Illustration of LRC coded storage system with burst errors}
\label{fig:illustration}
\end{figure}
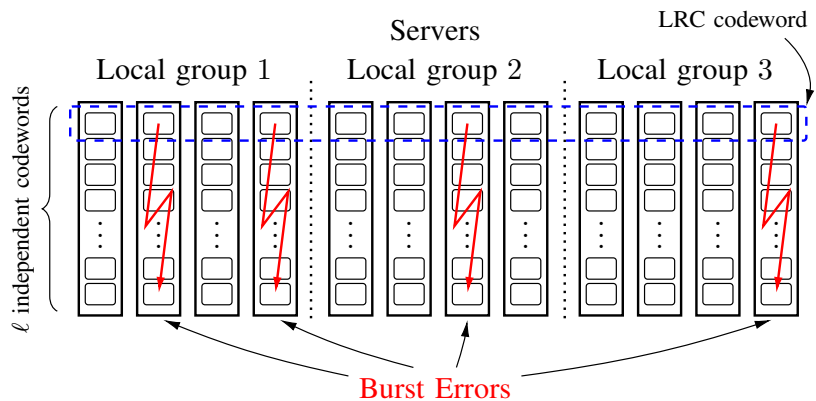

\textbf{Structure of the Paper}:
Section~\ref{sec:prelim} gives notations and definitions.
In Section~\ref{sec:listDecoding}, we show that a large class of optimal LRCs can be list-decoded beyond the alphabet-dependent Johnson radius. Among others, we show that if the number or locally list-decodable errors, normalized by the repair set size, is larger than the normalized global Johnson radius, we can increase the overall decoding radius (cf.~Theorem~\ref{thm:ListDecodingLRCs} and Lemma~\ref{lem:localUniqueDecoding}).
We give bounds on the resulting list size and decoding complexity, which are exponential in the number of repair sets, but polynomial in the length of the code when considered under an asymptotic regime where the number of repair sets is constant.
The new decoder can also be interpreted as a probabilistic unique decoder, which is of low complexity, but fails with a probability that we can bound using probabilities that certain other list decoders return a unique result.

In Section~\ref{sec:ILRC}, we first specialize the decoder of Section~\ref{sec:listDecoding} to LRCs that are subcodes of GRS codes, and whose local codes are GRS codes. Using the Guruswami--Sudan list decoder, we obtain an explicit algorithm that attains the new decoding radius described in Section~\ref{sec:listDecoding}. 
Using McEliece's results on the probability that the list size of the Guruswami--Sudan decoder is one~\cite{McEliece2003}, we obtain explicit bounds on the success probability of the unique decoder and show that it converges to $1$ for the code length going to infinity.
Furthermore, in Section~\ref{subsec:decSupercode}, we show that the decoding radius can be further increased by interleaved decoding using the existing decoders for interleaved GRS codes.

In Section~\ref{sec:PMDS}, we consider decoding of interleaved PMDS codes. We show that we can decode some interleaved PMDS codes beyond their minimum distance with high probability by the decoding algorithm for high-order interleaved codes by Metzner and Kapturowski \cite{metzner1990general}. The advantage of this decoder is that it only relies on the inherent properties of PMDS codes, and is therefore applicable to \emph{any} class of PMDS codes. We derive bounds on the success probability and show that for some families of PMDS code parameters we can correct up to $n-k-1$ errors with success probability approaching $1$ as the length goes to infinity.

The introduced methods of decoding give an improvement for many different rate regimes, complementing each other in the required relation between rate of the local codes and the global code (cf.~Figure~\ref{fig:rate_tuples} in Section~\ref{sec:conclusion}).

\section{Preliminaries}\label{sec:prelim}

We denote by $\F{q}$ a finite (extension) field with $q=p^m$ elements, where~$p$ is a prime and $m$ is a positive integer. We write $[a]$ for the set of integers $\{i : 1 \leq i \leq a , i \in \mathbb{Z} \}$. A $q$-ary code of length $n$, dimension $k$, and distance $d$ is denoted by $[n,k,d]_q$, and if the parameters $d$ and/or $q$ are not of interest, we omit them.

Let $\code$ be an $[n,k]$ code and $\mathcal{R} \subseteq [n]$ be a set of coordinates. Denote by $\code|_{\mathcal{R}}$ the code obtained by restricting the code~$\code$ to the coordinates of $\mathcal{R}$, i.e., puncturing (deleting) the positions $[n]\setminus \mathcal{R}$. We define \emph{shortening} the code $\code$ in position $i$ by a fixed value $\gamma$ as $\code' = \{c | c\in \code ,  c_i = \gamma \}_{[n] \backslash i}$.

\subsection{Locally Repairable Codes}

A code is said to have locality $r$ if every position can be recovered from at most $r$ other codeword positions. If multiple erasures can be tolerated within such a \emph{local repair set}, the code is said to have $(r,\varrho)$ locality.

\begin{definition}[$(r,\varrho)$-locality]
  An $[n,k]$ code $\code$ has $(r,\varrho)$-locality if there exists a partition $\mathcal{P} = \{\mathcal{R}_1,\mathcal{R}_2,...,\mathcal{R}_\mu\}$ of $[n]$ into sets of cardinality $|\mathcal{R}_j| \leq r+\varrho-1$ such that for the distance of the code restricted to the positions of $\mathcal{R}_j$ it holds that $d\left(\code|_{\mathcal{R}_j}\right) \geq \varrho \;, \forall \, j\in [\mu]$. 
\end{definition}
A Singleton-like upper bound on the achievable distance of an $[n,k,r,\varrho]$ LRC was derived in \cite{Gopalan2012} for $\varrho=2$ and generalized to $\varrho \geq 2$ in \cite{Kamath2014} to
\begin{equation} \label{eq:boundDistanceLRC}
  d \leq n-k+1-\left( \left\lceil \frac{k}{r} \right\rceil -1 \right) (\varrho-1) \ .
\end{equation}
In the following we refer to codes achieving this bound with equality as \emph{optimal} LRCs.

Several classes of optimal LRCs are known \cite{tamo2016optimal, Silberstein2013}, including one of particular interest for this work, the so-called Tamo--Barg LRCs \cite{Tamo2014}. Tamo--Barg LRCs are constructed as subcodes of RS codes, therefore the requirement on the field size is only in the order of the code length $n$ and they can be decoded by any of the well-studied RS decoders.

In the following, $\mathcal{R}_j$ is referred to as the $j$-th local repair set and the code $\code|_{\mathcal{R}_j}$ as the $j$-th local code. For simplicity, we only consider LRCs where every local code is of the same length $n_l = |\mathcal{R}_1| = ...= |\mathcal{R}|_{\mu}$ with $n_l \mid n$ and $r \mid k$.  We denote a $q$-ary code of length~$n$, dimension~$k$, locality~$r$ and local distance~$\varrho$ by $[n,k,r,\varrho]_q$, and if the field size $q$ is not important, we omit it.

\subsection{Partial MDS Codes}

In Section~\ref{sec:PMDS}, we consider interleaved decoding of a special class of codes with locality, namely \emph{partial MDS} (PMDS) codes \cite{Blaum2013,gabrys2018constructions,Blaum2016,calis2016general,Horlemann-Trautmann2017}, also referred to as \emph{maximally recoverable codes} \cite{Huang2007, Chen2007, Gopalan2014,gopalan2017,martinez2019universal}. The distinctive property of these codes is that they guarantee to correct any pattern of erasures that is information-theoretically correctable, i.e., every set of codeword positions which is not necessarily linearly dependent by the locality constraint or the code dimension, is linearly independent. While this increases the number of correctable erasure patterns, it generally comes at the cost of a larger required field size compared to other LRCs. We state the definition of PMDS codes given in \cite{Blaum2013} in terms of our notation.
\begin{definition}[PMDS codes]\label{def:partialMDS}
   Let $\Code$ be an $[n,k,r,\varrho]$ LRC code, where the local codes are $[r+\varrho-1,r,\varrho]$ MDS codes. We say that the code $\code$ is a \emph{partial MDS code} if for any set $E \subset [n]$, where $E$ is obtained by picking $\varrho-1$ positions from each local repair set, the distance of the code punctured in these positions is $d(\code|_{[n]\setminus E}) = n-\frac{n(\varrho-1)}{r+\varrho-1}-k+1$.
 \end{definition}
 Note that the definition of PMDS codes implies the optimality of the code with respect to the Singleton-like bound \cite{Gopalan2012,Kamath2014} given in (\ref{eq:boundDistanceLRC}).
 \begin{remark}
   A more general form of PMDS codes, where the distance, i.e., number of tolerable erasures, can be different in each local repair set is often considered in literature. For simplicity, we focus on PMDS with the same distance in each local code in this work but note that the decoding approach is also valid for PMDS codes with different distances in the local codes.
 \end{remark}

\subsection{Interleaved Codes}
Interleaved codes are direct sums of a number of constituent codes, where, by assuming that errors occur at the same positions in all constituent codewords (burst errors), one is often able to decode far beyond half the minimum distance and even the Johnson radius.
We will only consider homogeneous interleaved codes over linear codes in this work, for which the constituent codes are all the same and linear.

\begin{definition}[See, e.g., \cite{metzner1990general,krachkovsky1997decoding}]
Let $\Code[n,k,d]$ be a linear code over $\Fq$ and $\ell \in \NN$ be called the interleaving degree. The corresponding $\ell$-interleaved code is defined by
\begin{align*}
\IC[\ell; n,k,d] := \left\{ \C = \left[\begin{smallmatrix}
      \c_1 \\ \c_2 \\ \vphantom{\int\limits^x}\smash{\vdots} \\ \c_\ell
\end{smallmatrix}\right] \, : \, \c_i \in \Code \right\}.
\end{align*}
\end{definition}

The assumed error model is as follows.
We want to reconstruct a codeword $\C$ from a received word of the form $\R = \C + \E$, where $\E \in \Fq^{\ell \times n}$ is an error matrix. Let $\Eset$ be the set of indices of non-zero columns of $\E$, then we say that an error matrix is of weight $t$ if $|\Eset| = t$.
This error model is often called a ``burst error'' and motivated by many applications, such as replicated file disagreement location~\cite{metzner1990general}, data-storage applications~\cite{krachkovsky1997decoding} (cf.~Figure~\ref{fig:illustration}), suitable outer codes in concatenated codes~\cite{metzner1990general,krachkovsky1998decoding,haslach1999decoding,justesen2004decoding,schmidt2005interleaved,schmidt2009collaborative}, an ALOHA-like random-access scheme \cite{haslach1999decoding}, decoding non-interleaved codes beyond half-the-minimum distance by power decoding \cite{schmidt2010syndrome,kampf2014bounds,rosenkilde2018power,puchinger2019improved}, and recently for code-based cryptography \cite{elleuch2018interleaved,Holzbaur2019crypto}.

It is well-known that by interpreting the columns of the interleaved codewords as elements of an extension field $\mathbb{F}_{q^\ell}$, the resulting code is a linear code over $\mathbb{F}_{q^\ell}$ with the same parameters $[n,k,d]$. If the constituent codes are RS codes, then the resulting code is also an RS codes with the same evaluation points (which are in a subfield $\Fq$ of the code's field $\mathbb{F}_{q^\ell}$), cf.~\cite{sidorenko2008decoding}. Note that the mapping also provides a one-to-one correspondence of burst errors in $\Fq^{\ell \times n}$ and Hamming errors of the same weight in $\mathbb{F}_{q^\ell}^n$.

For some constituent codes, for instance RS or some AG codes, there are efficient decoders that correct many errors beyond half the minimum distance and even the Johnson radius with high probability.
Beyond these radii, the known algorithms fail for some error patterns, but succeed for a fraction of errors close to $1$.
The first such algorithm was given in \cite{krachkovsky1997decoding} for interleaved RS codes and corrects up to $\tfrac{\ell}{\ell+1}(n-k)$ errors.
Since then, many decoders with better complexity and larger decoding radius, as well as some bounds on the probability of decoding failure have been derived \cite{bleichenbacher2003decoding,coppersmith2003reconstructing,parvaresh2004multivariate,brown2004probabilistic,parvaresh2007algebraic,schmidt2007enhancing,schmidt2009collaborative,cohn2013approximate,nielsen2013generalised,wachterzeh2014decoding,puchinger2017irs,yu2018simultaneous}.
One decoder of special interest for this work was introduced by Metzner and Kapturowski in \cite{metzner1990general} and will be discussed in more detail in Section~\ref{sec:metznerKapturowski}.
It is a generic decoder, which works for interleaved codes of high interleaving degree and arbitrary constituent codes.

\section{List Decoding of Errors in Locally Repairable Codes} \label{sec:listDecoding}

List decoding is a powerful technique, where instead of returning a unique codeword or a decoding failure, the decoder returns a list of \emph{all codewords} within a given distance from the received word. It has been shown combinatorially that any $q$-ary code can be list-decoded up to the $q$-ary Johnson radius \cite{Johnson1962,bassalygo1965}, which only depends on the code alphabet, length, and minimum distance, with a maximum list size that is polynomial in the code length.
Note however that for many code classes, it is a challenging task to find explicit decoders up to the Johnson radius.
In this section, we introduce a decoding method for LRCs that achieves a decoding radius that exceeds the Johnson radius of the a large class of LRCs, by making use of the additional redundancy required for the locality in the decoding process. There are only few other known non-trivial classes of codes \cite{Guruswami2012,guruswami2008,parvaresh2005} for which it is known that they can be decoded beyond the Johnson radius. In addition to this theoretical interest, the proposed decoder is also of interest for practice, as it applies to a very large class of LRCs without requiring a change in the structure of the code.

In this work we consider list decoding of \emph{errors}, but for sake of completeness we include a discussion of erasure list decoding in Appendix~\ref{app:erasure_list_decoding}.

\subsection{New Decoding Radius}\label{subsec:newdecodingradius}

A $q$-ary code of length~$n$ is called~$(\tau,L)$-list-decodable if the Hamming sphere of radius~$\tau$ centered at any vector~$\ve{v}$ of length~$n$ always contains at most~$L$ codewords~$\c\in  \code$. It is known \cite{Johnson1962,bassalygo1965} that any $q$-ary code of length~$n$ and distance~$d$ is list-decodable up to the $q$-ary Johnson radius, which is given by the largest $\tau_J$ such that
\begin{align}
  0&\leq (\theta_qn-\tau_J)^2 - \theta_qn(\theta_qn-d) \label{eq:johnsonCondition}\\
    &\Longrightarrow \;\; \tau_J = \theta_q n\left(1-\sqrt{1-\frac{d}{n\theta_q}} \right)  \ , \label{eq:johnsonradius}
\end{align}
where $\theta_q = 1-\frac{1}{q}$ and the maximum list size is upper bounded by
\begin{equation}
  L \leq \frac{\theta_q dn}{\tau^2-\theta_qn(2\tau-d)} \ . \label{eq:maximal_list_size}
\end{equation}
In the following, if $q$ is not explicitly given, we consider the alphabet-independent case of $q \rightarrow \infty$, i.e., $\theta_q = 1$.
For any radius $\tau$, we denote the number of list-decodable errors, i.e., the largest integer smaller than~$\tau$, by
\begin{align}
t=\ceil{\tau-1},\; \text{where} \;\tau-1 \leq t < \tau \label{eq:definet} \ ,
\end{align}
with the corresponding subscript, e.g., $t_J$ for the Johnson radius $\tau_J$. Further, for an LRC we denote by~$\tau_g$ and~$\tau_l$ the global and local decoding radius and define~$t_g$ and~$t_l$ accordingly, i.e., as the total number of correctable errors and the number of errors correctable in each local code.

Generally, it has to be assumed that the list size increases exponentially in the code length~$n$ when the radius exceeds~\eqref{eq:johnsonradius}. While it is known that there are codes for which the bound is not tight and the list-decoding radius exceeds the Johnson radius \cite{parvaresh2005,Guruswami2006,guruswami2008,Guruswami2012, Rudra2013}, the behavior of most codes is still mostly an open problem. In the following, we show that the list-decoding radius of certain LRCs exceed the Johnson radius, i.e., the complexity and list size grow \emph{polynomially} in the length when the number of local repair sets~$\frac{n}{n_l}$ is constant.

Before giving the main statement of this section, we first establish two lemmas on the distribution of the of errors in an LRC and the relation of the Johnson radius to the code length.

\begin{lemma}  \label{lem:sigma}
	Let~$\code$ be an~$[n,k,r,\varrho]$ LRC.  For a codeword~$\c \in \code$ and any word~$\ve{w}$ with~$\dt{\c}{\ve{w}} \leq t_g$, let~$\mathcal{I}\subseteq \left[\frac{n}{n_l}\right]$ be the set of repair set indices~$i$ with~$\dt{\c|_{\mathcal{R}_i}}{\ve{w}|_{\mathcal{R}_i}} \leq t_l, \; \forall \; i \in \mathcal{I}$. Then the cardinality of~$\mathcal{I}$ is bounded by $\mathcal{I} \geq \ceil{\sigma}$ with
	\begin{equation}
	\sigma = \max\left\{0,\frac{n}{n_l}-\frac{\tau_g}{\tau_l} \right\} .\label{eq:sigmaineq}
	\end{equation}
\end{lemma}
\begin{IEEEproof}
	Trivially, the cardinality of~$\mathcal{I}$ is non-negative. The maximum number of repair sets~$\mathcal{R}_j$ with~$\dt{\c_{\mathcal{R}_j}}{\ve{w}_{\mathcal{R}_j}} > t_l$ such that~$\dt{\c}{\ve{w}} = \sum_{j=1}^{\frac{n}{n_l}} \dt{\c_{\mathcal{R}_j}}{\ve{w}_{\mathcal{R}_j}} \leq t_g$ is given by~$\floor{\frac{t_g}{t_l+1}}$. Subtracting from the total number of repair sets~$\frac{n}{n_l}$ gives
	\begin{align*}
	\frac{n}{n_l}-\floor{\frac{t_g}{t_l+1}} \stackrel{(a)}{=} \ceil{\frac{n}{n_l}-\floor{\frac{t_g}{t_l+1}}} \geq \left\lceil \frac{n}{n_l}-\frac{t_g}{t_l+1}\right\rceil \geq \ceil{\frac{n}{n_l}-\frac{\tau_g}{\tau_l}} \ ,
	\end{align*}
        where $(a)$ holds because the left hand side is an integer.
\end{IEEEproof}

\begin{lemma}\label{lem:increasingInN}
  For any fixed $d,\ell\geq 1$ and $q>1$ let
  \begin{align*}
    h(n) \coloneqq \theta_q n\left(1-\left(1-\frac{d}{n\theta_q}\right)^{\frac{\ell}{\ell+1}} \right) \ ,
  \end{align*}
  with $\theta_q = q-\frac{1}{q}$. The function $h(n)$ is monotonically decreasing in $n$ for $n \geq \frac{d}{\theta_q}$.
\end{lemma}
\begin{IEEEproof}
  The proof is given in Appendix~\ref{app:proof_of_lemma_increasingInN}.
\end{IEEEproof}
Note that $n \geq \frac{d}{\theta_q}$ needs to hold for any $q$-ary code of length $n$ and distance $d$ \cite{bassalygo1965}.
The following theorem provides our most general statement, which is valid for any LRC.

\begin{theorem}[List Decoding of LRCs] \label{thm:ListDecodingLRCs}
	Let $L_{(n,d,\tau)}$ denote the maximum list size when list decoding an~$[n,k,d]_q$ code with radius~$\tau$. Then an $[n,k,r,\varrho]_q$ LRC is~$(\tau_g,L_g)$-list-decodable, with 
	\begin{equation}
	\tau_g = \left\{
	\begin{array}{ll}
	\frac{d}{\varrho} \cdot \tau_{J,l} \, , &\text{if}\;\; \ceil{\sigma}>0 \\
	\theta_q\left(n-\sqrt{n\left(n-\frac{d}{\theta_q}\right)}\right)   \, , & \text{else} \ ,
	\end{array} \right.   \label{eq:jblrc}
	\end{equation}
    where $\theta_q=1-\frac{1}{q}$ and $\tau_{J,l}$ denotes the $q$-ary Johnson radius of the local codes.
    The list size is upper bounded by
    \begin{equation}\label{eq:listSize}
      L_g \leq \binom{\frac{n}{n_l}}{\ceil{\sigma}} L_{(n_l,\varrho,\tau_{J,l})}^{\ceil{\sigma}} L_{(n-\ceil{\sigma} n_l,d,\tau_g)} \ .
    \end{equation}
\end{theorem}
\begin{IEEEproof}
  Denote by $\mathcal{I}\subseteq \left[\frac{n}{n_l}\right]$ the set of repair set indices $i$ with~$\dt{\c_{\mathcal{R}_i}}{\ve{w}_{\mathcal{R}_i}} \leq t_l$. By Lemma~\ref{lem:sigma} the cardinality of this set is~$|\mathcal{I}| \geq \ceil{\sigma}$. These repair sets can be list-decoded locally. For every combination of codewords in these local lists, shortening the code in the corresponding~$|\mathcal{I}|n_l$ positions gives an~$(n',k',d')$ code\footnote{The dimension $k'$ of the shortened code depends on the particular code. However, the Johnson radius, and thereby the following arguments, only depend on the length $n'$ and distance $d'$ of the shortened code.} of length $n'=n-|\mathcal{I}| n_l$ and distance $d' \geq d$. Let $\tau'$ be the $q$-ary Johnson radius of the shortened code. By setting $\ell=1$ in Lemma~\ref{lem:increasingInN} we obtain
  \begin{align*}
    \tau' = \theta_q\left(n-|\mathcal{I}| n_l - \sqrt{(n-|\mathcal{I}| n_l)\left(n-|\mathcal{I}| n_l-\frac{d'}{\theta_q}\right)}\right) \geq \theta_q\left(n-\sigma n_l - \sqrt{(n-\sigma n_l)\left(n-\sigma n_l-\frac{d}{\theta_q}\right)}\right) \coloneqq \tau_g
  \end{align*}
  for any $n-|\mathcal{I}|n_l \geq d$. Note that if $n-|\mathcal{I}|n_l \leq d$ the shortened code is necessarily of dimension $k'=0$ and decoding is trivial.
  With $n-\sigma n_l = \frac{\tau_g}{\tau_l} n_l$ we obtain
  \begin{align}
    \tau_g &= \theta_q\left( \frac{\tau_g}{\tau_l} n_l- \sqrt{\frac{\tau_g}{\tau_l} n_l\left(\frac{\tau_g}{\tau_l} n_l-\frac{d}{\theta_q}\right)}\right) \nonumber \\
    \tau_g \left( \frac{n_l}{\tau_l} - \theta_q^{-2}\right)^2 &= \frac{n_l}{\tau_l} \left( \frac{n_j}{\tau_l} - \frac{d}{\theta_q}\right) \nonumber \\
    \tau_g &=  \frac{d \tau_l}{2 \tau_l-\frac{\tau_l^2}{\theta_q n_l}} \ . \label{eq:tau2}
  \end{align}
  Each local code is a $q$-ary code of length $n_l$ and distance $\varrho$, and can therefore be decoded up to the local Johnson radius of
  \begin{align*}
    \tau_{J,l} = \theta_q \left( n_l-\sqrt{n_l\left(n_l-\frac{\varrho}{\theta_q}\right)}\right) \ .
  \end{align*}
  Setting $\tau_l = \tau_{J,l}$ gives
  \begin{align*}
    \tau_g &=  \frac{d \tau_{J,l}}{2 \theta_q \left( n_l-\sqrt{n_l\left(n_l-\frac{\varrho}{\theta_q}\right)}\right)-\frac{\theta_q^2 \left( n_l-\sqrt{n_l\left(n_l-\frac{\varrho}{\theta_q}\right)}\right)^2}{\theta_q n_l}} = \frac{d}{\varrho} \tau_{J,l} \ .
  \end{align*}
There are at most~$\binom{\frac{n}{n_l}}{\ceil{\sigma}}$ choices for the~$\ceil{\sigma}$ list-decodable repair sets and for each choice there are at most~$L_{(n_l,\varrho,\tau_{J,l})}^{\ceil{\sigma}}$ distinct possibilities to shorten the received word. The list size of each shortened code is upper bounded by~$L_{(n-\ceil{\sigma} n_l,d,\tau_g)}$ and the upper bound on the global maximum list size $L_g$ follows.
\end{IEEEproof}

  \begin{remark}
    Observe that the radius and list size defined in Theorem~\ref{thm:ListDecodingLRCs} are equal to the Johnson radius and the corresponding list size when $\sigma=0$. For $\sigma>0$, the radius is larger than the Johnson radius while the list size is still polynomial in $n$ for fixed $\frac{n}{n_l}$. Therefore, the radius of Theorem~\ref{thm:ListDecodingLRCs} generalizes the Johnson radius in the sense that it contains the pair of radius (\ref{eq:johnsonradius}) and list size (\ref{eq:maximal_list_size}) as a special case. In particular, notice that the upper bound on the list size of (\ref{eq:maximal_list_size}) is not simply exponential in $n$ for $\tau>\tau_J$ (which would suffice to say that the bound does not imply list decodability anymore) but does not hold for $\tau>\tau_J$ and therefore \textbf{does not allow to make any statement on the list size} for such $\tau$. For many classes of codes, e.g., RS codes, it is an open problem whether they are list decodable beyond the Johnson radius with polynomial list size. Hence, even for finite parameters with $\sigma>0$, i.e., in the non-asymptotic regime, Theorem~\ref{thm:ListDecodingLRCs} is an improvement compared to the Johnson radius, regardless of the specific upper bound on the list $L_g$, as the Johnson radius and the corresponding bound on the list size are unsuitable to give any statement for $t$ with $\tau_J < t <\tau_g$.
  \end{remark}

\begin{remark}\label{rem:list_size_bound_not_tight}
For $\sigma>0$, the list size bound in Theorem~\ref{thm:ListDecodingLRCs} is not tight, both for the worst and average case.

\textbf{1) Worst-case list size:}
It becomes apparent from the proof of Theorem~\ref{thm:ListDecodingLRCs} that the upper bound \eqref{eq:listSize} on the list size can only be attained if
\begin{itemize}
\item the list size in each local repair set is maximal, i.e., equals $=L_{(n_l,\varrho,\tau_{J,l})}$ and
\item for each of the $\binom{\frac{n}{n_l}}{\ceil{\sigma}}$ combinations of $\ceil{\sigma}$ repair sets and all $L_{(n_l,\varrho,\tau_{J,l})}^{\ceil{\sigma}}$ combinations the codewords in the output lists of the local decoders in these repair sets, after shortening these repair sets, the decoder returns a list of $L_{(n-\ceil{\sigma} n_l,d,\tau_g)}$ codewords.
\end{itemize}
For a given set of $\ceil{\sigma}$ repair sets and combination of local codewords in the local decoder's output, the number of errors in the remaining repair sets is at most $\tau_g-\chi$, where $\chi$ is the sum of the distances of the chosen local codewords to the corresponding repair sets of the received word. Note that $\chi$ is known to the decoder after fixing the repair sets and local codewords.
Hence, a decoder can globally decode the shortened received word with a decoder of radius $\tau_g-\chi$ instead of $\tau_g$, decreasing the maximal list size for this combination of repair sets and local codewords from $L_{(n-\ceil{\sigma} n_l,d,\tau_g)}$ to $L_{(n-\ceil{\sigma} n_l,d,\tau_g-\chi)}$.
It is easy to see that for a fixed combination of local repair sets, there is at most one combination of local codewords with $\chi=0$ (in fact, if $\chi=0$ for one combination, the list size must be $1$ for all these local repair sets).
Hence, the sum of the resulting list sizes of all repair sets/local codewords combinations is strictly smaller than the bound \eqref{eq:listSize}.

By carefully analyzing possible distances of local codewords to the received word for maximal local list sizes, one may obtain a better worst-case bound.
We present a first-order improvement of the worst-case bound in Theorem~\ref{thm:ListDecodingLRCsImproved} in Appendix~\ref{app:improved_list_bound}, which makes use of the fact that a local list either contains only one element or, if the list is larger, all codewords in the list have a minimal distance to the received word.
The improvement is, however, not significant: we can save at most a factor $L_{(n-\ceil{\sigma} n_l,d,\tau_g)}$ compared to \eqref{eq:listSize} and the terms exponential in the number of local repair sets remain (cf.~Remark~\ref{rem:improved_list_size_bound_no_significant_improvement} in Appendix~\ref{app:improved_list_bound}).
In the main part of the paper, we present the bound~\eqref{eq:listSize} to focus on the principle of our list decoding approach. The improvement discussed before is technical and the improvement is marginal, but it might yield to more significant improvements by further analyzing the relation of local and global list sizes and is therefore shown in Appendix~\ref{app:improved_list_bound}.
As a second possible approach for improving the worst-case list size we will discuss a connection of the list size to low-weight codewords of a code in Remark~\ref{rem:list_decoding_connection_to_low_weight_codewords} below. By bounding the number of such codewords, one may also be able to obtain a better upper bound on the worst-case list size~\eqref{eq:listSize}.
In both approaches, the weight distribution of the local and global codes is a key towards an improved result, which is however not known (and in general not the same) for all codes---making a general statement technical and hard to achieve.

If the decoder had, for each codeword with distance at most $\tau_g$ to the received word, additional side information which $\lceil \sigma \rceil$ local decoders were successful (i.e., contain the correct local codeword in the list), it would be possible to save up to a factor $\binom{\frac{n}{n_l}}{\ceil{\sigma}}$ compared to the worst-case list size bound~\eqref{eq:listSize}.
For the worst-case list size, this is only relevant if this side information is not probabilistic (e.g., if the information is given from an external source with additional information about the received word).
We are not aware of a method to extract this information intrinsically in a non-probabilistic fashion without first obtaining the entire output list.

\textbf{2) Average-case list size:}
We will see in Sections~\ref{subsec:probdec} and \ref{subsec:ProbTB} that for certain families of LRCs, we can show that for random errors of weight at most $\tau_g$, the probability of obtaining a global list size greater than $1$ is small. Hence, the expected list size is usually significantly smaller than the upper bound \eqref{eq:listSize}.

Slightly more general, we can also extract the above mentioned additional side information (on which local decoders were successful) as soft information from the local lists.
In this case, we end up with a scenario where the side information is correct with a certain probability and the resulting output list size is significantly smaller.
The analysis goes, however, beyond the scope of this paper. The probabilistic behavior of the list size for random received words is not known for all code classes, making a general statement technical and hard to achieve. For cases in which it is known (e.g., MDS codes), it might be possible to generalize the results in Sections~\ref{subsec:probdec} and \ref{subsec:ProbTB}.
Note also that for many code parameters, the probability to obtain a list size $>1$ is quite small, so it is not clear if we could obtain a significantly better list size bound with the more general probabilistic analysis.
\end{remark}

\begin{remark}\label{rem:radiusWithFloor}
  Note that the radius derived in Theorem~\ref{thm:ListDecodingLRCs} is obtained by using the bound on the number of locally correctable repair sets derived in Lemma~\ref{lem:sigma}. While this allows for deriving a closed form expression for the maximum radius, the actual number of correctable errors is larger when considering the floor operation, i.e., shortening to obtain a code of length $n_J=n-\left(\frac{n}{n_l}-\left\lfloor \frac{t_g}{t_l+1} \right\rfloor \right)n_l$. Then, the number of correctable errors $\bar{t}_g$ is given by the largest integer $\bar{t}_g$ such that
  \begin{align}
    0< \bar{t}_g^2 + \theta_q \left\lfloor \frac{\bar{t}_g}{t_l+1} \right\rfloor n_l(d-2\bar{t}_g) \ . \label{eq:radiusWithFloor}
  \end{align}
\end{remark}

\begin{example}\label{eg:code63}
Consider an~$[63,16,8,14]$ optimal LRC achieving~\eqref{eq:boundDistanceLRC} with equality, i.e.,~$d=35$. It follows that BMD decoding corrects up to~$t_{\mathrm{BMD}} = 17$ errors uniquely and with~\eqref{eq:johnsonradius} we get a list-decoding radius of~$\tau_J < 21$, i.e.,~$t_J = 20$. By Theorem~\ref{thm:ListDecodingLRCs}, we obtain~$\tau_g < 22.18$, i.e.,~$t_g = 22$. Hence, two additional errors can be corrected. When considering the floor as discussed in Remark~\ref{rem:radiusWithFloor}, the number of correctable errors is $\bar{t}_g = 24$.
\end{example}
\begin{example}
  Consider an optimal $[15,6,3,3]$ LRC of distance $d=8$. Equation~\eqref{eq:jblrc} gives $\tau_g \approx 4.9$ and it follows that $4$ errors can be corrected, the same number as given by the Johnson radius for any code of length $n=15$ and distance $d=8$. However, when considering the floor operations as in Remark~\ref{rem:radiusWithFloor}, which are bounded in the proof of Lemma~\ref{lem:sigma}, the error correction capability is given by the largest integer $\bar{t}_g$ that fulfills (\ref{eq:radiusWithFloor}), which is $\bar{t}_g= 5$.
  For more examples, see Table~\ref{tab:exampleParameters}.
\end{example}

Whether the obtained radius is larger than the Johnson radius of the code can be determined by a simple criterion.
\begin{corollary} \label{cor:gain}
  The decoding radius~$\tau_g$ of~\eqref{eq:jblrc} is larger than the global Johnson radius~$\tau$ if and only if the sum of the local distances is larger than the code distance, i.e.
  \begin{equation} \label{eq:gaincond}
    \tau_g > \tau \;\; \Longleftrightarrow \;\; \frac{n}{n_l} \cdot \varrho > d .
  \end{equation}

\end{corollary}
\begin{IEEEproof}
  It follows from Lemma~\ref{lem:increasingInN} that~$\tau_g > \tau$ if~$\sigma >0$. Substituting~\eqref{eq:jblrc} for~$\sigma > 0$ into~\eqref{eq:sigmaineq} gives
  \begin{align*}
    \frac{n}{n_l} - \frac{\tau_g}{\tau_{J,l}} = \frac{n}{n_l} - \frac{d}{\varrho} &> 0 .
  \end{align*}
\end{IEEEproof}

\begin{remark}\label{rem:list_decoding_connection_to_low_weight_codewords}
Equation~\eqref{eq:gaincond} provides some intuition why the decoding radius can be increased. For linear codes, the list-decoding problem is closely related to low-weight codewords, as every incorrect codeword in the Hamming sphere of radius~$\tau$ centered around the received word~$\ve{w}$ can be written as~$\ve{c}+\ve{c}'$, where~$\ve{c}$ is the correct codeword and~$\ve{c}'$ is a codeword with~$\wt(\ve{c}') < 2 \tau$. Further, it has to hold that~$\wt(\ve{c'}) \geq d$ and~$\wt(\ve{c}'_{\mathcal{R}_i}) \geq \varrho \;, \; i \in \left[\frac{n}{n_l}\right]$. It follows that if~\eqref{eq:gaincond} holds, low-weight codewords must have entire repair sets that are all zero. When taking all codewords from~$\code$ that have the same all-zero repair sets, we obtain shorter codes of same distance and the number of such shortened codes depends on the number of repair sets. As every low-weight codeword has to be in one of these, the total number of low-weight codewords in~$\code$ depends only on the number of repair sets and the number of low-weight codewords in codes of shorter length and distance~$d$.
\end{remark}

To derive the radius in Theorem~\ref{thm:ListDecodingLRCs} we set the local decoding radius to be the Johnson radius of the local codes. While this choice is valid for any code, it is also possible to increase the radius for different $\tau_l$.
\begin{lemma}\label{lem:localUniqueDecoding}
  The decoding radius of Theorem~\ref{thm:ListDecodingLRCs} is larger than the $q$-ary Johnson radius if
  \begin{equation*}
    \frac{\tau_l}{n_l} > \theta_q\left(1-\sqrt{1-\frac{d}{n\theta_q }}\right) \ ,
  \end{equation*}
  i.e., if the normalized local decoding radius is larger than the normalized $q$-ary Johnson radius.
\end{lemma}
\begin{IEEEproof}
  By (\ref{eq:tau2}) there is a gain if
  \begin{align*}
    \frac{d\tau_l}{2\tau_l-\frac{\tau_l^2}{\theta_qn_l}} > \theta_q\left(n-\sqrt{n\left(n-\frac{d}{\theta_q}\right)}\right) \ .
  \end{align*}
  Rewriting as a condition on the normalized local decoding radius gives
  \begin{align*}
    \frac{\tau_l}{n_l} &> 2\theta_q-\frac{d}{n}\frac{1}{1-\sqrt{1-\frac{d}{n\theta_q}}}\\
                       &= \frac{2\theta_q\left(1-\sqrt{1-\frac{d}{n\theta_q}} \right)-\frac{d}{n}}{1-\sqrt{1-\frac{d}{n\theta_q}}}\\
                       &= \frac{\theta_q\left(1-\sqrt{1-\frac{d}{n\theta_q}} \right)^2}{1-\sqrt{1-\frac{d}{n\theta_q}}}\\
                       &= \theta_q\left(1-\sqrt{1-\frac{d}{n\theta_q}} \right) \ .
  \end{align*}
\end{IEEEproof}

\begin{remark}
  A different approach to the list-decoding problem would be to view each codeword from the local codes as a symbol over $\F{q^{n_l}}$, as, e.g., in \cite{goparaju2014binary}. Let $\code$ be an optimal $[n,k,r,\varrho]$ LRC. Viewing each local repair set as a symbol over $\F{q^{n_l}}$ gives an $\F{q}$-linear subcode of a code of length $\frac{n}{n_l}$ code over $\F{q^{n_l}}$. Denote this code by $\hat{\code}$ and note that error correction in this code allows to correct some \emph{number of incorrect local repair sets}. However, as the number of errors is the number of incorrect symbols over $\F{q}$ (and not $\F{q^{n_l}}$), a single error in a local repair set would cause the resulting symbol over $\F{q^{n_l}}$ to be incorrect.
    By Lemma~\ref{lem:sigma} we know that $\ceil{\sigma}$ local repair sets can be corrected locally, i.e., the list of at least $\ceil{\sigma}$ local repair sets contains the correct local codeword. In the context of viewing the local repair sets as symbols over $\F{q^{n_l}}$ this means that we can find a combination of $\ceil{\sigma}$ elements of the local lists that give correct \emph{positions} of the codeword of $\hat{\code}$. Further, it follows directly from the proof of Lemma~\ref{lem:sigma} that there exists a distribution of $\tau_g$ errors such that the remaining $\frac{n}{n_l}-\ceil{\sigma}$ local repair sets contain a number of errors that cannot be corrected locally. In this case, assuming we found the correct combination of local list entries, we can partition the symbols of the code $\hat{\code}$ into two disjoint (known) sets, one containing the $\ceil{\sigma}$ correct positions and the other containing the $\frac{n}{n_l}- \ceil{\sigma}$ incorrect positions. As it is \emph{always} beneficial in terms of error correction capability to shorten a code in positions that are known to be correct\footnote{We can assume that the shortened position is non-trivial (i.e., the position is not the same in all codewords), as otherwise it clearly provides no benefit for decoding. Trivially, the maximum amount of information a symbol can contain is one dimension. As shortening one non-trivial position reduces the dimension of the code by one, it is not possible to use this symbol in a way that provides more information about the received word.} the best strategy is to shorten the code in all $\ceil{\sigma}$ correct positions. This leaves only incorrect position, so the only possibility of correct decoding is that the shortened positions already contain an information set, in which case no further decoding is required after shortening.
	As the shortening strategy described above is the same as in the approach of Theorem~\ref{thm:ListDecodingLRCs}, this corresponds to a trivial case in Theorem~\ref{thm:ListDecodingLRCs} in which no global decoding (over $\F{q}$) is necessary.
	This implies that the application of the same approach while viewing the local repair sets as symbols over $\F{q^{n_l}}$ does not offer any advantage.
\end{remark}

\subsection{List-Decoding Algorithm} \label{subsec:listdecodingalgo}

To achieve the decoding radius of Theorem~\ref{thm:ListDecodingLRCs}, several steps have to be taken sequentially, as shown in \algoref{algo:fanta}. While Lemma~\ref{lem:sigma} guarantees that at least~$\ceil{\sigma}$ repair sets can be decoded, it does not guarantee that all repair sets for which the local decoder is able to return a local codeword are decoded correctly. For this reason, all combinations of seemingly correct local repair sets have to be tried in order to guarantee finding the correct one. Note that if this list of correctable local codes was known, the complexity could be decreased significantly, but with the chosen approach of only decoding locally in the first step, this is in general not possible.

\begin{algorithm}

	\KwData{ $[n,k,r,\varrho]_q$ LRC $\code$; Received word~$\ve{w}=\c+\ve{e}$ with~$\c\in \code$}
	\KwResult{List of codewords within radius~$\tau_g$ of~$\ve{w}$}

	\ForEach{Local code}{
		Decode up to~$\tau_l$~$\Rightarrow$~$\xi \geq \ceil{\sigma}$ repair sets with~$L_l \geq 1$}
	\ForEach{of the~$\binom{\xi}{\ceil{\sigma}}$ combinations of local repair sets with~$L_l \geq 1$ \label{step5}}{
		\ForEach{combination of codewords in the current~$\ceil{\sigma}$ local lists \label{step6}}{
			Shorten~$w$ and decode as~$(n-\ceil{\sigma} n_l,d)$ code up to radius~$\tau_g$}
		}

		Return all codewords~$\c'$ with~$\dt{\ve{w}}{\c'} \leq t_g$

	\caption{List Decoder} \label{algo:fanta}
\end{algorithm}

Note that \algoref{algo:fanta} can be improved in terms of complexity, e.g., by considering the number of errors corrected in the local codes and decreasing the decoding radius of the shortened code accordingly (see Remark~\ref{rem:list_size_bound_not_tight}).
Similarly, if the decoder has some additional knowledge which helps to correctly determine one or more local codewords, fewer combinations of repair sets have to be analyzed, thereby reducing the complexity.
However, as this is not the focus of this work, such performance optimizations are not considered here. Its complexity is polynomial in $n$ when the number of repair sets~$\frac{n}{n_l}$ is constant, as~$\xi= O(n^{\frac{n}{n_l}})$ grows exponentially otherwise. The bound on the list size given in (\ref{eq:listSize}) is obtained by assuming the worst case, i.e., the maximal list size, in each step of \algoref{algo:fanta}. As mentioned in Section~\ref{subsec:newdecodingradius}, this bound on the list size is not tight and we expect that the maximum list size, and thereby also the worst-case and average-case complexity of \algoref{algo:fanta}, is considerably overestimated.

\subsection{Probabilistic Unique Decoder} \label{subsec:probdec}
Even for a moderate number of local repair sets, the worst case complexity of \algoref{algo:fanta} can be rather high. In \stepref{step5} all combinations of corrected local repair sets have to be tried because an undetected error event might occur, i.e., a local code might return a list with~$L_l >0$ that does not contain the correct codeword. Further, in \stepref{step6}, all combinations of the codewords in the local lists have to be tried to guarantee finding one that consists only of correct local codewords. It follows that whether these steps are required depends on the probability of the local list size being larger than one and on the probability of a local list with~$L_l>0$ not containing the correct local codeword.

We define a \emph{probabilistic unique decoder}, where instead of shortening the code in every combination of elements of the lists of the local repair sets, we only decode a single shortened code. Explicitly, the decoder performs local list decoding, chooses $\frac{n}{n_l}-\left\lfloor \frac{t_g}{t_l+1} \right\rfloor$ of the repair sets with the smallest, non-zero list size, shortens the code in these repair sets, and decodes this shortened code. If the codewords from the lists of the local repair sets chosen for shortening are correct and the list decoder of the shortened code returns a unique codeword, the decoding is successful.
Note, that for this probabilistic decoder the constraint of~$\frac{n}{n_l} = \mathrm{const.}$ can be omitted, as its complexity grows only linear with the number of local repair sets.

\begin{theorem}[Probabilistic Decoding]\label{thm:probabilisticSuccessProb}
	An~$[n,k,r,\varrho]$ LRC can be uniquely decoded up to radius~$\tau_g$ of~\eqref{eq:jblrc} with probability
	\begin{equation}
		\Pr\{L_g=1\} \geq (1-P_E)^{\left\lfloor \frac{t_g}{t_l+1}\right\rfloor } \cdot \Pr\{L_{(n_l,\varrho,t_l)} = 1 \}^{\frac{n}{n_l} - \left\lfloor \frac{t_g}{t_l+1}\right\rfloor} \cdot \Pr\left\{L_{\left( \left\lfloor \frac{t_g}{t_l+1} \right\rfloor n_l,d,t_g\right)} = 1\right\}, \label{eq:sucprob}
	\end{equation}
	where $P_E$ denotes the maximum probability that an incorrect codeword is within distance~$t_l$ for any number of errors in the repair set.
\end{theorem}
\begin{IEEEproof}
	From Lemma~\ref{lem:sigma} it is known that there are at least $\frac{n}{n_l}- \left\lfloor \frac{t_g}{t_l+1}\right\rfloor$ local repair sets that can be corrected locally, i.e., where the list returned by a list decoder with radius $\tau_l$ contains the correct local codeword. The first term bounds the probability that the list of all other repair sets is either empty or also contains the correct local codeword. The second term gives the probability that the list size of the local repair sets that are guaranteed to contain the correct word, only contain this correct word. Finally, the third term gives the probability that the list returned by the list decoder of the shortened code also contains only the correct word.
\end{IEEEproof}
This is a pessimistic bound as it does not take the distribution of errors into account, i.e., at every step of the decoder it is assumed that the worst case number of errors occur. However, even this simple bound can give good results, as will be shown in Section~\ref{subsec:ProbTB} for a specific class of LRCs.

\subsection{Asymptotic Behavior} \label{subsec:asymptotic}

When considering codes without locality, the asymptotic behavior is usually characterized by regarding the normalized decoding radius over the normalized distance. For codes with locality the distance depends not only on the length and dimension, but also the locality~$r$ and local distance~$\varrho$, which yields different views on the asymptotic behavior.\\ 
Consider an optimal~$[n,k,r,\varrho]$ LRC with~$r\mid k$ and~${(r+\varrho-1)\mid n}$. By~\eqref{eq:boundDistanceLRC} the code rate is given by
\begin{equation}
	R = \left(1-\frac{d}{n} + \frac{\varrho}{n} \right) \frac{r}{r+\varrho-1}
	= \left(1-\frac{d}{n} + \frac{\varrho}{n_l} \frac{n_l}{n} \right) R_l \ , \label{eq:rate}
\end{equation}
where~$R_l$ denotes the rate of the local codes. It follows that the rate~$R$ only depends on the normalized distance~$\frac{d}{n}$, the local normalized distance~$\frac{\varrho}{n_l}$, the number of repair sets~$\frac{n}{n_l}$ and the local rate~$R_l$.

In the following we consider the asymptotic regime where the locality~$r$ and the local distance~$\varrho$ scale with the code length $n$ such that the number of repair sets~$\frac{n}{n_l}$ is constant, as well as the ratio~$\frac{d}{\varrho}$ between local and global distance.
Figure~\ref{fig:cw_asympt} gives a graphical illustration of this scaling, where~$(a)$ depicts a short codeword and~$(b)$ and~$(c)$ depict codewords of longer codes. Note that, as indicated by the marked redundancy, the short code has the same normalized distance as the other two. The difference between~$(b)$ and~$(c)$ is due to the scaling of the parameters, where for~$(b)$ the local distance and the repair set size are the same as in $(a)$, while for~$(c)$ both scale with~$n$.
We are interested in the latter, which can be interpreted in several ways, e.g., assume each repair set corresponds to a data center and the codeword symbols are distributed over several servers. Adapting the code to an increasing number of servers in each data center corresponds to increasing the size of each repair set while keeping the normalized distance (local storage overhead) constant. Thus, we characterize LRCs asymptotically by a fixed relation~$\beta = \frac{n \varrho}{n_l d}$ between the normalized local and global distance.
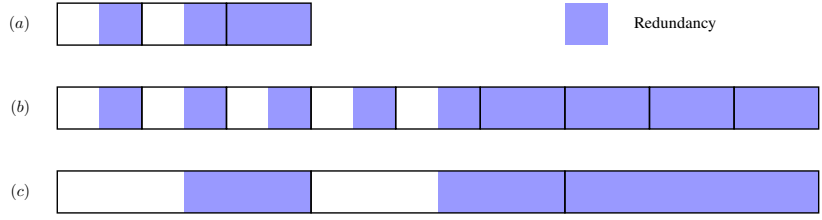
\begin{figure}
	\centering
	\resizebox{0.6\columnwidth}{!}{

\def\x{0.05\columnwidth} 

\begin{tikzpicture}

	\node[draw=none,anchor = east] at (\x*-0.5,\x*4.5)  {$(a)$};
	\node[draw=none,anchor = east] at (\x*-0.5,\x*2.5)  {$(b)$};
	\node[draw=none,anchor = east] at (\x*-0.5,\x*0.5)  {$(c)$};
	
	\fill[blue!40!white] (\x*12,\x*5) rectangle (\x*13,\x*4);
	\node[draw=none,anchor = west] at (\x*13.5,\x*4.5)  {Redundancy};
	
	\fill[blue!40!white] (\x*1,\x*5) rectangle (\x*2,\x*4);
	\fill[blue!40!white] (\x*3,\x*5) rectangle (\x*4,\x*4);
	\fill[blue!40!white] (\x*4,\x*5) rectangle (\x*6,\x*4);
	
	\draw[line width = 0.3mm] (\x*0,\x*5) rectangle (\x*2,\x*4);
	\draw[line width = 0.3mm] (\x*2,\x*5) rectangle (\x*4,\x*4);
	\draw[line width = 0.3mm] (\x*4,\x*5) rectangle (\x*6,\x*4);

	\fill[blue!40!white] (\x*1,\x*3) rectangle (\x*2,\x*2);
	\fill[blue!40!white] (\x*3,\x*3) rectangle (\x*4,\x*2);
	\fill[blue!40!white] (\x*5,\x*3) rectangle (\x*6,\x*2);
	\fill[blue!40!white] (\x*7,\x*3) rectangle (\x*8,\x*2);
	\fill[blue!40!white] (\x*9,\x*3) rectangle (\x*10,\x*2);
	\fill[blue!40!white] (\x*10,\x*3) rectangle (\x*18,\x*2);
	
	\draw[line width = 0.3mm] (\x*0,\x*3) rectangle (\x*2,\x*2);
	\draw[line width = 0.3mm] (\x*2,\x*3) rectangle (\x*4,\x*2);
	\draw[line width = 0.3mm] (\x*4,\x*3) rectangle (\x*6,\x*2);
	\draw[line width = 0.3mm] (\x*6,\x*3) rectangle (\x*8,\x*2);
	\draw[line width = 0.3mm] (\x*8,\x*3) rectangle (\x*10,\x*2);
	\draw[line width = 0.3mm] (\x*10,\x*3) rectangle (\x*12,\x*2);
	\draw[line width = 0.3mm] (\x*12,\x*3) rectangle (\x*14,\x*2);
	\draw[line width = 0.3mm] (\x*14,\x*3) rectangle (\x*16,\x*2);
	\draw[line width = 0.3mm] (\x*16,\x*3) rectangle (\x*18,\x*2);

	\fill[blue!40!white] (\x*3,\x*1) rectangle (\x*6,\x*0);
	\fill[blue!40!white] (\x*9,\x*1) rectangle (\x*12,\x*0);
	\fill[blue!40!white] (\x*12,\x*1) rectangle (\x*18,\x*0);
	
	\draw[line width = 0.3mm] (\x*0,\x*1) rectangle (\x*6,\x*0);
	\draw[line width = 0.3mm] (\x*6,\x*1) rectangle (\x*12,\x*0);
	\draw[line width = 0.3mm] (\x*12,\x*1) rectangle (\x*18,\x*0);
	
\end{tikzpicture}}
	\caption{Illustration of asymptotic scaling of parameters}
	\label{fig:cw_asympt}
\end{figure}

To compare our list-decoding radius~\eqref{eq:jblrc} with the Johnson radius~\eqref{eq:johnsonradius}, we write
\begin{equation}
\frac{\varrho}{n_l} = \beta \cdot \frac{d}{n} \ . \label{eq:rhonl3}
\end{equation}
Then for the case of $\ceil{\sigma}>0$ the normalized decoding radius of Theorem~\ref{thm:ListDecodingLRCs} is given by
\begin{align}
\frac{\tau_g}{n} &= \frac{d\tau_l}{n \varrho} = \frac{d}{n} \frac{\theta_qn_l}{\varrho}\left(1-\sqrt{1-\frac{\varrho}{n_l\theta_q}}\right)\nonumber\\
                 &= \beta^{-1}\theta_q\left(1-\sqrt{1-\beta\frac{d}{n\theta_q}}\right), \quad \text{with} \;  \beta \cdot \frac{d}{n\theta_q} \leq 1. \label{eq:asympnorm}
\end{align}
Thus, the normalized decoding radius of the global code, given by~\eqref{eq:asympnorm}, depends only on $\beta$, the normalized distance of the code, and the code alphabet.
When the normalized distances of the global and the local codes are equal, i.e.,~$\beta = 1$, the radius equals the $q$-ary Johnson radius~\eqref{eq:johnsonradius}. For any~$\beta>1$ our decoding radius provides a gain up to the point where~$\beta \frac{d}{n\theta_q}= 1$ and the curves meets the Singleton bound. Figure \ref{fig:asympplot} shows the normalized decoding radii for different values of~$\beta$ in the most general case of $q \rightarrow \infty$, i.e., $\theta = 1$, which is valid for any code, independent of the alphabet.

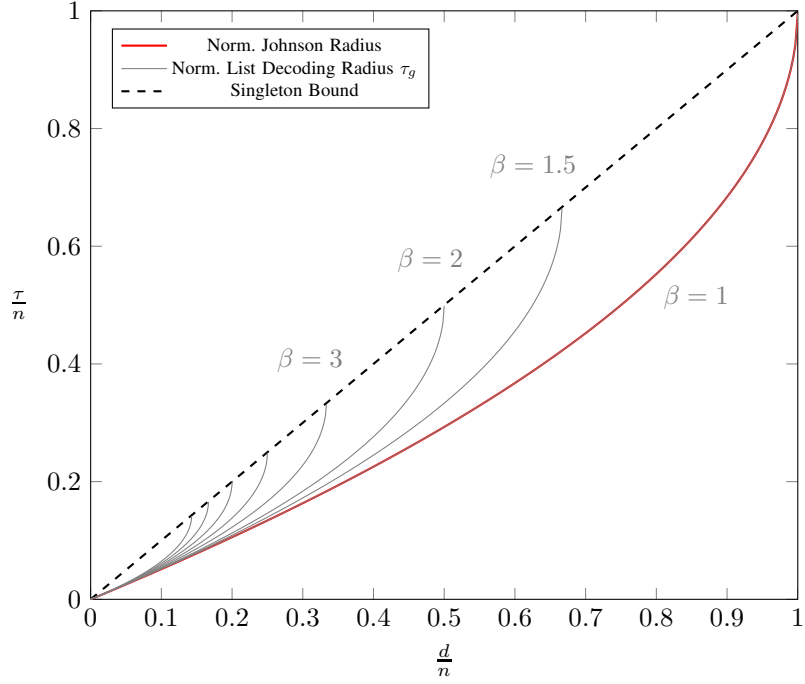
\begin{figure}[h]
	\centering
	\begin{tikzpicture}
\pgfplotsset{compat = 1.3}
\begin{axis}[
	legend style={nodes={scale=0.7, transform shape}},
	width = 0.6\columnwidth,
	xlabel = $\frac{d}{n}$,
	xlabel style = {nodes={scale=0.8, transform shape}},
	ylabel = $\frac{\tau}{n}$,
	ylabel style={rotate=-90,nodes={scale=0.85, transform shape}},
	xmin = 0,
	xmax = 1,
	ymin = 0,
	ymax = 1,
	legend pos = north west]

\addplot[color=red,
		domain = 0:1,
		samples = 300,
		thick]
		{1-sqrt(1-x)};
\addlegendentry{Norm. Johnson Radius}

\addplot[color=gray,
domain = 0:1,
samples = 300]
{x*(1/(2-(1-sqrt(1-min(1*x,1))))}
node[pos=0.66, anchor=north west] {$\beta = 1$};
\addlegendentry{Norm. List Decoding Radius $\tau_g$}

\addplot[color=black,
domain=0:1,
samples=2,
dashed,
thick]
{x};
\addlegendentry{Singleton Bound}

\addplot[color=gray,
domain = 0:1/1.5,
samples = 300]
{x*(1/(2-(1-sqrt(1-min(1.5*x,1))))};

\addplot[color=gray,
domain = 0:1/2,
samples = 300]
{x*(1/(2-(1-sqrt(1-min(2*x,1))))};

\addplot[color=gray,
domain = 0:1/3,
samples = 300]
{x*(1/(2-(1-sqrt(1-min(3*x,1))))};

\addplot[color=gray,
domain = 0:1/4,
samples = 300]
{x*(1/(2-(1-sqrt(1-min(4*x,1))))};

\addplot[color=gray,
domain = 0:0.201,
samples = 300]
{x*(1/(2-(1-sqrt(1-min(5*x,1))))};

\addplot[color=gray,
domain = 0:1/6,
samples = 300]
{x*(1/(2-(1-sqrt(1-min(6*x,1))))};

\addplot[color=gray,
domain = 0:1/7,
samples = 300]
{x*(1/(2-(1-sqrt(1-min(7*x,1))))};

\addplot[draw = none,
color=gray,
domain = 0:1,
samples = 2]
{x}
node[pos=0.7, anchor=south east] {$\beta = 1.5$}
node[pos=0.54, anchor=south east] {$\beta = 2$}
node[pos=0.37, anchor=south east] {$\beta = 3$};

\end{axis}
\end{tikzpicture}
	\caption{Normalized list-decoding radius~$\tau_g$ of Theorem~\ref{thm:ListDecodingLRCs}, for $\theta_q=1$, local decoding up to the $q$-ary Johnson radius, and~$\beta = \frac{n \varrho}{d n_l}$, compared to the normalized global Johnson radius for $\theta_q=1$.}
	\label{fig:asympplot}
\end{figure}

In Figure~\ref{fig:asympplot_constant_list_size}, we show a similar plot, but with separate curves for given maximal list sizes $L$ bounded by the formula in \eqref{eq:listSize} using the list size bound \eqref{eq:maximal_list_size} for the local and global list decoders.
For given
\begin{itemize}
\item number of local repair sets $\mu$ (here: $\mu=3$),
\item ratio $\beta = \frac{n \varrho}{n_l d}$ (here: $\beta \in \{1,1.5,2,3\}$), and
\item maximal list size $L$ (here: $L \in \{300,1000,3000,10000,100000\}$),
\end{itemize}
we plot the maximal relative decoding radius (achieved with maximal list size $L$ for all codes with the given parameters $\mu$, $\tfrac{d}{n}$, and $\beta$) over the relative minimum distance $\tfrac{d}{n}$.
\begin{figure}[h]
	\centering
	\input{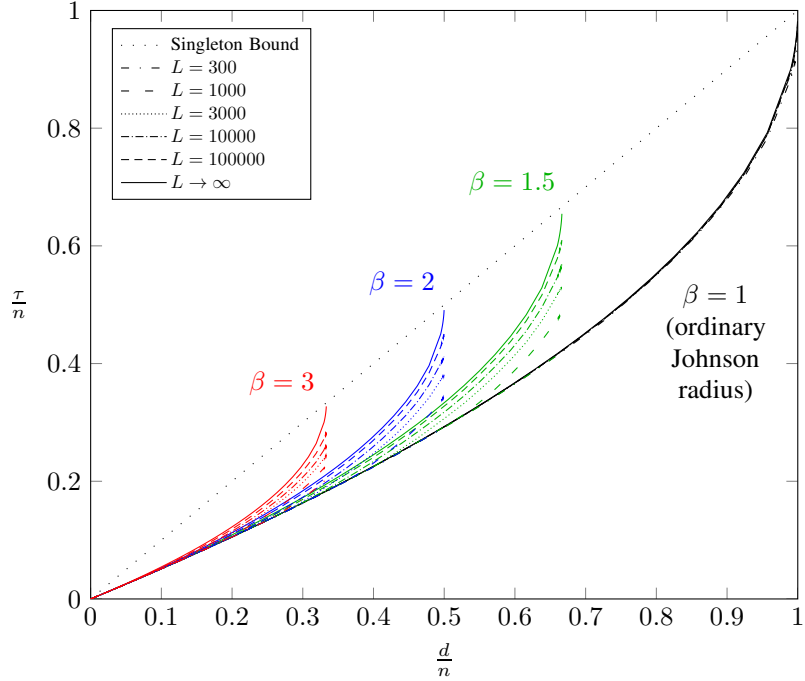}
	\caption{Normalized list-decoding radius~$\tau_g$ of Theorem~\ref{thm:ListDecodingLRCs} for a maximal list size $L$ according to \eqref{eq:listSize}, $\theta_q=1$, $\beta = \frac{n \varrho}{d n_l}$ and number of repair sets $\mu=3$.}
	\label{fig:asympplot_constant_list_size}
\end{figure}
We also show the limit of this curve for $L \to \infty$.\footnote{The asymptotic curves ($L \to \infty$) shown in Figure~\ref{fig:asympplot_constant_list_size} do not constitute the radii achieved by taking the relative decoding radius of any decoder with infinite list size---such a decoder could output the entire code and would have always relative radius $1$. Instead, we see the maximal relative decoding radius for the family of codes with fixed parameters $\mu$, $\beta$, and $\tfrac{d}{n}$ as a function of $L$ and compute the limit for $L \to \infty$. Note also that this family contains codes of arbitrary code length, and hence, for any given finite $L$, there is a code in the family whose cardinality is larger than $L$.}
The curves for $\beta=1$ equal the ordinary Johnson radii for the given maximal list sizes and the asymptotic limit (note that the difference between relative radii for finite $L$ and the asymptotic curve is so small that the curves almost coincide in this plot). Thus, a direct comparison of the Johnson radius and the decoding radius of the new decoder for the same fixed maximal list size is possible.

It can be seen that for $\mu=3$, we obtain improvements beyond the Johnson radius ($\beta=1$) for list sizes $L \geq 1000$, or $L \geq 300$ for $\beta=3$.
To get close to the asymptotic radius ($L \to \infty$), we need to allow huge list sizes according to \eqref{eq:listSize}.
Furthermore, this list size gets even larger for growing $\mu$ (this is not shown in the plot).
Recall, however, that \eqref{eq:listSize} is a very rough worst-case bound on the list size and we expect the actual maximal value to be much smaller.
We will see in Section~\ref{subsec:ProbTB} that for LRCs whose local and global codes are subcodes of MDS codes, we often have $\Pr(L=1) \approx 1$ for random errors, even close to the maximum decoding radius.

\section{Decoding Locally Repairable Subcodes of Reed--Solomon Codes}
\label{sec:ILRC}

In the previous section we proved a new list-decoding radius that is valid for any LRC with $(r+\varrho-1)\mid n$. Further, the proof of Theorem~\ref{thm:ListDecodingLRCs} directly implies a decoder based on list decoding of the local codes, shortening of the code, and list decoding of the shortened code. In general, decoding a code up to the Johnson radius is a difficult problem, but for specific code classes, such as GRS codes, such a decoder exists~\cite{Guruswami1999}. Based on this, we will give an explicit decoding algorithm for LRCs that are subcodes of GRS codes and where the local repair sets are (subcodes of) GRS codes (we will refer to these codes as \emph{GRS-subcode LRCs}), such as the popular class of Tamo--Barg LRCs~\cite{Tamo2014}.

\subsection{An Explicit Decoder for GRS-Subcode LRCs} \label{sec:tamobarg}

\algoref{algo:fanta} provides a decoding procedure up to the radius of \eqref{eq:jblrc}. To be feasible, it requires an efficient list-decoding algorithm of the global and local code, as well as an efficient way to shorten the code by known positions. While shortening is a commonly used method to decrease the length of a code, it is usually done at the encoder, where it suffices to set information symbols to zero. To shorten a code by some known positions at the decoder, all codewords that differ in the known positions need to be removed from the codebook. While this gives a code of desired distance and dimension, the structure of the code is lost and it is unclear how to decode in this newly obtained code. In this section, we address this problem for GRS codes and show how to efficiently apply \algoref{algo:fanta} to list decoding GRS subcode LRCs, such as the Singleton-optimal Tamo--Barg codes~\cite{Tamo2014}.

\algoref{algo:fanta} consists of three major steps: decoding locally, shortening the code, and decoding the shortened code. As, by definition, the local codes of a GRS subcode LRC are GRS codes, we can list-decode the $[n_l,r,\varrho]$ local codes up to the Johnson radius~\eqref{eq:johnsonradius}.
For shortening, denote the number of positions in a word $\ve{w}=\c+\ve{e}$ with $\c\in \code$ that are known to be free of error by $\delta$. The $[n,k,d]$ code $\code$ can be shortened by removing all codewords from the codebook that differ from $\ve{w}$ in these positions. The obtained code is an $[n-\delta,k-\delta,d]$ code, which is non-linear in general. Further, the structure of the shortened code is generally unknown, making efficient decoding difficult. To obtain a \emph{linear and structured} shortened code, we give a bijective map from the $[n-\delta,k-\delta,d]$ code to an $[n-\delta,k-\delta,d]$ GRS code.

\begin{definition}[Generalized Reed-Solomon Code]\label{def:GRS}
  Given a set of \emph{code locators} $\alpha_0, \ldots, \alpha_{n-1} \in \F{}$ with $\alpha_i \neq \alpha_j\  \forall \ i\neq j$ and a set of column multipliers $\nu_0,\ldots, \nu_{n-1} \in \F{} \setminus \{0\}$, the $[n,k]$ generalized Reed-Solomon code is defined as
  \begin{align*}
    \code \coloneqq \{ (\nu_0f(\alpha_0),\nu_1f(\alpha_1),\ldots,\nu_{n-1}f(\alpha_{n-1})) \ | \ f(x) \in \F{}[x], \deg(f(x)) < k \} \ .
  \end{align*}
\end{definition}
For ease of notation we define the following polynomial.
\begin{definition}\label{def:polred}
  For a polynomial $f(x)$ define
  \begin{align*}
    f^{\beta}(x) &\coloneqq  \frac{f(x)-f(\beta)}{x-\beta} \ .
  \end{align*}
  For an ordered subset $\mathcal{S} \subset \mathcal{A}$ define $f^{\mathcal{S}}(x)$ as the repeated application of the previous definition.
\end{definition}
\begin{lemma} \label{lem:RSred}
  Let $\code$ be an $[n,k,d]$ GRS code with code locators $\mathcal{A} = \{\alpha_0,\cdots, \alpha_{n-1}\}$. Then for any set $\mathcal{S} \subset \mathcal{A}$ with $|\cS| \leq k$ the code
  \begin{align*}
    \code^{\cS} \coloneqq \{ (\nu_0f^{\cS}(\alpha_0),\nu_1f^{\cS}(\alpha_1),\ldots,\nu_{n-1}f^{\cS}(\alpha_{n-1})) \ | \ f(x) \in \F{}[x], \deg(f(x)) < k \} \ .
  \end{align*}
  is an $[n-|\cS|,k-|\cS|, d]$ GRS code.
\end{lemma}
\begin{IEEEproof}
By Definition~\ref{def:GRS}, we need to show that for any $f(x)$ with $\deg(f(x)) < k$ it holds that ${\deg(f^{\cS}(x)) \leq k-|\cS|}$. The polynomial $f'(x) = f(x) - f(\alpha_0)$ has a root at $\alpha_0$ and hence the polynomial $f^{\alpha_0}(x)$ with $f'(x) = f^{\alpha_0}(x) (x-\alpha_0)$ exists. It follows that $ \deg(f^{\alpha_0}(x)) = \deg(f'(x))-1 = \deg(f(x)) -1 < k-1$. The generalization to $f^{\cS}(x)$ follows by induction.
\end{IEEEproof}
Since most positions in a codeword are free of error, we define a relation between the error vector of the shortened code and the original code. Then, instead of recovering the original codeword from the decoded shortened codeword, we can obtain the complete error vector directly from the error vector of the shortened code. For ease of notation we fix $\nu_0 = \ldots = \nu_{n-1} = 1$, i.e., we consider RS codes, in the following. Note that a codeword of an $[n,k,d]$ RS code with errors $e_0,...,e_{n-1}$ can be described as the evaluation of the polynomial
\begin{align*}
  g(x) = f(x) + \hat{e}(x)
\end{align*}
with
\begin{align}\label{eq:errorLagrange}
  \hat{e}(x) = \sum_{i\in [n]} e_i \prod_{\substack{j\in [n] \\ j\neq i}} \frac{x-\alpha_j}{\alpha_i-\alpha_j} = \sum_{i\in [n]} e_i \lambda_i(x) \ .
\end{align}
The polynomials $\lambda_i(x)$ are Lagrange basis polynomials and therefore $\hat{e}(\alpha_i) = e_i \ \forall \ i\in [n]$.
\begin{lemma}
  For a vector $\ve{e}\in\Fq^n$ denote $\mathcal{E} = \{i \ | \ e_i \neq 0, i \in [n] \}$. Given a set $\mathcal{A}=\{\alpha_0,...,\alpha_{n-1}\}$, a subset $\cS \subset \mathcal{A} \setminus \{\alpha_i \ | \ i \in \mathcal{E}\}$ with $|\cS| \leq k$, and a polynomial $f(x)$, let $g(x) \coloneqq f(x)+\hat{e}(x)$ with $\hat{e}(x)$ as in (\ref{eq:errorLagrange}) and $g^{\cS}(x)$ as in Definition~\ref{def:polred}. Then
	\begin{align*}
	g^{\mathcal{S}}_i(x) &= f^{\mathcal{S}}(x) + \hat{e}^{\mathcal{S}}(x)
	\end{align*}
    with
    \begin{align*}
      \hat{e}^{\mathcal{S}}(x) = \frac{\hat{e}(x)}{\prod_{\beta \in \mathcal{S}} (x-\beta)} \ .
    \end{align*}
\end{lemma}
\begin{IEEEproof}
	For $\delta>0$ and any $i \in \cS$, applying Definition~\ref{def:polred} gives
	\begin{align*}
	g^{\alpha_i}(x) &= \frac{(f(x)+\hat{e}(x)) - (f(\alpha_{i})+\hat{e}(\alpha_i))}{x-\alpha_{i}} \\
          &=  \frac{f(x) - f(\alpha_{i})}{x-\alpha_{i}} + \frac{\hat{e}(x) - \hat{e}(\alpha_i)}{x-\alpha_{i}} \\
	&= f^{\alpha_i}(x) + \hat{e}^{\alpha_i}(x) \ .
	\end{align*}
        The lemma statement follows from $\hat{e}(\alpha_i)=e_i = 0$ for $i\notin \mathcal{E}$ and $\mathcal{S} \cap \{\alpha_i \ | \ i \in \mathcal{E}\} = \emptyset$.
\end{IEEEproof}
With Lemma~\ref{lem:RSred} and the Guruswami-Sudan decoder~\cite{Guruswami1999}, all necessary tools for decoding up to the alphabet-independent radius $\tau_g$ of Theorem~\ref{thm:ListDecodingLRCs} are given.

\begin{figure}
	\centering
	\input{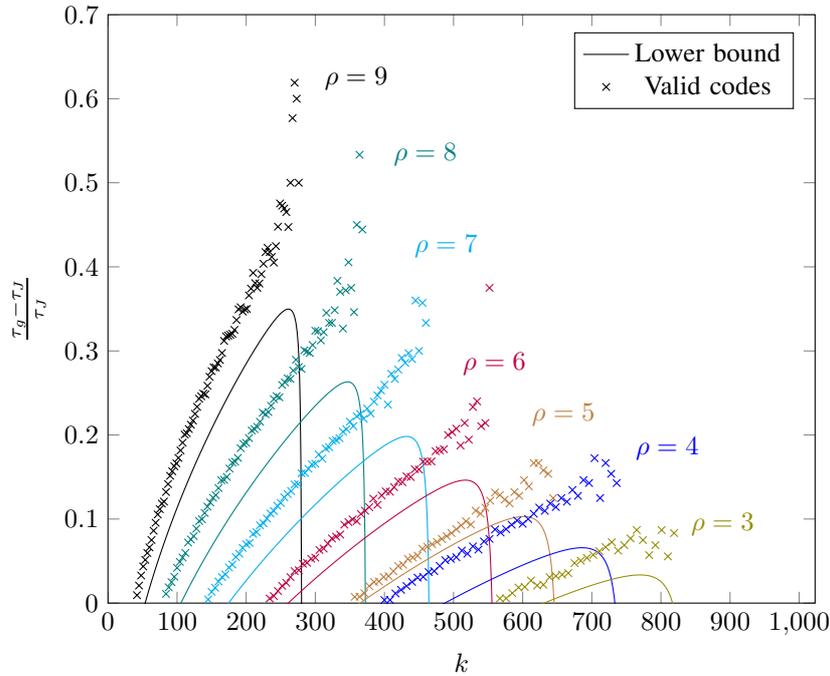}
	\caption{Relative gain in the decoding radius in relation to the alphabet-independent Johnson radius for optimal LRCs of length $n=1023$ and repair set size $n_l=11$, where $\tau_J$ denotes the Johnson radius \eqref{eq:johnsonradius} for the respective parameters, the lines are obtained from $\tau_g$ in Theorem~\ref{thm:ListDecodingLRCs} and for the crosses the number of correctable errors is given by the largest integer $\bar{t}_g$ that fulfills (\ref{eq:radiusWithFloor}).}
	\label{fig:relativegain}
\end{figure}

\figref{fig:relativegain} shows the relative gain for optimal LRCs of length $n=1023$ and repair set size $n_l=11$ for different values of $\varrho$. For each $\varrho$, a lower bound on the relative gain is given, i.e., the fraction by which our bound in Theorem~\ref{thm:ListDecodingLRCs} exceeds the Johnson radius of \eqref{eq:johnsonradius}. Each cross depicts the gain obtained for an LRC with $r\mid k$ and $(r+\varrho-1)\mid n$, when considering the exact values for all ceiling and floor operations as discussed in Remark~\ref{rem:radiusWithFloor}.

\subsection{Probabilistic Unique Decoding of GRS-Subcode LRCs} \label{subsec:ProbTB}

In Section~\ref{subsec:probdec} we introduced a simple probabilistic unique decoder whose success probability depends on the likelihood of a miscorrection, as well as the probability of the list sizes being equal to one.
For MDS codes, and thereby for RS codes, these probabilities are known to be small for a wide range of parameters \cite{McEliece1986,Cheung1988,McEliece2003}, which we will use to give numerical results and show the asymptotic behaviour of the probability of successful probabilistic unique decoding in GRS-Subcode LRCs. Here, we only consider GRS codes, but note that the same results apply also to MDS codes in general, if the MDS conjecture of $q \rightarrow \infty$ for $n\rightarrow \infty$ is correct.

Denote by $u$ the number of errors that actually occurred. Then, it was shown in \cite{McEliece1986,McEliece2003} that when list decoding\footnote{In \cite{McEliece1986} only unique decoding is considered, i.e., the case of $2t \leq d-1$. However, as noted in \cite{McEliece2003}, the bound also holds for $2t>d-1$.} up to $t$ errors in an $[n,k]$ RS code, the probability $P_E$ of a codeword other than the correct codeword being in the list, i.e., within distance $t$ of the received word, is bounded by
\begin{align}\label{eq:PERS}
  P_E(u) \leq \tilde{P}_E(n,d,q,t) = \frac{1}{(q-1)^{d-1}}\sum_{s=0}^t (q-1)^s \binom{n}{s} \ .
\end{align}
There are upper bounds that slightly improve upon this bound, in particular for $u\leq d-1$ \cite{McEliece1986,McEliece2003,Cheung1988}, but as these improvements are not substantial and complicate the following expressions we employ the bound (\ref{eq:PERS}), which is \emph{independent of the actual number of errors} $u$. Note that this bound also holds for any subcode of an $[n,k]$ RS code of distance $d$.

As the correct codeword is guaranteed to be in the list for a list decoder of radius $\tau$ if $\wt(e)\leq \left\lceil \tau -1 \right\rceil = t$, it follows that
\begin{align*}
  \Pr \{ L_{(n,d,t)}=1 | \wt(e) \leq t\} \geq 1-\tilde{P}_E(n,d,q,t) \ .
\end{align*}
Lemma~\ref{lem:sigma} guarantees that at least $\frac{n}{n_l}-\floor{\frac{\bar{t}_g}{t_l+1}}$ local repair sets contain at most $t_l$ errors, so the probability that all lists returned by the local decoders are either empty or contain only the correct local codeword is given by $\left(1-\tilde{P}_E(n_l,\varrho,q,t_l) \right)^{\frac{n}{n_l}}$ . If this is the case, the shortening of the code is guaranteed to be correct, i.e., the obtained word is of the form $\ve{c}'+\ve{e}$, where $\wt(\ve{e}) \leq \bar{t}_g$ and $\ve{c}'$ is a codeword of a code of length $n' = n- \left(\frac{n}{n_l} - \floor{\frac{\bar{t}_g}{t_l+1}} \right) n_l = \floor{\frac{\bar{t}_g}{t_l+1}} n_l$ and distance $d$. The shortening of any LRC that is a subcode of an GRS code such as, e.g., Tamo--Barg LRCs \cite{Tamo2014}, results in a subcode of an RS code. Hence, the probability of obtaining an unique decoding result for the shortened code can be bounded by (\ref{eq:PERS}). Overall, the probability $P_E$ of decoding success for unique decoding $\bar{t}_g$ errors, as defined in Remark~\ref{rem:radiusWithFloor}, in an $[n,k,r,\varrho]$ GRS-subcode LRC with the decoder described in Section~\ref{subsec:probdec} is upper bounded by
\begin{align}\label{eq:sucprobTB}
 \Pr\{L_g=1\} \geq \left(1-\tilde{P}_E(n_l,\varrho,q,t_l) \right)^{\frac{n}{n_l}} \left(1-\tilde{P}_E\left(\left\lfloor \frac{\bar{t}_g}{t_l+1} \right\rfloor n_l,d,q,\bar{t}_g\right)\right) \ .
\end{align}

\tabref{tab:sucprob} provides a lower bound on the success probabilities obtained by \eqref{eq:sucprobTB} for different LRC parameters. The columns labeled $\tau_J$, $\tau_{J,l}$, and $\tau_{g}$ give the $q$-ary Johnson radius, the $q$-ary local Johnson radius, and the radius of Theorem~\ref{thm:ListDecodingLRCs}, respectively. The values show that the computationally expensive case, where multiple repair sets have undetected error events and the local lists contain incorrect codewords, is highly unlikely and it is possible to efficiently decode beyond the global Johnson radius with a low probability of failure.

As noted in Section~\ref{subsec:listdecodingalgo}, the given upper bound on the list size can be very large. For example, using the introduced decoder to decode the maximum number of errors $t_g = 171$ in a $[500,99,33,68]_{2^9}$ LRC (cf. Table~\ref{tab:sucprob}), the upper bound of (\ref{eq:listSize}) on the obtained list size is given by $L_g \lessapprox 2.2 \times 10^6$, compared to the an upper bound of $\approx 476$ for decoding up to the Johnson radius ($t_j = 159$). However, the results of Table~\ref{tab:sucprob} show that despite the high bound on the \emph{maximum} list size, the probability that the output of the list decoder is unique is very close to $1$.

\begin{table}
  \caption{Success probabilities \eqref{eq:sucprob} of probabilistic unique decoding of~$t_g$ errors for different parameters of GRS-subcode LRCs, where $\tau_g$ is given by the \eqref{eq:jblrc} and $\bar{t}_g$ is the largest integer that fulfills (\ref{eq:radiusWithFloor}), both for $\theta_q=1$. The local and global alphabet-independent Johnson radii are denoted by $\tau_{J,l}$ and $\tau_J$, respectively.}\label{tab:sucprob}
\begin{center}
\begin{tabular}{CCCCCCC|CCCC|C}
n&k&r&\varrho&q&n_l&d&\substack{\mathrm{Eq.~}(\ref{eq:johnsonradius}) \\\tau_{J,l}} & \substack{\mathrm{Eq.~}(\ref{eq:johnsonradius}) \\\tau_J} & \substack{\mathrm{Thm.}~\ref{thm:ListDecodingLRCs} \\ \tau_g} & \substack{\mathrm{Eq.}~(\ref{eq:radiusWithFloor}) \\ \bar{t}_g} & \substack{\mathrm{Eq.}~(\ref{eq:sucprobTB})\\ \Pr\{L_g=1\}\geq} \\ \hline
1023& 99& 3& 9& 1024& 11& 669& 6.31& 421.22& 469.01& 491& 0.95973\\
1023& 99& 3& 9& 4096& 11& 669& 6.31& 421.22& 469.01& 491& 0.99744\\
1023& 99& 3& 9& 8192& 11& 669& 6.31& 421.22& 469.01& 491& 0.99936\\
 1023& 120& 4& 8& 1024& 11& 701& 5.26& 449.06& 460.51& 483& 0.95974\\
 1023& 120& 4& 8& 4096& 11& 701& 5.26& 449.06& 460.51& 483& 0.99744\\
 1023& 120& 4& 8& 8192& 11& 701& 5.26& 449.06& 460.51& 483& 0.99936\\
 1023& 220& 5& 7& 1024& 11& 546& 4.37& 324.45& 340.61& 354& 0.97108\\
 1023& 220& 5& 7& 4096& 11& 546& 4.37& 324.45& 340.61& 354& 0.99817\\
 1023& 220& 5& 7& 8196& 11& 546& 4.37& 324.45& 340.61& 354& 0.99954\\
 500& 99& 33& 68& 512& 100& 268& 43.43& 159.41& 171.17& 175& 1-10^{-35}\\
 500& 99& 33& 68& 1024& 100& 268& 43.43& 159.41& 171.17& 175& 1-10^{-42}\\
 500& 99& 33& 68& 2048& 100& 268& 43.43& 159.41& 171.17& 175& 1-10^{-50}\\
  63& 16& 8& 14& 64& 21& 35& 8.88& 21.0& 22.19& 24& 0.99938\\
 63& 16& 8& 14& 128& 21& 35& 8.88& 21.0& 22.19& 24& 0.99998\\
 63& 16& 8& 14& 256& 21& 35& 8.88& 21.0& 22.19& 24& 1-10^{-6}
\end{tabular}
\end{center}

\end{table}
We now consider the asymptotic regime. First note that for GRS codes the distance is given by $d=n-Rn+1$ and the alphabet-independent, i.e., $\theta_q=1$, Johnson radius (\ref{eq:johnsonradius}) can be expressed in terms of the rate as $\tau_J = n\left(1-\sqrt{R-\frac{1}{n}}\right)$. For $n\rightarrow \infty$, any constant rate $R<1$, and $t_J$ given by the alphabet-independent Johnson radius~(\ref{eq:johnsonradius}), it holds that
\begin{align*}
\tilde{P}_E(n,d,q,t_J) &= \frac{1}{(q-1)^{d-1}}\sum_{s=0}^{t_J} (q-1)^s \binom{n}{s}\\
&\leq \frac{1}{(q-1)^{d-1}}q^{t_J} \binom{n}{t_J} \\ 
&\leq \frac{1}{(q-1)^{d-1}}q^{t_J} \left(\frac{n \cdot e}{t_J}\right)^{t_J} \\
&\leq q^{-(d-1)\log_q(q-1)+t_J\left[ 1 + \log_q\left(\frac{n \cdot e}{t_J} \right) \right]}
\end{align*}
Since for RS code $q \to \infty$ for $n \to \infty$, we have $\log_q(q-1) \to 1$ for $n \to \infty$ and
\begin{align*}
\log_q\left(\frac{n \cdot e}{t_J}\right) &= \log_q\left(\frac{e}{1-\sqrt{R-\tfrac{1}{n}}}\right) \to 0 \quad (n \to \infty).
\end{align*}
Hence, we have
\begin{align*}
 \lim_{n\rightarrow \infty}\tilde{P}_E(n,d,q,t_J) &\leq \lim_{n,q\rightarrow \infty} q^{-(d-1)\log_q(q-1)+t_J\left[ 1 + \log_q\left(\frac{n \cdot e}{t_J} \right) \right]} \\
 &= \lim_{n,q\rightarrow \infty} q^{t_J-(d-1)} \\
 &= \lim_{n,q\rightarrow \infty} q^{n\left(1-\sqrt{R-\frac{1}{n}}\right)-(n-Rn+1)+1}\\
  &= \lim_{n,q\rightarrow \infty} q^{n\left(R-\sqrt{R-\frac{1}{n}}\right)}\stackrel{(\ast)}{=} 0 \ ,
\end{align*}
where $(\ast)$ holds because $\lim_{n\rightarrow \infty} \Big( R-\sqrt{R-\frac{1}{n}}\Big) < 0$ for any $0<R<1$.

Now consider the asymptotic regime as defined in Section~\ref{subsec:asymptotic}, i.e., a fixed number of local repair sets $\frac{n}{n_l}$, a fixed code rate $R=\frac{k}{n}$, and a normalized local distance of $\frac{\varrho}{n_l}=\beta\cdot \frac{d}{n}$ for some constant $\beta>1$. 
In this regime both the local and the global distance grow linearly in $n$. Hence the asymptotic success probability of the unique decoder, when decoding up to the alphabet-independent Johnson radius locally and globally, is
\begin{align*}
  \lim_{n\rightarrow \infty} \Pr\{L_g=1\} = \lim_{n\rightarrow \infty} (1-\tilde{P}_E(n_l,\varrho,q,t_l) )^{\frac{n}{n_l}} \left(1-\tilde{P}_E\left(\left\lfloor \frac{t_g}{t_l+1} \right\rfloor n_l,d,q,t_g\right)\right) = 1 \ .
\end{align*}

\subsection{Improved Decoding Using Interleaved Reed--Solomon Codes}
\label{subsec:decSupercode}

In data storage, as in data transmission, codes over small fields are generally favorable as they allow for lower complexity decoding of errors or recovery from erasures.
Hence, several codewords are stored simultaneously on a set of servers (cf.~Figure~\ref{fig:illustration}).
These codewords can be viewed as a codeword of an interleaved code.
Furthermore, if errors occur, they are likely to affect a server or hard drive sector as discussed in the introduction.
In the interleaved code interpretation of the stored data, this corresponds to a burst error and collaborative decoding of the stored data promises an improvement over separate decoding.
An advantage of interleaved codes, e.g., compared to other codes over a larger alphabet, is that in most cases, i.e., when erasures or only up to $<d/2$ errors occur, it is sufficient to consider each codeword separately, thereby keeping the decoding complexity low. Only in a worst-case scenario where $\geq d/2$ errors occur, the stored codewords can be viewed as an interleaved code, hence increasing the decoding radius and resolving the errors with high probability.

As discussed in Section~\ref{sec:tamobarg}, any LRC that is a subcode of a GRS code, e.g., Tamo--Barg LRCs \cite{Tamo2014}, can be decoded using a decoder of its GRS supercode. Hence, also its $\ell$-interleaved codes can be decoded using an interleaved decoder for GRS codes. For instance, the decoder in \cite{schmidt2009collaborative} can correct a random burst error of weight $t < t_{\mathrm{max}} \coloneqq \tfrac{\ell}{\ell+1}(d-1)$, where $d$ is the minimum distance of the constituent LRC, with success probability\footnote{By failure, we mean that the decoder does not return a codeword and by miscorrection if it returns a codeword different to the transmitted one. Since we consider subcodes of GRS codes, some miscorrections will turn into failures (i.e., if a miscorrected GRS codeword is not an element of the LRC subcode), but the sum of miscorrections and failures will stay the same.} $1-\Pfailure-\Pmiscorrection$ where $\Pfailure$ and $\Pmiscorrection$ are defined as in \cite{schmidt2009collaborative} ($\Pfailure$ goes to $0$ exponentially in $t-t_\mathrm{max}$ and usually $\Pmiscorrection \ll \Pfailure$, see~\cite{schmidt2009collaborative} for more details).
Already for small code parameters, this probability can be quite close to $1$ if the interleaving degree is high (which is usually the case in a real storage system) as the following example shows.

\begin{example}
Consider an $[n=15,k=8,r=4,\varrho=2]$ storage code of distance $d=7$ operating on bytes, i.e., over the field $F_{2^8}$. The unique decoding radius of this code is $\left\lfloor \frac{d-1}{2} \right\rfloor = 3$. Now assume burst errors occurring on hard-drive sectors of typical size $512$ bytes. This results in an interleaving order of $\ell = 512$ and an interleaved decoding radius of $t = 5$. The bound mentioned above gives a success probability $> 1- 10^{-1223}$.
\end{example}

For GRS subcode LRCs we can combine interleaved decoding and the local-global decoding strategy introduced in Section~\ref{sec:listDecoding}. We will analyze the resulting decoding radius in the following.

There are various decoding algorithms for interleaved GRS codes, e.g., \cite{krachkovsky1997decoding,bleichenbacher2003decoding,coppersmith2003reconstructing,parvaresh2004multivariate,brown2004probabilistic,parvaresh2007algebraic,schmidt2007enhancing,schmidt2009collaborative,cohn2013approximate,nielsen2013generalised,wachterzeh2014decoding,puchinger2017irs,yu2018simultaneous}. Among these decoders, the maximal ``expected'' decoding radius for a given code of length $n$ and minimum distance $d$ is
\begin{align}
\tauIRS = n\left(1- \left(\frac{n-d}{n}\right)^{\frac{\ell}{\ell+1}} \right), \label{eq:IRS_radius}
\end{align}
which can be achieved in polynomial time by the algorithms in \cite{parvaresh2004multivariate,cohn2013approximate,puchinger2017irs}.

We use the term ``expected'' here since for none of the above mentioned algorithms we have a decoding guarantee comparable to a polynomial-time list decoder. To be precise, the algorithms can be subdivided into two classes: list decoders with exponential worst-case list size and partial unique decoders that fail for some error patterns. We are not aware of bounds on the probability (given a random error of prescribed weight) that the list size is small or that decoding fails, for the decoders in \cite{parvaresh2004multivariate,cohn2013approximate,puchinger2017irs} at error weight close to the maximal decoding radius \eqref{eq:IRS_radius}. However, numerical results indicate that decoding succeeds with very high probability up to the given radius.

To obtain a precise statement in the following, we assume that we know the probability of \emph{unique decoding success} $\Psucc(n,d,\ell,t,\mathrm{dec})$ for a given decoder $\mathrm{dec}$ and number of errors $t \leq \tauIRS$ at the maximal decoding radius $\tauIRS$, i.e., the probability that a list decoder ($\mathrm{dec}$ being one of the algorithms in \cite{parvaresh2004multivariate,cohn2013approximate}) returns a list of size $1$ or that a partial unique decoder ($\mathrm{dec}$ as in \cite{puchinger2017irs}) succeeds in finding a unique closest codewords to the received word (the probability sample space is the set of error words of weight $t$). Further, we assume that $\Psucc(n,d,\ell,t_1,\mathrm{dec}) \geq \Psucc(n,d,\ell,t_2,\mathrm{dec})$ for any $t_1 < t_2 \leq \tauIRS$, i.e.\ the success probability is not smaller for fewer errors.

Using similar arguments as in the proofs of Theorem~\ref{thm:ListDecodingLRCs}\footnote{In Theorem~\ref{thm:ListDecodingLRCs}, we used the alphabet-dependent Johnson radius to derive the decoding radius of the new decoder. Although \eqref{eq:IRS_radius} appears to be a natural generalization of the alphabet-independent Johnson radius formula for $\ell>1$, we are not aware of any work that achieves an alphabet-dependent bound greater than \eqref{eq:IRS_radius}. Such an improvement seems to be possible by combining \cite{parvaresh2004multivariate,cohn2013approximate,puchinger2017irs} and the algorithm in \cite{augot2011list} for decoding alternant codes up to the alphabet-dependent Johnson radius. Nevertheless, we solely use the alphabet-independent bound here since this extension is out of the scope of this paper.} and Theorem~\ref{thm:probabilisticSuccessProb}, we can in principle obtain a probabilistic unique decoder for interleaving degree $\ell$, where for $\ceil{\sigma}$ correctly decoded local codes, the resulting decoding radius is the largest $\tau$ fulfilling
\begin{equation*}
0 < ( (n-\ceil{\sigma} n_l)-\tau)^{\ell+1} - (n-\ceil{\sigma} n_l)((n-\ceil{\sigma} n_l)-d)^\ell \ .
\end{equation*}

\begin{theorem}\label{thm:interleavedGeneral}
Consider an $\ell$-interleaved $[n,k,r,\varrho]_q$ LRC which is a subcode of an $[n,k',d]_q$ GRS code and where every local code is a (subcode of a) $[r+\varrho-1,r,\varrho]_q$ GRS code. Then there is an efficient decoding algorithm correcting $t_g=\ceil{\tau_g-1}$ errors, where $t_g$ is the largest integer such that
\begin{equation}
0 < ( (n-\ceil{\sigma} n_l)-t_g)^{\ell+1} - (n-\ceil{\sigma} n_l)((n-\ceil{\sigma} n_l)-d)^\ell, \label{eq:interleaved_max_radius_general}
\end{equation}
that succeeds with probability
\begin{align*}
  P_{suc} \geq (1-\tilde{P}_E(n_l,\varrho,q^\ell,t_l))^{\left\lfloor \frac{t_g}{t_l+1}\right\rfloor } \cdot \Psucc(n_l,\varrho,\ell,t_l,\mathrm{dec})^{\frac{n}{n_l} - \left\lfloor \frac{t_g}{t_l+1}\right\rfloor} \cdot \Psucc(n-\ceil{\sigma} n_l,d,\ell,t_g,\mathrm{dec}),
\end{align*}
where $\mathrm{dec}$ is one of the decoders in \cite{parvaresh2004multivariate,cohn2013approximate,puchinger2017irs} and $t_l$ is the number of errors locally correctable by these decoders.
\end{theorem}

\begin{proof}
  The maximal radius follows directly from the same argument as in the proof of Theorem~\ref{thm:ListDecodingLRCs}.

  Similar to Theorem~\ref{thm:probabilisticSuccessProb}, the success probability consists of three terms. The first term gives the probability that all repair sets that are not used for shortening return either a nothing or the correct decoding result. It is well known that an $\ell$-interleaved GRS code over a field $\F{q}$ is a GRS code over $\F{q^\ell}$. Hence, the bound of \cite{McEliece2003} given in (\ref{eq:PERS}) applies, with the field size $q^\ell$. The second term gives the probability that the decoder returns a unique result for the $\ceil{\sigma}$ repair sets that are used for shortening, i.e., are guaranteed to have $\leq t_l$ errors. The third term gives the probability that the shortened code is decoded successfully.
\end{proof}

For $\ell=2$, we can resolve the condition in \eqref{eq:interleaved_max_radius_general} and obtain the following radius.

\begin{theorem}\label{thm:RadiusRSSubcodes}
A $2$-interleaved $[n,k,r,\varrho]_q$ LRC which is a subcode of an $[n,k,d]_q$ GRS code and where every local code is a $[r+\varrho-1,r,\varrho]_q$ GRS code can be efficiently decoded up to a radius of
  \begin{align*}
\tau_g = d \cdot \frac{2-\tfrac{\varrho}{n_l}}{\left(1-\frac{\varrho}{n_l}\right)^{4/3}+\left(1-\frac{\varrho}{n_l}\right)^{2/3}+1}
\end{align*}
with success probability
\begin{align*}
  P_{suc} \geq (1-\tilde{P}_E(n_l,\varrho,q^\ell,t_l))^{\left\lfloor \frac{t_g}{t_l+1}\right\rfloor } \cdot \Psucc(n_l,\varrho,2,t_l,\mathrm{dec})^{\frac{n}{n_l} - \left\lfloor \frac{t_g}{t_l+1}\right\rfloor} \cdot \Psucc(n-\ceil{\sigma} n_l,d,2,t_g,\mathrm{dec}),
\end{align*}
where $\mathrm{dec}$ is one of the decoders in \cite{parvaresh2004multivariate,cohn2013approximate,puchinger2017irs}.
\end{theorem}
\begin{IEEEproof}
  It follows from \ref{lem:increasingInN} that replacing $\ceil{\sigma}$ with $\sigma$ in the condition of Theorem~\ref{thm:interleavedGeneral} gives a valid decoding radius. It follows that for $\ell=2$ and $\sigma$ as defined as in \ref{lem:sigma}, the decoding radius is given by the largest $\tau_g$ such that
\begin{align*}
0&\leq \left(\frac{\tau_g (n_l-\tau_l)}{\tau_l} \right)^{3} - \frac{n_l \tau_g}{\tau_l} \left(\frac{n_l \tau_g-d\tau_l}{\tau_l} \right)^{2} \nonumber \\
&= \tau_g^3 \frac{n_l^3-3n_l^2\tau_l+3n_l\tau_l^2-\tau_l^3}{\tau_l^3}- \tau_g \frac{n_l}{\tau_l^3} \big[n_l^2\tau_g^2-2n_l\tau_g d \tau_l + d^2 \tau_l^2\big] \nonumber \\
\overset{\tau_l>0, \tau_g>0}{\Longleftrightarrow} \quad 0 &\leq \tau_g^2 \Big[-3n_l^2+3n_l\tau_l-\tau_l^2 \Big] + \tau_g \Big[ 2 n_l^2 d \Big] + \Big[ -n_l d^2 \tau_l \Big] \label{eq:quadratic_inequality}
\end{align*}
Since $-3n_l^2+3n_l\tau_l-\tau_l^2<0$, the inequality is fulfilled between the two roots of the polynomial in $\tau_g$, i.e.,
\begin{align*}
\tau_{g,1/2} = \frac{2n_l^2d \pm \sqrt{4 n_l^4d^2 - 4n_l d^2\tau_l (3 n_l^2-3n_l\tau_l + \tau_l^2)}}{2(3 n_l^2-3n_l\tau_l + \tau_l^2)}.
\end{align*}
With $\tau_l = n_l(1-(\tfrac{n_l-\varrho}{n_l})^{2/3})$, we have
\begin{align*}
\tau_{g,1/2} &= d \cdot\frac{1 \pm \sqrt{1 - (1-(\tfrac{n_l-\varrho}{n_l})^{2/3}) (3 - 3(1-(\tfrac{n_l-\varrho}{n_l})^{2/3}) + (1-(\tfrac{n_l-\varrho}{n_l})^{2/3})^2)}}{3 - 3(1-(\tfrac{n_l-\varrho}{n_l})^{2/3}) + (1-(\tfrac{n_l-\varrho}{n_l})^{2/3})^2} \\
&= d \cdot \frac{1 \pm (\tfrac{n_l-\varrho}{n_l})}{(\tfrac{n_l-\varrho}{n_l})^{4/3}+(\tfrac{n_l-\varrho}{n_l})^{2/3}+1} \\
\end{align*}
and the theorem statement follows.
\end{IEEEproof}

In Table~\ref{tab:exampleParameters} we give some example of parameters and the respective decoding radii achieved by the different decoders of Section~\ref{subsec:listdecodingalgo} and Section~\ref{subsec:decSupercode}. Note that especially for LRCs with a small number of local repair sets $\mu=\frac{n}{n_l}$, the number of correctable errors $\bar{t}_g$ when considering the floors in the derivation of the number of correctable errors (cf. Remark~\ref{rem:radiusWithFloor}), can give a large relative improvement compared to the radius of Theorem~\ref{thm:ListDecodingLRCs}.
\begin{table}\centering
\caption{Decoding radii for different parameters. The alphabet-independent Johnson-radii for the local and the global code are denoted by $\tau_{J,l}$ and $\tau_{J,g}$ respectively. The radius for $\ell=2$-interleaved decoding given in (\ref{eq:IRS_radius}) is denoted by $\tau_{J,\ell=2}$. The radius of Theorem~\ref{thm:ListDecodingLRCs} is given by $\tau_g$ and the number of correctable errors when considering the floors by $\bar{t}_g$. The radius of Theorem~\ref{thm:RadiusRSSubcodes} is denoted by $\tau_{g,\ell=2}$.}

\begin{tabular}{CCCCCCC|CC||CC|CC||C|C}
  n&k&r&\varrho&q&n_l&d&R_{\text{global}} & R_{\text{local}} & \substack{\text{Eq. (\ref{eq:johnsonradius})} \\ \tau_{J,l}}&\substack{\text{Eq. (\ref{eq:johnsonradius})}\\ \tau_{J,g}}&\substack{\text{Thm.~\ref{thm:ListDecodingLRCs}}\\\tau_{g}} & \substack{\text{Eq.~(\ref{eq:radiusWithFloor})} \\\bar{t}_g} & \substack{\text{Eq.~(\ref{eq:IRS_radius})}\\\tau_{J,\ell=2}}&\substack{\text{Thm.~\ref{thm:RadiusRSSubcodes}}\\\tau_{g,\ell=2}} \\ \hline
15& 6& 3& 3& \infty& 5& 8&0.40 & 0.60 & 1.84& 4.75& 4.9& 5& 5.98& 6.09\\
 30& 16& 4& 3& \infty& 6& 9&0.53 &0.66& 1.76& 4.9& 5.27& 5 & 6.35& 6.66\\
 30& 15& 3& 3& \infty& 5& 8&0.50 & 0.60 & 1.84& 4.31& 4.9& 5 & 5.6& 6.09\\
 63& 16& 8& 14& \infty& 21& 35& 0.25 & 0.38 & 8.88& 21& 22.19& 24 & 26.31& 27.26\\
 63& 40& 5& 3& \infty& 7& 10& 0.63 & 0.71 & 1.71& 5.22& 5.69& 5 & 6.86& 7.27\\
 500& 99& 33& 68& \infty& 100& 268& 0.20 & 0.33 & 43.43& 159.41& 171.17& 175 & 200.33& 209.73
\end{tabular}

\label{tab:exampleParameters}
\end{table}

\section{Decoding of PMDS codes beyond their Minimum Distance}
\label{sec:PMDS}

In this section, we focus on interleaved PMDS (maximally recoverable) codes, as in Definition~\ref{def:partialMDS}. For these codes we propose an explicit decoder that can correct \emph{beyond the minimum distance} for many error patterns (i.e., error positions).
The results do not assume any structure on the underlying PMDS code other than the PMDS property and a sufficiently large interleaving order. Hence, it is a generic decoder for interleaved PMDS codes that runs in cubic time in the code length independent of the chosen PMDS code.
The idea is based on interleaving a PMDS code (which is again a PMDS code, just over a larger alphabet) and then applying the Metzner--Kapturowski algorithm \cite{metzner1990general}, which is an efficient decoder for generic interleaved codes of high interleaving order that corrects up to $d-2$ errors if the error matrix has full rank.

We derive a new condition on the set of error positions which allows the decoder to correct more than $d-2$ errors.
For PMDS codes, we derive bounds on the number of such sets corresponding to errors of a given weight, and show that for many PMDS codes, the relative number of error patterns of weight $n-k-1$ that can be corrected is close to $1$.

Among others, we show that any family of (high-order interleaved) PMDS codes that contains codes of length greater than any given integer and any code that fulfills the rate restriction
\begin{equation*}
 \left(\frac{R_\mathsf{local}}{e}\right)^{\varrho-2} > R_\mathsf{global} \ ,
\end{equation*}
contains codes that can correct a fraction of error patterns of weight $n-k-1$ arbitrarily close to $1$.

\subsection{A Generalization of Metzner and Kapturowski's Statement} \label{sec:metznerKapturowski}

Metzner and Kapturowski proved in \cite[Theorem~2]{metzner1990general} that a codeword $\C$ of an interleaved code with minimum distance $d$ can be uniquely recovered from a corrupted word $\C+\E$ if
\begin{enumerate}
\item The number of errors is $t \coloneqq |\Eset| \leq d-2$ and
\item the error matrix $\E$ has full rank $t$ (this implicitly assumes that the interleaving order is high, i.e., $\ell \geq t$).
\end{enumerate}
However, the first condition is very restrictive and not necessary for the decoder to work.
In fact, the proof of \cite[Theorem~2]{metzner1990general} only assumes an implication of the first property: The $t+1$ columns of the parity-check matrix indexed by the error positions $\Eset$ and any other integer in $\{1,\dots,n\} \setminus \Eset$ must be linearly independent.
We will give this property a name in the following definition.

\begin{definition}\label{def:t+1-independent}
Let $\H \in \Fq^{(n-k) \times n}$ be a parity-check matrix of a linear $[n,k,d]$ code $\Code$.
A set $\Eset \subseteq \{1,\dots,n\}$ with $t = |\Eset|$ is called \emph{$(t+1)$-independent (with respect to $\H$)} if
\begin{align*}
\rank \left(\H_{\Eset \cup \{i\}}\right) = t+1	\quad \forall i \in \{1,\dots,n\} \setminus \Eset,
\end{align*}
where $\H_{\Eset \cup \{i\}}$ is the matrix consisting of the columns of $\H$ indexed by $\Eset \cup \{i\}$.
\end{definition}

Note that for $t\leq d-2$, any set is $t+1$-independent, and for $t\geq n-k$, no set is $t+1$-independent.
Using this definition, we can formally state a generalization of the result in \cite[Theorem~2]{metzner1990general}.
The proof resembles Metzner and Kapturowski's argumentation, but we include it for completeness.

\begin{theorem}\label{thm:MK_generalization}
Let $\Code$ be a linear code with parity-check matrix $\H$ and $\E \in \Fq^{\ell \times n}$ be an error matrix whose $t$ non-zero columns are indexed by the set $\Eset \subset \{1,\dots,n\}$.
Furthermore, let $\S = \H \E^\top \in \Fq^{(n-k) \times \ell}$ be the syndrome matrix, $\P \in \Fq^{(n-k) \times (n-k)}$ be an invertible matrix such that $\P \S$ is in reduced row Echelon form, and $\zeta$ be the number of zero rows in $\P \S$ (i.e. $\zeta = n-k-\rank(\S)$). We denote the lowest $\zeta$ rows of $\P \H$ by $\Q \in \Fq^{\zeta \times n}$.
If $\E$ satisfies the conditions
\begin{enumerate}[label=\roman*)]
\item\label{itm:t+1-independence_condition} $\Eset$ is $(t+1)$-independent w.r.t.\ $\H$ (\textbf{${\boldsymbol{(t+1)}}$-independence condition}) and
\item $\rank(\E)=t$ (\textbf{full-rank condition}),
\end{enumerate}
then the zero columns of $\Q$ are exactly the error positions $\Eset$.
\end{theorem}

\begin{proof}
Since the rank of the error equals its number of non-zero columns, a vector is in the right kernel of $\E$ if and only if it is zero at the error positions $\Eset$. Due to $\0 = \Q \E^\top$ ($\Leftrightarrow$ $\E \Q^\top = \0$), the rows of $\Q$ must be in the right kernel of $\E$ and thus $\Q$ has zero columns at the error positions.

We prove that the other columns are non-zero. Assume the contrary, i.e., let $i \notin \Eset$ be a zero column of $\Q$.
Then, the matrix $(\P \H)_{\Eset \cup \{i\}}$, consisting of the $t+1$ columns of $\P\H$ indexed by $\Eset \cup \{i\}$, has at least $\zeta = n-k-\rank(\S) \geq n-k-t$ zero rows (recall that $\Q$ is, by definition, the submatrix of $\P\H$ consisting of the last $\zeta$ rows, and thus, as $\Q$ is zero in all columns indexed by $\Eset \cup \{i\}$, the last $\zeta$ rows of $(\P \H)_{\Eset \cup \{i\}}$ are zero).
Hence, the left kernel of $(\P \H)_{\Eset \cup \{i\}}$ has dimension at least $n-k-t$ (e.g., the identity vectors $\e_{t+1},\dots,\e_{n-k}$ are contained in the kernel).
On the other hand, as $\Eset$ is $(t+1)$-independent, we have $\rank((\P \H)_{\Eset \cup \{i\}}) = t+1$ and by the rank nullity theorem, the left kernel of $(\P \H)_{\Eset \cup \{i\}}$ has dimension $n-k-t-1$, which is a contradiction.
\end{proof}

The decoding strategy implied by Theorem~\ref{thm:MK_generalization} is outlined in Algorithm~\ref{alg:MK}.
Note that Line~\ref{line:erasure_correction} performs erasure correction after determining the error positions. This gives a unique result if $\Eset$ is $t+1$-independent since the columns of $\H$ indexed by $\Eset$ are linearly independent, hence the complementary columns of a generator matrix are linearly independent, i.e., $\{1,\dots,n\} \setminus \Eset$ contains an information set.
As the decoder is only based on linear-algebraic operations, it can be implemented in $O(\max\{\ell,n\}n^2)$ operations over the base field $\Fq$ of the code $\Code$.

Note that the second condition, $\rank\left(\E_\Eset\right)=t$, is fulfilled for most error matrices with $t$ non-zero columns if the interleaving order $\ell$ is large enough.
In the following subsections, we will see that Theorem~\ref{thm:MK_generalization} is indeed a significant improvement over \cite[Theorem~2]{metzner1990general} since there are PMDS with only a few error positions $\Eset$ that are not $t+1$-independent for $d-2< t < n-k$.

\printalgoIEEE{
\DontPrintSemicolon
\KwIn{Parity-check matrix $\H$, received word $\R = \C + \E$}
\KwOut{Transmitted codeword $\C$}

$\S \gets \H \R^{\top} \in \Fq^{(n-k)\times\ell}$.

Determine $\P\in \Fq^{(n-k)\times(n-k)}$ s.t. $\P \S$ is in reduced row echelon form.

$\zeta \gets$ number of zero rows in $\P \S$.

$\Q \in \Fq^{\zeta \times n} \gets$ last $\zeta$ rows of $\P \H$ 

Determine $\B\in\Fq^{(n-k-\zeta) \times n}$ s.t.\ the columns of $\B$ which correspond to the zero-columns of $\Q$ form an identity matrix and the remaining columns of $\B$ are zero.

Determine $\A\in\Fq^{\ell \times (n-k-\zeta)} $ s.t. $(\H\B^{\top})\A^{\top} = \S$. \label{line:erasure_correction}

$\C \gets \R - \A\B \in \Fq^{\ell \times n}$.

\Return{$\C$}
\caption{Metzner--Kapturowski Algorithm \cite{metzner1990general}}
\label{alg:MK}
}

\subsection{PMDS Codes With Many $(t+1)$-Independent Positions}

A set of erasures $\Eset$ can be corrected if and only if its complement $\bar{\Eset} \coloneqq \{1,\dots,n\} \setminus \Eset$ contains an \emph{information set}, i.e., indexes $k$ linearly independent columns of the generator matrix. 
The authors of \cite{tamo2016optimal} studied a family of optimal LRCs, which in some parameter range are able to correct $n-k$ erasures with probability approaching $1$ for large code lengths. This follows from showing that the number of information sets relative to the number of all sets with $k$ elements tends to $1$ for $n \to \infty$.

We will use a similar approach in the following to show that the relative number of $(t+1)$-independent positions with $t \leq n-k-1$ tends to $1$ for a family of PMDS codes.

\begin{definition}[\!\!\cite{tamo2016optimal}]
Let $\Rset_1,\dots,\Rset_{n/(r+\varrho-1)}$ be the repair sets of an $[n,k,r,\varrho]$ PMDS code. We define the set
\begin{equation*}
\Sset_{\mu} \coloneqq \{ \Scal \subseteq \{1,\dots,n\} \, : \, |\Scal|=\mu , \, |\Scal \cap \mathcal{R}_i| \leq r \, \forall i  \} \ .
\end{equation*}
\end{definition}

The following was shown in \cite{tamo2016optimal} for a special class of PMDS codes and holds in general for PMDS codes.
\begin{lemma}\label{lem:information_sets_Sk}
Let $\G$ be a generator matrix of an $[n,k,r,\varrho]$ PMDS code.
Let $\Scal \subseteq \{1,\dots,n\}$ be of cardinality $k$.
Then, the columns of $\G_\Scal$ (i.e., the columns of $\G$ indexed by $\Scal$) are linearly independent if and only if $\Scal \in \Sset_k$.
\end{lemma}

\begin{proof}
The statement was proven for the codes in \cite{tamo2016optimal} within the proof of \cite[Lemma~7]{tamo2016optimal}. It holds in general for PMDS codes since any set $\Scal \in \Sset_k$ corresponds to $k$ columns of a generator matrix of the MDS code obtained by puncturing the PMDS code at $\varrho-1$ positions of each local repair set not in $\Scal$. This puncturing is possible since $\Scal$ intersects with each local repair set in at most $r$ positions, so there are at least $\varrho-1$ positions left in each repair set. It is well-known that any $k$ columns of an MDS code's generator matrix are linearly independent.
\end{proof}

The following lemma is necessary to relate the sets in $\Sset_{\mu}$ to the $(t+1)$-independent property. We will also use the statement, in Section~\ref{ssec:Pf_finite}, to derive bounds on the number of $(t+1)$-independent sets.

\begin{lemma}\label{lem:t+1-independent_criterion}
Let $\H \in \Fq^{n-k \times n}$ be a parity-check matrix of an $[n,k,r,\varrho]$ PMDS code.
Let $\Eset \subseteq \{1,\dots,n\}$ be a set of cardinality $t$ and $\bar{\Eset} \coloneqq \{1,\dots,n\} \setminus \Eset$ its complement. For each local repair set, indexed by $i \in \{1,\dots,\mu\}$, we define
\begin{align*}
\mathcal{T}_i^\Eset &\coloneqq \bar{\Eset} \cap \mathcal{R}_i, &&\text{(set of error-free positions in each local repair set)} \\
T_i^\Eset &\coloneqq \left|\mathcal{T}_i^\Eset\right|, &&\text{(number of error-free positions in each local repair set)}\\
O_i^\Eset &\coloneqq \max\left\{0, T_i^\Eset-r\right\}, &&\text{(\textbf{excess} $\coloneqq$ number of error-free positions in each local repair set exceeding $r$)}
\end{align*}
Then, $\Eset$ is $(t+1)$-independent if and only if
\begin{align*}
\sum_{i=1}^{\mu} O_i^\Eset \leq \begin{cases}
n-k-t-1, &\text{if $\exists \, j$ : $0 < T_j^\Eset \leq r$,} \\
n-k-t, &\text{else.}
\end{cases}
\end{align*}
(I.e., if the \textbf{overall excess} is small enough)
\end{lemma}

\begin{proof}
We have
\begin{align*}
&\Eset \text{ is $(t+1)$-independent} \\
\Leftrightarrow \quad &\forall \, i \in \bar{\Eset}: \, \rank(\H_{\Eset \cup \{i\}}) = t+1 \\
\Leftrightarrow \quad &\forall \, i \in \bar{\Eset} \, \exists \, J_i \subseteq \{1,\dots,n\}, \, |J_i| = n-k-t-1 \,: \rank (\H_{\Eset \cup \{i\} \cup J_i}) = n-k \\
\overset{(\ast)}{\Leftrightarrow} \quad &\forall \, i \in \bar{\Eset} \, \exists \, S_i \in \Sset_k \, : \, S_i \subseteq \overline{\Eset \cup \{i\}} \\
\Leftrightarrow \quad &\forall \, i \in \bar{\Eset} \, \exists \, S_i \in \Sset_k \, : \, S_i \subseteq \bar{\Eset} \setminus \{i\},
\end{align*}
where $(\ast)$ is due to the fact that $S \in \Sset_k$ if and only if $S$ is an information set, which again holds if and only if the columns of $\H$ indexed by the complementary positions $\bar{S}$ are linearly independent.

Hence, $\Eset$ is $(t+1)$-independent if and only if there is an information set in the complement of $\Eset$ even when we remove an arbitrary element from it.

By the above notation, the $\mathcal{T}_i^\Eset$ form a partition of $\bar{\Eset}$, i.e., $\mathcal{T}_i^\Eset \cap \mathcal{T}_j^\Eset = \emptyset$ for all $i \neq j$ and $\cup_{j=1}^{\mu} \mathcal{T}_j^\Eset = \bar{\Eset}$.
Thus, $\Eset$ is $(t+1)$-independent if and only if after removing any element $i$ from $\cup_{j=1}^{\mu} \mathcal{T}_j^\Eset$, there is still a subset $S_i \subseteq \cup_{i=1}^{\mu} \mathcal{T}_i^\Eset$ with $S_i \in \Sset_k$.

This again is equivalent to saying that if we remove an \emph{arbitrary} element $i$ from $\bigcup_{j=1}^{\mu} \mathcal{T}_j^\Eset$, then there are $n-k-t-1$ elements (\emph{which we can choose}), say $D_i \subseteq \bar{\Eset} \setminus \{i\}$, such that
\begin{align*}
\left(\bigcup_{i=1}^{\mu} \mathcal{T}_i^\Eset \right) \setminus \left(\{i\} \cup D_i\right) \in \Sset_k.
\end{align*}
This is equivalent to
\begin{align*}
\left| \mathcal{T}_i^\Eset \setminus \left[ \left(\{i\} \cup D_i\right) \cap \mathcal{R}_i \right] \right| \leq r \quad \forall \, i,
\end{align*}
i.e., that we can remove enough elements from each local repair set such that there are at most $r$ elements left in each local repair set. The cardinality $O_i^\Eset$ defined above corresponds to the excess of $\bar{\Eset}$ in each local repair set, i.e., the number of elements in $\bar{\Eset} \cap \mathcal{R}_i$ that exceeds $r$. In general, we thus need to keep the sum of the $O_i^\Eset$ small as follows:

If there is a $j$ such that $0<T_j^\Eset\leq r$, then $i$ could be from this local repair set $j \in \mathcal{T}_j^\Eset$, so the overall excess after removing $i$ is exactly $\sum_{i=1}^{\mu} O_i^\Eset$, which must not exceed $n-k-t-1$ (note that by removing an element from $\mathcal{T}_j^\Eset$, we get $T_j^\Eset \gets T_j^\Eset-1 <r$, so $O_j^\Eset = 0$ is unchanged).
Otherwise, all $T_j^\Eset$ are either $0$ (we cannot remove an element from this repair set) or $>r$ (in this case, $O_i^\Eset$ is reduced after removing an element). Hence, after removing $i$, the overall excess is $(\sum_{i=1}^{\mu} O_i^\Eset)-1$, so we get the condition $\sum_{i=1}^{\mu} O_i^\Eset \leq n-k-t$.
\end{proof}

From Lemma~\ref{lem:t+1-independent_criterion}, we get the following sufficient condition for a set to be $(t+1)$-independent.

\begin{lemma}\label{lem:t+1-independent_S_k+1}
Let $\H \in \Fq^{(n-k) \times n}$ be a parity-check matrix of an $[n,k,r,\varrho]$ PMDS code.
Further, let $\Eset$ be a set of cardinality $|\Eset| =: t \leq n-k-1$.
If there is a set $\Scal \in \Sset_{k+1}$ that is a subset of the complement of $\Eset$, i.e., $\Scal \subseteq \bar{\Eset} \coloneqq \{1,\dots,n\} \setminus \Eset$, then $\Eset$ is $(t+1)$-independent.
\end{lemma}

\begin{proof}
By assumption, there is a set $\mathcal{J} \subset \bar{\Eset}$ such that $\bar{\Eset} = \Scal \cup \mathcal{J}$ and $\Scal \cap \mathcal{J} = \emptyset$.
Using the notation of Lemma~\ref{lem:t+1-independent_criterion}, we have
\begin{align*}
T_i^{\Eset} = \left|\bar{\Eset} \cap \mathcal{R}_i\right| = \left|\Scal \cap \mathcal{R}_i\right| + \left|\mathcal{J} \cap \mathcal{R}_i\right| \overset{\Scal \in \Sset_{k+1}}{\leq} r + \left|\mathcal{J} \cap \mathcal{R}_i\right|,
\end{align*}
i.e., $O_i^\Eset \leq \left|\mathcal{J} \cap \mathcal{R}_i\right|$, and thus
\begin{align*}
\sum_{i=1}^{\mu} O_i^\Eset \leq \sum_{i=1}^{\mu} \left|\mathcal{J} \cap \mathcal{R}_i\right| = |\mathcal{J}| = \left|\bar{\Eset}\right|-|\Scal| = (n-t)-(k+1) = n-k-t-1.
\end{align*}
By Lemma~\ref{lem:t+1-independent_criterion}, the claim follows.
\end{proof}

\begin{remark}
For $t=n-k-1$ and $r<k$ (note that the case of $r=k$ gives an MDS code without locality), Lemma~\ref{lem:t+1-independent_S_k+1} becomes an if and only if statement, i.e., $\Eset$ is $(t+1)$-independent if and only if $\bar{\Eset} \in \Sset_{k+1}$.

The proof is as follows. Note first that $|\bar{\Eset}|=k+1$. If $\Eset$ is $(t+1)$-independent and $t = n-k-1$, then we must have $\sum_{i} O_i^\Eset = 0$ if there is a $j$ with $0 < T_j^\Eset \leq r$ and $\sum_{i} O_i^\Eset \leq 1$ otherwise.
From $\sum_{i} O_i^\Eset = 0$, it directly follows that $T_j^\Eset \leq r$ for all $j$, which implies $\bar{\Eset} \in \Sset_{k+1}$.
If $\sum_{i} O_i^\Eset = 1$, we have $T_j^\Eset = 0$ except for one $j$ which fulfills $T_j^\Eset=r+1$. Since $|\bar{\Eset}|=\sum_j T_j^\Eset = k+1$, we have $r=k$, a contradiction. The claim follows.
\end{remark}

Due to Lemma~\ref{lem:t+1-independent_S_k+1}, the relative amount of $(t+1)$-independent positions can be lower-bounded using the set $\Sset_{k+1}$ as follows.

\begin{lemma}
Let $t \leq n-k-1$. Then,
\begin{equation*}
\frac{|\{ \Eset \subseteq \{1,\dots,n\} \, : \, |\Eset| = t, \, \text{$(t+1)$-independent} \}|}{\binom{n}{t}} \geq \frac{|\Sset_{k+1}|}{\binom{n}{k+1}}.
\end{equation*}
\end{lemma}

\begin{proof}
  A set is $(t+1)$-independent if and only if its complement contains an element of $\Sset_{k+1}$. Consider the bipartite graph with the elements of $\Sset_{k+1}$ on the left and all $(t+1)$-independent sets $\Eset$ of cardinality $t$ on the right. We draw an edge between any $\Scal \in \Sset_{k+1}$ and $\Eset$ if $\Scal \subseteq \{1,...,n\} \setminus \Eset$. Clearly the degree of any vertex on the left is exactly $\binom{n-(k+1)}{(n-t)-(k+1)}$, as we can add $|\{1,...,n\} \setminus \Eset| - |\Scal| = (n-t)-(k+1)$ elements to each set $\Scal$ to obtain a set $\Eset$ and there are $n-|\Scal| = n-(k+1)$ elements that are not already in $\Scal$. The degree of the vertices on the right is upper bounded by $\binom{n-t}{k+1}$, the number of subsets of $\{1,...,n\} \setminus \Eset$ of cardinality $k+1$. This yields
  \begin{align*}
    \binom{n-t}{k+1} |\{ \Eset \subseteq \{1,\dots,n\} \, : \, |\Eset| = t, \, \text{$(t+1)$-independent} \}| &\geq \binom{n-(k+1)}{(n-t)-(k+1)} |\Sset_{k+1}| &\\
    \frac{(n-t)!}{((n-t)-(k+1))! \cdot (k+1)!} |\{ \Eset \subseteq \{1,\dots,n\} \, : \, |\Eset| = t, \, \text{$(t+1)$-independent} \}| &\geq \frac{(n-(k+1))!}{(n-(k+1)-t)! \cdot t!}|\Sset_{k+1}| \\
    \frac{|\{ \Eset \subseteq \{1,\dots,n\} \, : \, |\Eset| = t, \, \text{$(t+1)$-independent} \}|}{\binom{n}{t}}&\geq\frac{|\Sset_{k+1}|}{\binom{n}{k+1}} \ ,
  \end{align*}
  where equality holds for $t=n-k-1$.
\end{proof}

The following theorem is a generalization of \cite[Theorem~3]{tamo2016optimal}, which lower-bounds $\nicefrac{\Sset_{k}}{\binom{n}{k}}$ for the special case $\varrho=2$.

\begin{lemma}\label{lem:Partial_MDS_bound_S_k+1}
Let $\Code$ be an $[n,k,r,\varrho]$ PMDS code.
Then,
\begin{align}
\frac{|\Sset_{k+1}|}{\binom{n}{k+1}} \geq 1-2^{\log_2(n) - (r+1) \log_2\left( \binom{r+\varrho-1}{\xi}^{-\frac{1}{r+1}} \frac{n}{k+1}\right)}, \label{eq:bound_on_S_k+1}
\end{align}
where $\xi \coloneqq \min\left\{\varrho-2,\left\lfloor\frac{r+\varrho-1}{2}\right\rfloor\right\}$.
\end{lemma}

\begin{proof}
We have
\begin{align}
&\binom{n}{k+1}-|\Sset_{k+1}| = |\overline{\Sset_{k+1}}| \notag \\
&= \left| \left\{ \Scal \subseteq \{1,\dots,n\} \, : \, |\Scal|=k+1 , \, \exists i \, : \, |\Scal \cap \mathcal{R}_i| > r \right\} \right| \notag \\
&\leq \sum_{i=1}^{\numgroups} \left| \left\{ \Scal \subseteq \{1,\dots,n\} \, : \, |\Scal|=k+1 , \, \, |\Scal \cap \mathcal{R}_i| > r \right\} \right| \notag \\
&\leq \mu \sum_{j=r+1}^{r+\varrho-1} \binom{r+\varrho-1}{j} \binom{n-(r+\varrho-1)}{k+1-j} \label{eq:P_t=n-k-1_union_bound_step} \\
&\overset{(\ast)}{\leq} \mu \left(\frac{k+1}{n}\right)^{r+1}\binom{n}{k+1} \sum_{j=r+1}^{r+\varrho-1} \underbrace{\binom{r+\varrho-1}{j}}_{\leq \binom{r+\varrho-1}{\xi}} \notag \\
&\leq \mu (\varrho-1) \left(\frac{k+1}{n}\right)^{r+1}\binom{n}{k+1} \binom{r+\varrho-1}{\xi},
\end{align}
where $(\ast)$ follows from (recall that $j \in \{r+1,\dots,r+\varrho-1 \}$) 
\begin{align*}
\binom{n-(r+\varrho-1)}{k+1-j} &\leq \binom{n-j}{k+1-j} \\
&= \binom{n}{k+1} \prod_{i=0}^{j-1} \frac{k+1-i}{n-i} \\
&\leq \binom{n}{k+1} \prod_{i=0}^{j-1} \frac{k+1}{n} \\
&\leq \left(\frac{k+1}{n}\right)^j \binom{n}{k+1}.
\end{align*}
Hence, we have
\begin{align*}
\frac{|\Sset_{k+1}|}{\binom{n}{k+1}} &\geq 1-\mu (\varrho-1) \left(\binom{r+\varrho-1}{\xi}^{\frac{1}{r+1}} \frac{k+1}{n}\right)^{r+1} \\
  &\geq 1-n \left(\binom{r+\varrho-1}{\xi}^{\frac{1}{r+1}} \frac{k+1}{n}\right)^{r+1} \\
&= 1-2^{\log_2(n) - (r+1) \log_2\left( \binom{r+\varrho-1}{\xi}^{-\frac{1}{r+1}} \frac{n}{k+1}\right)} \ ,
\end{align*}
which proves the claim.
\end{proof}

Note that for $\varrho \leq r+2$, we always have $\xi = \varrho-2$.

Using the bound in Lemma~\ref{lem:Partial_MDS_bound_S_k+1}, we are able to formulate conditions on the local and global distance of a family of PMDS codes for which the relative size of $\Sset_{k+1}$ compared to all cardinality-$(k+1)$ subsets of $\{1,\dots,n\}$ approaches $1$ for growing code length.

\begin{lemma}\label{lem:Partial_MDS_bound_S_k+1_asymptotical}
Let $\{\Code_n\}$ be a family of $[n,k_n,r_n,\varrho_n]$ PMDS LRC with
\begin{align}
\binom{r_n+\varrho_n-1}{\xi_n}^{-\frac{1}{r_n+1}}  &> C_1 \frac{k_n+1}{n}\label{eq:rate_condition}\\
    r_n+1 &\geq \frac{C_2 \log_2(n)}{\log_2(C_1)} \label{eq:number_of_local_groups_condition}
\end{align}
for some $C_1,C_2>1$, where $\xi_n \coloneqq \min\left\{\varrho_n-2,\left\lfloor\frac{r_n+\varrho_n-1}{2}\right\rfloor\right\}$.
Then,
\begin{align*}
\frac{|\Sset_{k_n+1}|}{\binom{n}{k_n+1}} \to 1 \quad (n \to \infty).
\end{align*}
\end{lemma}

\begin{proof}
It is easy to see that the exponent of $2$ in the bound \eqref{eq:bound_on_S_k+1} converges to minus infinity under the given conditions.
\end{proof}

\begin{remark}
Condition~\eqref{eq:number_of_local_groups_condition} puts a rate constraint on the code. However, if $r_n$ grows faster to infinity than $\varrho_n$, the following argument shows that we can choose arbitrary rates.
We study the asymptotic behavior of $\binom{r_n+\varrho_n-1}{\xi_n}^{-\frac{1}{r_n+1}}$ for $r_n \in \omega(\varrho_n)$ (i.e., $r_n$ grows asymptotically strictly faster than $\varrho_n$):
\begin{align*}
\binom{r_n+\varrho_n-1}{\xi_n}^{-\frac{1}{r_n+1}} &= \frac{1}{\binom{r_n+\varrho_n-1}{\varrho_n-2}^{\frac{1}{r_n+1}}} \\
&\geq \frac{1}{\left(\frac{e(r_n+\varrho_n-1)}{\varrho_n-2}\right)^{\frac{\varrho_n-2}{r_n+1}}} \\
&=\frac{1}{\underbrace{e^{\frac{\varrho_n-2}{r_n+1}}}_{\to \, 1} \cdot \underbrace{\left(1+\frac{r_n+1}{\varrho_n-2}\right)^{\frac{\varrho_n-2}{r+1}}}_{\to \, 1} } \to 1.
\end{align*}
Note that we use that if $r_n$ grows faster than $\varrho_n$, at some point we have $\xi_n = \varrho_n-2$.
\end{remark}

\begin{remark}
By a similar argument as above, Condition~\eqref{eq:rate_condition} in Lemma~\ref{lem:Partial_MDS_bound_S_k+1_asymptotical} can be replaced by
\begin{equation*}
 \left(\frac{R_\mathsf{local}}{e}\right)^{\varrho-2} \approx \left(\frac{r+1}{e(r+\varrho-1)}\right)^{\varrho-2} > C_1 \frac{k+1}{n} \approx C_1 R_\mathsf{global}
\end{equation*}
\end{remark}

Lemma~\ref{lem:Partial_MDS_bound_S_k+1_asymptotical} implies that under the given conditions, asymptotically almost any set of $k+1$ indices is in $\Sset_{k+1}$, and thus, almost any set of $t \leq n-k-1$ error positions is $(t+1)$-independent.

\begin{remark}
Since any cardinality-$k$ subset of $\Sset_{k+1}$ is an information set, Lemma~\ref{lem:Partial_MDS_bound_S_k+1_asymptotical} also implies that the codes satisfying Conditions \eqref{eq:rate_condition} and \eqref{eq:number_of_local_groups_condition}, and can correct almost all $n-k$ erasures asymptotically. This constitutes a generalization of the statement in \cite[Theorem~3]{tamo2016optimal}, which proves the special case $\varrho=2$.
\end{remark}

\subsection{Decoding PMDS Codes Beyond the Minimum Distance}

Using Lemma~\ref{lem:Partial_MDS_bound_S_k+1_asymptotical}, we can state the following explicit class of PMDS codes correcting almost any error up to weight $n-k-1$.

\begin{theorem}\label{thm:PMDS_family_correcting_n-k-1}
Let $\{\Code_n\}$ be a family of $[n,k_n,r_n,\varrho_n]$ PMDS LRC over a field $q_n$, where $q_n \to \infty$ for $n \to \infty$ and the parameters $r_n, \varrho_n$ fulfill
conditions \eqref{eq:rate_condition} and \eqref{eq:number_of_local_groups_condition} of Lemma~\ref{lem:Partial_MDS_bound_S_k+1_asymptotical} for fixed constants $C_1,C_2>1$.
Furthermore, let $\ell_n = n-k_n-1$. Then, the family $\{\Code_n'\}$ of $[n,k_n,r_n,\varrho_n]$ codes over the fields of size $q_n^{\ell_n}$ obtained by interpreting the $\ell_n$-interleaved codes of $\Code_n$ as linear codes over the large field $\mathbb{F}_{q_n^{\ell_n}}$, fulfill the following properties:
\begin{itemize}
\item the codes $\Code_n'$ are PMDS,
\item $\Code_n'$ corrects up to $n-k_n-1$ errors with probability approaching $1$ for $n \to \infty$ (assuming uniformly distributed errors of given weight), and
\item the decoding complexity is $O(n^3)$ operations over $\mathbb{F}_{q_n}$.
\end{itemize}
\end{theorem}

\begin{proof}
A (homogeneous) interleaved code is a linear code over the large field of the same parameters as the constituent code. Since puncturing the interleaved code corresponds to puncturing the constituent codes, the definition of PMDS codes directly implies that the interleaved code is also PMDS.

For showing the correction capability, first note that interpreting elements of $\mathbb{F}_{q_n^{\ell_n}}$ as vectors in $\mathbb{F}_{q_n}^{\ell_n}$ gives a bijective mapping between all Hamming errors of weight $t$ in $\mathbb{F}_{q_n^{\ell_n}}^n$ and all burst errors in $\mathbb{F}_{q_n}^{\ell_n \times n}$ of weight $t$.
We use Theorem~\ref{thm:MK_generalization}, Lemma~\ref{lem:Partial_MDS_bound_S_k+1_asymptotical}, and the fact that the fraction of $\ell_n \times t_n$ matrices over the field of size $q_n$ of rank $t_n$ is at least $1-\nicefrac{4}{q_n}$ for $q_n\geq 4$, cf.~\cite[Lemma~3.13]{Overbeck_Diss_InterleveadGab}.
The probability that a random error pattern of weight $t<n-k$ cannot be corrected is therefore
\begin{align*}
&\mathrm{P}(\E \text{ cannot be corrected}) \\
&=\mathrm{P}\left(\Scal \not\subseteq \overline{\supp(\E)} \, \forall \Scal \in \Sset_{k_n+1} \, \lor \,  \rank_{\Fq}(\E)<t_n \right) \\
&= \mathrm{P}\left(\Scal \not\subseteq \overline{\supp(\E)} \, \forall \Scal \in \Sset_{k_n+1}\right) + \underbrace{\mathrm{P} \left(\rank_{\Fq}(\E)<t_n \right)}_{< \, 4/q_n} \\
&< 1-\frac{|\Sset_{k_n+1}|}{\binom{n}{k_n+1}} + \frac{4}{q_n} \to 0 \quad (n \to \infty)
\end{align*}
since $q_n \to \infty$ for $n \to \infty$.

As for the complexity, we apply the Metzner--Kapturowski algorithm on the received matrix in $\mathbb{F}_{q_n}^{\ell_n \times n}$, which runs in complexity $O(n^3)$ over the small field $\mathbb{F}_{q_n}$ since $\ell_n \leq n$.
\end{proof}

\begin{remark}
If the assumption $q_n \to \infty$ for $n \to \infty$ in Theorem~\ref{thm:PMDS_family_correcting_n-k-1} is not fulfilled for a class of codes, this would disprove the MDS conjecture.
\end{remark}

In the following we give the overall field size $Q_n$ for some families of PMDS codes.

\begin{corollary}
Let the family $\{\Code_n\}$ be a subset of the code class in \cite{tamo2016optimal}. Then, the field size is given by
\begin{align*}
\log Q_n &\in O\left( n^{2+\log(\log(n))}\log(n) \right).
\end{align*}
and the overall decoding complexity in bit operations is
\begin{align*}
O^\sim\!\left( n^{4+\log(\log(n))} \right),
\end{align*}
where $O^\sim$ neglects logarithmic terms in $n$.
\end{corollary}

\begin{proof}
The field size of the codes in \cite{tamo2016optimal} is
\begin{align*}
q_n &\in O \left( \left( \frac{nr_n}{r_n+\varrho_n-1} \right)^{(r_n-1)r_n^{r_n+\varrho_n-1}k_n+1} \right) \\
&\subseteq O\left(n^{n [\log(n)]^{\log(n)}}\right)
=O\left( e^{n^{1+\log(\log(n))}\log(n)} \right).
\end{align*}
Hence, we must choose an $\ell_n$-interleaved code of the above mentioned code, where $\ell_n=n-k_n-1$. Then, the overall field size is $Q_n = q^{\ell_n}$, i.e.,
\begin{equation*}
\log Q \in O\left( n^{2+\log(\log(n))}\log(n) \right). 
\end{equation*}
Using the bases described in \cite{couveignes2009elliptic}, field operations in $\mathbb{F}_{q_n}$ cost $O^\sim(\log(q_n))$ bit operations. Hence, the overall complexity is given by
\begin{equation*}
O^\sim\!\left( n^3 \log(q_n) \right) \subseteq O^\sim\!\left( n^{4+\log(\log(n))} \right). \qedhere
\end{equation*}
\end{proof}

\begin{corollary}
Let the family $\{\Code_n\}$ be a subset of the code class in \cite{gabrys2018constructions}. Then, the field size is given by
\begin{align*}
\log Q_n &\in O \left( n^3 log(log(n)) \right).
\end{align*}
and the overall decoding complexity is $O^\sim(n^5)$ bit operations.
\end{corollary}
For the special case of $\varrho = 2$ the probability of successful decoding can be stated exactly.
\begin{corollary} \label{col:PsucExact}
  The probability of successfully decoding $t$ errors in an $[n,k,r,\varrho=2]$ PMDS code is given by
  \begin{equation*}
    P_{\mathsf{suc}} =  P\{\rank (E) = t\} -\frac{|\Sset_{k+1}|}{\binom{n}{k+1}} \ ,
  \end{equation*}
  where, as shown in \cite[Proof of Theorem~3]{tamo2016optimal},
  \begin{equation*}
    |\Sset_{k+1}| = \sum_{j=1}^{\left\lfloor \frac{k+1}{r+1} \right\rfloor} (-1)^{j-1}\binom{n/(r+1)}{j} \binom{n-j(r+1)}{k+1-j(r+1)}
  \end{equation*}
  and the fraction of full rank matrices $E \in \Fq^{\ell \times t}$ \cite{migler2004} is
  \begin{equation*}
    P\{\rank(E) = t\} = q^{-t \ell} \prod_{j=0}^{t-1} (q^\ell-q^j) \ .
  \end{equation*}

\end{corollary}

The following example shows that the success probability is reasonably close to $1$ even for small parameters. In Section~\ref{ssec:Pf_finite} we will explore the probability of failure in more detail.

\begin{example}
  Consider the PMDS code as defined in \cite{tamo2016optimal} with parameters $[n=15,k=8,r=4,\varrho=2]$ over the field $\Fq$ with $q=16^{k+1} = 2^{36}$. The code is of distance $d=7$, fulfilling the bound (\ref{eq:boundDistanceLRC}) on the distance of an LRC. The unique decoding radius of this code is $t=\left\lfloor \frac{d-1}{2} \right\rfloor = 3$. Given a full rank error matrix, the decoder introduced in \cite{metzner1990general} guarantees decoding of up to $t=d-2=5$ errors. In the case of $t=n-k-1=6$ errors, the error matrix is of full rank with probability $>1-10^{-10}$ and Corollary~\ref{col:PsucExact} gives the probability of success as $P_{\mathsf{dec}} \approx \frac{125}{143}$.
\end{example}

\subsection{Failure Probability for Finite Parameters}\label{ssec:Pf_finite}

In the last subsection, we have seen that there are families of PMDS codes for which the probability of successful decoding approaches $1$ for the code length going to infinity. Our central tool was to show that the relative number of $(t+1)$-independent sets among all cardinality-$t$ subsets of $\{1,\dots,n\}$ goes to $1$ for these code families. For this, we derived a lower bound on the relative number of these sets in Lemma~\ref{lem:Partial_MDS_bound_S_k+1}, which approaches $1$ for $n \to \infty$.

However, for finite $n$, the bound is not necessarily tight and usually orders of magnitude away from the actual probability for $d-2 < t < n-k-1$ (recall that the Metzner--Kapturowski algorithm always works for $t \leq d-2$).
In the following, we show how to efficiently compute the exact probability of drawing a $(t+1)$-independent set for finite parameters.
The algorithm uses dynamic programming to efficiently compute the expression in the following theorem recursively.
The resulting complexity in bit operations is cubic in the code length and linear in the number of errors.

\begin{theorem}\label{thm:exact_failure_prob}
Let $n$, $k$, $t$, $r$, $\varrho$, $n_l = r+\varrho-1$, $\mu = n/(r+\varrho-1)$ (number of local repair sets) be given.
For $\eta\geq 1$, $\tau \geq 0$, $\sigma \geq 0$, $\beta \in \{0,1\}$ define
\begin{align}
\mathcal{W}(\eta,\tau,\sigma,\beta)\! &\coloneqq \! \left\{\! \omegaVec\! \in \! \{0,\dots,n_l\}^\eta : \sum_{i} \omega_i = \tau,  \,  \sigma + \sum_{i=1}^{\eta} \max\{0,\omega_i-r\} \!> \!\begin{cases}
n-k-t-1, &\text{if $\beta=1$ or $\exists \, i$ : $0 < \omega_i \leq r$,} \\
n-k-t, &\text{else.}
\end{cases}
\right\}, \notag\\
\mathsf{W}(\eta,\tau,\sigma,\beta)\! &\coloneqq \! \sum_{\omegaVec \in \mathcal{W}(\eta,\tau,\sigma,\beta)} \prod_{i=1}^{\eta} \binom{n_l}{\omega_i}. \label{eq:W_value}
\end{align}
For an error set $\Eset$ that is chosen uniformly at random from the cardinality-$t$ subsets of $\{1,\dots,n\}$, the probability that it is not $(t+1)$-independent is
\begin{align}
\Pr\left( \Eset \text{ is \textbf{not} $(t+1)$-independent} \right) = \frac{\mathsf{W}(\mu,n-t,0,0)}{\binom{n}{t}}. \label{eq:exact_failure_probability}
\end{align}
\end{theorem}

\begin{proof}
The elements $\omegaVec \in \mathcal{W}(\mu,n-t,0,0)$ correspond exactly to all distributions of the error-free positions to the local repair sets such that the set $\Eset$ is not $(t+1)$-independent (cf.~Lemma~\ref{lem:t+1-independent_criterion}). Thus, we just need to count how many errors with this distribution there are, which is given by the product $\prod_{i=1}^{\mu} \binom{r+\varrho-1}{\omega_i}$. The statement follows directly by the fact that $\Eset$ is uniformly distributed.
\end{proof}

It is not immediately obvious why Theorem~\ref{thm:exact_failure_prob} is formulated in such a complicated way, i.e., why we define $\mathsf{W}(\eta,\tau,\sigma,\beta)$ with four parameters and then use it only for $[\eta,\tau,\sigma,\beta] = [\mu,n-t,0,0]$.
The reason is that we can efficiently recursively compute $\mathsf{W}(\mu,n-t,0,0)$ using the general definition of $\mathsf{W}(\eta,\tau,\sigma,\beta)$ as follows.

\begin{lemma}\label{lem:W_recursion}
Let $n$, $k$, $t$, $r$, $\varrho$, $n_l = r+\varrho-1$, $\mu = n/(r+\varrho-1)$ be given.
Furthermore, let $\mathsf{W}(\eta,\tau,\sigma,\beta)$ be defined as in Theorem~\ref{thm:exact_failure_prob} for $\eta\geq 1$, $\tau \geq 0$, $\sigma \geq 0$, $\beta \in \{0,1\}$.
For $\eta >1$, we have the recursion
\begin{align}
\mathsf{W}(\eta,\tau,\sigma,\beta) &= \sum_{\omega_1=0}^{\min\{\tau,n_l\}} \binom{n_l}{\omega_1} \mathsf{W}\big(\eta-1,\tau-w_1,\sigma+\max\{0,\omega_1-r\},\beta'(\beta,\omega_1)\big), \label{eq:W_eta>1_recursion}
\end{align}
and for $\eta=1$, we have
\begin{align}
\mathsf{W}(1,\tau,\sigma,\beta) = \begin{cases}
\tbinom{n_l}{\tau}, &\text{if } \tau\leq n_l  \quad \land  \quad \sigma+\max\{0,\tau-r\} > n-k-t-\beta'(\beta,\tau), \\
0, &\text{else.}
\end{cases}\label{eq:W_eta=1_recursion}
\end{align}
where
\begin{align*}
\beta'(\beta,\omega_1) &= \begin{cases}
1, &\text{if } \beta=1 \, \lor 0<\omega_1 \leq r, \\
0, &\text{else.}
\end{cases}
\end{align*}
\end{lemma}

\begin{proof}
We start with the recursion \eqref{eq:W_eta>1_recursion}. Note that by definition of $\mathcal{W}$, we have
\begin{align*}
\mathcal{W}(\eta,\tau,\sigma,\beta) &=
\bigcup_{\omega_1=0}^{\min\{\tau,n_l\}} \Big\{[\omega_1,\omega_2,\dots,\omega_\eta] \, : \, [\omega_2,\dots,\omega_\eta] \in \mathcal{W}(\eta-1,\tau-w_1,\sigma+\max\{0,\omega_1-r\},\beta'(\beta,\omega_1))\Big\}.
\end{align*}
This is clear since for a fixed $\omega_1$, we have $\sum_i \omega_i = \tau$ if and only if $\sum_{i=2}^{\eta} \omega_i = \tau-\omega_1$ and we have
\begin{equation*}
\sigma + \sum_{i=1}^{\eta} \max\{0,\omega_i-r\} \!> \!\begin{cases}
n-k-t-1, &\text{if $\beta=1$ or $\exists \, i$ : $0 < \omega_i \leq r$,} \\
n-k-t, &\text{else.}
\end{cases}
\end{equation*}
if and only if (note that ``$\beta=1$ or $0 < \omega_1 \leq r$'' iff ``$\beta'(\beta,\omega_1)=1$'')
\begin{equation*}
\big(\sigma + \max\{0,\omega_i-r\}\big) +\sum_{i=2}^{\eta} \max\{0,\omega_i-r\} > \begin{cases}
n-k-t-1, &\text{if ($\beta=1$ or $0 < \omega_1 \leq r$) or $\exists \, i>1$ : $0 < \omega_i \leq r$,} \\
n-k-t, &\text{else.}
\end{cases}
\end{equation*}
Hence, for $\eta>1$, we have
\begin{align*}
\mathsf{W}(\eta,\tau,\sigma,\beta) &= \sum_{\omega_1=0}^{\min\{\tau,n_l\}} \binom{n_l}{\omega_1} \underbrace{\sum_{\omegaVec' \in \mathcal{W}\big(\eta-1,\tau-w_1,\sigma+\max\{0,\omega_1-r\},\beta'(\beta,\omega_1)\big)} \prod_{i=1}^{\eta-1} \binom{n_l}{\omega_i'}}_{= \mathsf{W}\big(\eta-1,\tau-w_1,\sigma+\max\{0,\omega_1-r\},\beta'(\beta,\omega_1)\big)}.
\end{align*}
For $\eta=1$, we have
\begin{align*}
\mathcal{W}(1,\tau,\sigma,\beta) = \begin{cases}
\Big\{[\tau]\Big\}, &\text{if } \tau\leq n_l  \quad \land  \quad \sigma+\max\{0,\tau-r\} > n-k-t-\beta'(\beta,\tau), \\
\emptyset, &\text{else.}
\end{cases}
\end{align*}
The claim follows immediately by definition of $\mathsf{W}$.
\end{proof}

Lemma~\ref{lem:W_recursion} implies an algorithm for computing $\mathsf{W}$ recursively.
If implemented naively, one could obtain an algorithm that calls itself up to roughly $n_l^{\mu}$ times as in each recursive step, we call for up to $n_l+1$ values of $w_1$ the algorithm recursively, with recursion depth at most $\mu-1$.
In this naive approach, the algorithm would call itself often for the same parameters.
Hence, we can turn it into a dynamic programming routine by memoizing already computed values of $\mathsf{W}(\eta,\tau,\sigma,\beta)$ in a table.
We outline this approach in Algorithm~\ref{alg:W_computation}.
Algorithm~\ref{alg:Pf_computation} computes the probability in \eqref{eq:exact_failure_probability} using this recursive algorithm and we bound its complexity in Theorem~\ref{thm:Pf_alg_correctness_complexity} below.

\printalgoIEEE{
\DontPrintSemicolon
\KwIn{
\begin{itemize}
\item Integers $\eta \geq 0$, $\tau \geq 0$, $\sigma \geq 0$, $\beta \in \{0,1\}$.
\item Global parameters $n$, $k$, $t$, $r$, $\varrho$, $n_l = r+\varrho-1$, $\mu = n/(r+\varrho-1)$. \\
\item Global integer table $\{\mathsf{T}(\eta,\tau,\sigma,\beta)\}_{\eta \geq 0, \tau \geq 0, \sigma \geq 0, \beta \in \{0,1\}}$ with $\mathsf{T}(\eta,\tau,\sigma,\beta)=\mathsf{W}(\eta,\tau,\sigma,\beta)$ if $\mathsf{W}(\eta,\tau,\sigma,\beta)$ has already been computed and $\mathsf{T}(\eta,\tau,\sigma,\beta)=-1$ otherwise.
\end{itemize}}
\KwOut{Integer $\mathsf{W}(\eta,\tau,\sigma,\beta)$ as in \eqref{eq:W_value}.}

\If{$\mathsf{T}(\eta,\tau,\sigma,\beta) \neq -1$}{
	\Return{$\mathsf{T}(\eta,\tau,\sigma,\beta)$}
} \Else{
	\If{$\eta=1$ \label{line:W_comp_start}}{
		\If{$\beta=1$ or $0<\tau \leq r$}{
			$\beta' \gets 1$
		} \Else{
			$\beta' \gets 0$
		}
		\If{$\tau \leq n_l$ and $\sigma+\max\{0,\tau-r\} > n-k-t-\beta'$}{
			$\res \gets \tbinom{n_l}{\tau}$
		} \Else{
			$\res \gets 0$
		}
	} \Else{
		$\res \gets 0$ \\
		\For{$w_1=0,\dots,\min\{\tau,n_l\}$}{
			\If{$\beta=1$ or $0<\omega_1 \leq r$}{
				$\beta' \gets 1$
			} \Else{
				$\beta' \gets 0$
			}
			$\res = \res + \tbinom{n_l}{\omega_1} \mathsf{W}\big(\eta-1,\tau-w_1,\sigma+\max\{0,\omega_1-r\},\beta'\big)$ \\
		}
	}
	$\mathsf{T}(\eta,\tau,\sigma,\beta) \gets \res$ \\
	\Return{$\res$ \label{line:W_comp_stop}}
}
\caption{$\mathsf{W}(\eta,\tau,\sigma,\beta)$}
\label{alg:W_computation}
}

\printalgoIEEE{
\DontPrintSemicolon
\KwIn{Global parameters $n$, $k$, $0\leq t\leq n$, $r$, $\varrho$, $n_l = r+\varrho-1$, $\mu = n/(r+\varrho-1)$}
\KwOut{Probability $\Pr\left( \Eset \text{ is \textbf{not} $(t+1)$-independent} \right)$ for $\Eset$ chosen uniformly at random from the subsets of $\{1,\dots,n\}$ of cardinality $t$}
$\{\mathsf{T}(\eta,\tau,\sigma,\beta)\}_{1 \leq \eta \leq \mu, \, 0 \leq \tau \leq n-t, \, 0 \leq \sigma \leq (\mu-1)(\varrho-1), \, 0 \leq \beta\leq 1} \gets $ global integer table filled with $-1$ initially \\
\Return{$\frac{\mathsf{W}(\mu,n-t,0,0)}{\binom{n}{t}}$}
\caption{Compute Probability $\Pr\left( \Eset \text{ is \textbf{not} $(t+1)$-independent} \right)$}
\label{alg:Pf_computation}
}

\begin{theorem}\label{thm:Pf_alg_correctness_complexity}
Algorithm~\ref{alg:Pf_computation} is correct and can be implemented with complexity
\begin{align*}
O^\sim\!\left( n^3 t \right)
\end{align*}
in bit operations, where $O^\sim$ is the ``soft-O'' notation, i.e., it neglects logarithmic factors in $n$ and $t$.
\end{theorem}

\begin{proof}
Correctness of the algorithm immediately follows from Lemma~\ref{lem:W_recursion}.

For the complexity, we need to analyze how often Lines~\ref{line:W_comp_start}--\ref{line:W_comp_stop} of Algorithm~\ref{alg:W_computation} are called.
This corresponds to exactly the number of entries of the table $\{\mathsf{T}(\eta,\tau,\sigma,\beta)\}$ that need to be computed.
Our claim is that we compute at most the entries with indices $1 \leq \eta \leq \mu$, $0 \leq \tau \leq n-t$, $0 \leq \sigma \leq (\mu-1)(\varrho-1)$, and $0 \leq \beta\leq 1$, i.e., in total at most
\begin{align*}
2\mu(n-t+1)\big[(\mu-1)(\varrho-1)+1\big] \leq 2\mu^2 \varrho n
\end{align*}
table entries.
This is clear for $\eta$ and $\tau$, since they are $\mu$ and $n-t$ initially, and during the recursion, their value is never increased and never drops below $1$ and $0$, respectively.
Also $\beta$ is either $0$ or $1$ in each function call.
For $\sigma$, we see that in each recursion, its value is increased by at most $\max\{0,w_1-r\} \leq r+\varrho-1-r = \varrho-1$. As the maximal recursion depth is $\mu-1$, we get the claimed upper bound.

In each function call, we perform at most $n_l$ integer additions and multiplications (assuming a negligible pre-computation of all the binomial coefficients $\tbinom{n_l}{0},\dots,\tbinom{n_l}{\lfloor n_l/2\rfloor}$). The cost (in bit operations) of each such operation is quasi-linear in the maximal bit size of the involved integers, cf.~\cite{harvey2019faster}.
We can, very roughly, upper-bound the size of the integers by
\begin{align*}
\leq \tbinom{n}{t} \leq n^t \leq 2^{t \log_2(n)}.
\end{align*}
Hence, the overall complexity is in $O^\sim\!\left( \mu^2 \varrho n n_l t \right) \subseteq O^\sim\!\left( n^3 t\right)$.
\end{proof}

\begin{remark}
In practice, Algorithm~\ref{alg:Pf_computation} often only needs to compute a small fraction of the table entries.
This means that the implemented top-down approach of dynamic programming (also called memoization) is typically faster than a bottom-up approach (i.e., computing the entire table iteratively for $\eta=1,\dots,\mu$).
Furthermore, the upper bound on the integer bit size in the proof of Theorem~\ref{thm:Pf_alg_correctness_complexity} is for most (or all) multiplications orders of magnitude away from the actual bit size. Hence, Algorithm~\ref{alg:Pf_computation} is in practice much faster than the upper complexity bound in Theorem~\ref{thm:Pf_alg_correctness_complexity} suggests.
\end{remark}

The following examples are computed with Algorithm~\ref{alg:Pf_computation}.

\begin{example}
We consider three example PMDS codes and compute the probability $\Pr\left( \Eset \text{ is \textbf{not} $(t+1)$-independent} \right)$ using Algorithm~\ref{alg:Pf_computation} for several values of $t$ for $d-1 \leq t \leq n-k-1$. Recall that for $t=d-2$, this probability is always $0$ and $t=n-k-1$ is the largest integer for which the probability is less than $1$.

\textbf{Parameter Set 1:} $n=45$, $k=16$ (rate $\approx 0.36$), $r=8$, $\varrho=8$ (local rate $\approx 0.53$, $\mu=3$ local repair sets).
\begin{center}
\begin{tabular}{l|l}
$t$ & $\Pr\left( \Eset \text{ is \textbf{not} $(t+1)$-independent} \right)$ \\
\hline
$>28$ & $1$ \\
$28 = n-k-1$ & $9.87 \cdot 10^{-2}$ \\
$27$ & $3.61\cdot 10^{-2}$ \\
$26$ & $1.10\cdot 10^{-2}$ \\
$25$ & $2.73\cdot 10^{-3}$ \\
$24$ & $5.13\cdot 10^{-4}$ \\
$23$ & $6.55\cdot 10^{-5}$ \\
$22 = d-1$ & $4.27\cdot 10^{-6}$ \\
$<22$ & $0$
\end{tabular}
\end{center}
We see that already for a few errors below the maximal radius $n-k-1$, the probability that the error positions are $(t+1)$-independent is relatively close to $1$, even for such a short code ($n=45$).

\textbf{Parameter Set 2:} $n=70$, $k=24$ (rate $\approx 0.34$), $r=8$, $\varrho=3$ (local rate $0.8$, $\mu=7$ local repair sets).
\begin{center}
\begin{tabular}{l|l}
$t$ & $\Pr\left( \Eset \text{ is \textbf{not} $(t+1)$-independent} \right)$ \\
\hline
$>45$ & $1$ \\
$45 = n-k-1$ & $1.68 \cdot 10^{-3}$ \\ 
$44$ & $9.38\cdot 10^{-5}$ \\ 
$43$ & $1.25\cdot 10^{-8}$ \\ 
$42 = d-1$ & $4.03 \cdot 10^{-10}$ \\ 
$<42$ & $0$
\end{tabular}
\end{center}
For $t=45$, the failure probability bound in \eqref{eq:P_t=n-k-1_union_bound_step} (i.e., after the union bound step) gives $\leq 1.68\cdot 10^{-3}$ and differs from the exact value only from the fifth digit on.

\textbf{Parameter Set 3:} $n=196$, $k=156$ (rate $\approx 0.80$), $r=26$, $\varrho=3$ (local rate $\approx 0.93$, $\mu=7$ local repair sets).
\begin{center}
\begin{tabular}{l|l}
$t$ & $\Pr\left( \Eset \text{ is \textbf{not} $(t+1)$-independent} \right)$  \\
\hline
$>39$ & $1$ \\
$39 = n-k-1$ & $7.62\cdot 10^{-2}$ \\
$38$ & $1.11 \cdot 10^{-2}$ \\
$37$ & $3.49 \cdot 10^{-4}$ \\
$36$ & $2.71 \cdot 10^{-5}$ \\
$35$ & $2.76 \cdot 10^{-7}$\\
$34$ & $1.50 \cdot 10^{-8}$\\
$33$ & $2.13 \cdot 10^{-11}$\\
$32$ & $9.31 \cdot 10^{-13}$\\
$31$ & $1.73 \cdot 10^{-17}$\\
$30 = d-1$ & $6.56 \cdot 10^{-19}$\\
$<30$ & $0$
\end{tabular}
\end{center}
For $t=39$, the failure probability bound in \eqref{eq:P_t=n-k-1_union_bound_step} gives $\leq 7.71 \cdot 10^{-2}$ and differs from the exact value on the second digit.
\end{example}

\section{Conclusion} \label{sec:conclusion}

In this work we derived a new list-decoding radius for locally repairable codes and gave an explicit algorithm that achieves it. The complexity and the list size are polynomial in the code length $n$ when considering scaling that preserves the number of local repair sets. The asymptotic behavior has been analyzed and a simple probabilistic unique decoder has been introduced.

Further, we considered interleaved decoding of LRCs and PMDS codes, showing that it increases the error tolerance of a storage system. We combined the approach used for increasing the list-decoding radius for improving the decoding radius of LRCs by local decoding with an interleaved decoder.
Further, we proved that the decoding radius of interleaved PMDS codes can be increased by the Metzner-Kapturowski decoder \cite{metzner1990general} beyond their minimum distance, with probability of successful decoding going to $1$ as the code length goes to infinity.

As future work, the list decoding algorithm and bound on the list size given in Section~\ref{sec:listDecoding} should be further studied with the goal to reduce the worst-case list size/complexity (cf. \emph{1)} in Remark~\ref{rem:list_size_bound_not_tight}). Moving from list decoding to probabilistic decoding, an interesting problem is the analysis and algorithmic exploitation of intrinsic side information, e.g., obtained from the local list size or from the distance of local codewords to the received word, to improve the success probability of probabilistic unique decoding (cf. \emph{2)} in Remark~\ref{rem:list_size_bound_not_tight}).
Further, list decoding by combining the local lists through list recovery, similar to the approach taken in \cite{zeh2016improved}, is a promising extension.

Another open problem is the application of results from Section~\ref{sec:PMDS} to PMDS codes under a weaker definition, as considered in \cite{calis2016general}.

Finally, we present in Figure~\ref{fig:rate_tuples} an illustration of tuples of global and local rate for which the new decoders are ``suitable'' (exact definition: see figure caption). It can be seen that for the majority of code parameters, at least one of the decoders is ``suitable''.

\begin{figure}
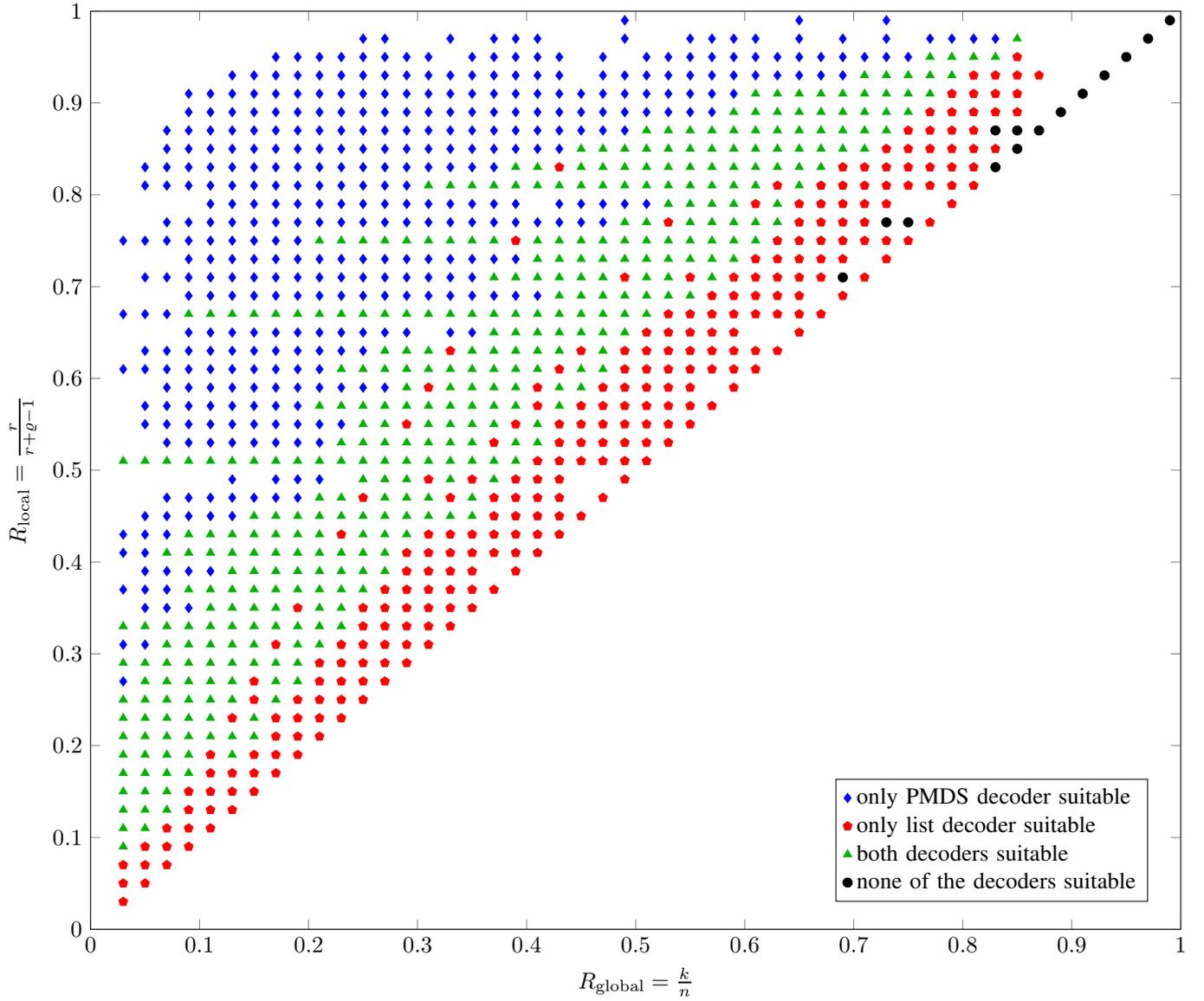

\begin{center}
\begin{tikzpicture}
\pgfplotsset{compat = 1.3}
\begin{axis}[
	width = \columnwidth,
	xlabel = {$R_\mathrm{global} = \frac{k}{n}$},
	ylabel = {$R_\mathrm{local} = \frac{r}{r+\varrho-1}$},
	xmin = 0,
	xmax = 1.0,
	ymin = 0,
	ymax = 1.0,
	legend pos = south east,
	legend cell align=left]

\def\mymark{diamond*}
\input{points_PMDS}
\addlegendentry{only PMDS decoder suitable};

\input{points_list}
\addlegendentry{only list decoder suitable};

\input{points_both}
\addlegendentry{both decoders suitable};

\input{points_none}
\addlegendentry{none of the decoders suitable};


%
%
%
\end{axis}
\end{tikzpicture}
\end{center}
\caption{Rate tuples for which the new list decoder of Section~\ref{sec:listDecoding} can correct more errors than the Johnson radius (event $\mathcal{L}$) and for which the new PMDS decoder of Section~\ref{sec:PMDS} can decode up to $n-k-1$ errors with probability at least $0.9$ (event $\mathcal{P}$).  Legend (\# of points):
\textcolor{red}{$\mathcal{L} \land \neg \mathcal{P}$ (25986)}, \textcolor{mygreen}{$\mathcal{L} \land \mathcal{P}$ (2299)}, \textcolor{blue}{$\neg \mathcal{L} \land \mathcal{P}$ (9806)}, \textcolor{black}{$\neg \mathcal{L} \land \neg \mathcal{P}$ (3922)}.
The points show all rate pairs achievable for optimal LRCs with $r \mid k$ and $(r+\varrho-1) \mid n$ for $100 \leq n \leq 200$. For clarity the plot does not show a point for every set of valid parameters, but instead we divide the plane into squares and a point is colored \textcolor{red}{red} or \textcolor{blue}{blue} if all parameter sets within the square give the respective improvement. We color the point \textcolor{mygreen}{green} if both improvements apply within a square and at least one of the improvements applies to every set of parameters within the square.}
\label{fig:rate_tuples}
\end{figure}

\vspace{4pt}
\bibliographystyle{IEEEtran}
\bibliography{main}

\begin{thebibliography}{10}
\providecommand{\url}[1]{#1}
\csname url@samestyle\endcsname
\providecommand{\newblock}{\relax}
\providecommand{\bibinfo}[2]{#2}
\providecommand{\BIBentrySTDinterwordspacing}{\spaceskip=0pt\relax}
\providecommand{\BIBentryALTinterwordstretchfactor}{4}
\providecommand{\BIBentryALTinterwordspacing}{\spaceskip=\fontdimen2\font plus
\BIBentryALTinterwordstretchfactor\fontdimen3\font minus
  \fontdimen4\font\relax}
\providecommand{\BIBforeignlanguage}[2]{{%
\expandafter\ifx\csname l@#1\endcsname\relax
\typeout{** WARNING: IEEEtran.bst: No hyphenation pattern has been}%
\typeout{** loaded for the language `#1'. Using the pattern for}%
\typeout{** the default language instead.}%
\else
\language=\csname l@#1\endcsname
\fi
#2}}
\providecommand{\BIBdecl}{\relax}
\BIBdecl

\bibitem{Holzbaur2018}
L.~Holzbaur and A.~Wachter-Zeh, ``List decoding of locally repairable codes,''
  in \emph{2018 IEEE International Symposium on Information Theory
  (ISIT)}.\hskip 1em plus 0.5em minus 0.4em\relax IEEE, 2018, pp. 1331--1335.

\bibitem{Holzbaur2019}
L.~Holzbaur, S.~Puchinger, and A.~Wachter-Zeh, ``On error decoding of locally
  repairable and partial mds codes,'' in \emph{2019 IEEE Information Theory
  Workshop (ITW)}.\hskip 1em plus 0.5em minus 0.4em\relax IEEE, 2019.

\bibitem{Dimakis2010}
\BIBentryALTinterwordspacing
A.~G. Dimakis, P.~B. Godfrey, Y.~Wu, M.~J. Wainwright, and K.~Ramchandran,
  ``Network coding for distributed storage systems,'' \emph{IEEE Trans. Inf.
  Theory}, vol.~56, no.~9, pp. 4539--4551, Sep. 2010. [Online]. Available:
  \url{http://dx.doi.org/10.1109/TIT.2010.2054295}
\BIBentrySTDinterwordspacing

\bibitem{Dimakis2011}
A.~G. {Dimakis}, K.~{Ramchandran}, Y.~{Wu}, and C.~{Suh}, ``A survey on network
  codes for distributed storage,'' \emph{Proceedings of the IEEE}, vol.~99,
  no.~3, pp. 476--489, March 2011.

\bibitem{Rashmi2012}
K.~V. Rashmi, N.~B. Shah, K.~Ramchandran, and P.~V. Kumar, ``{Regenerating
  codes for errors and erasures in distributed storage},'' in \emph{2012 IEEE
  International Symposium on Information Theory}.\hskip 1em plus 0.5em minus
  0.4em\relax IEEE, Jul. 2012, pp. 1202--1206.

\bibitem{Tamo2017}
I.~Tamo, M.~Ye, and A.~Barg, ``{Fractional decoding: Error correction from
  partial information},'' \emph{IEEE International Symposium on Information
  Theory - Proceedings}, no. 1030, pp. 998--1002, 2017.

\bibitem{Huang2007}
C.~{Huang}, M.~{Chen}, and J.~{Li}, ``Pyramid codes: Flexible schemes to trade
  space for access efficiency in reliable data storage systems,'' in
  \emph{Sixth IEEE International Symposium on Network Computing and
  Applications (NCA 2007)}, July 2007, pp. 79--86.

\bibitem{Huang2012}
\BIBentryALTinterwordspacing
C.~Huang, H.~Simitci, Y.~Xu, A.~Ogus, B.~Calder, P.~Gopalan, J.~Li, and
  S.~Yekhanin, ``Erasure coding in windows azure storage,'' in
  \emph{Proceedings of the 2012 USENIX Conference on Annual Technical
  Conference}, ser. USENIX ATC'12.\hskip 1em plus 0.5em minus 0.4em\relax
  Berkeley, CA, USA: USENIX Association, 2012, pp. 2--2. [Online]. Available:
  \url{http://dl.acm.org/citation.cfm?id=2342821.2342823}
\BIBentrySTDinterwordspacing

\bibitem{Gopalan2012}
P.~Gopalan, C.~Huang, H.~Simitci, and S.~Yekhanin, ``{On the Locality of
  Codeword Symbols},'' \emph{IEEE Transactions on Information Theory}, vol.~58,
  no.~11, pp. 6925--6934, Nov. 2012.

\bibitem{Sathiamoorthy2013}
\BIBentryALTinterwordspacing
M.~Sathiamoorthy, M.~Asteris, D.~Papailiopoulos, A.~G. Dimakis, R.~Vadali,
  S.~Chen, and D.~Borthakur, ``Xoring elephants: novel erasure codes for big
  data,'' in \emph{Proceedings of the 39th international conference on Very
  Large Data Bases}, ser. PVLDB'13.\hskip 1em plus 0.5em minus 0.4em\relax VLDB
  Endowment, 2013, pp. 325--336. [Online]. Available:
  \url{http://dl.acm.org/citation.cfm?id=2488335.2488339}
\BIBentrySTDinterwordspacing

\bibitem{Kamath2014}
G.~M. Kamath, N.~Prakash, V.~Lalitha, and P.~V. Kumar, ``{Codes With Local
  Regeneration and Erasure Correction},'' \emph{IEEE Transactions on
  Information Theory}, vol.~60, no.~8, pp. 4637--4660, Aug. 2014.

\bibitem{Papailiopoulos2014}
D.~S. {Papailiopoulos} and A.~G. {Dimakis}, ``Locally repairable codes,''
  \emph{IEEE Transactions on Information Theory}, vol.~60, no.~10, pp.
  5843--5855, Oct 2014.

\bibitem{Tamo2014}
I.~Tamo and A.~Barg, ``{A Family of Optimal Locally Recoverable Codes},''
  \emph{IEEE Trans. Inf. Theory}, vol.~60, no.~8, pp. 4661--4676, 2014.

\bibitem{Silberstein2015}
N.~Silberstein, A.~S. Rawat, and S.~Vishwanath, ``{Error-Correcting
  Regenerating and Locally Repairable Codes via Rank-Metric Codes},''
  \emph{IEEE Transactions on Information Theory}, vol.~61, no.~11, pp.
  5765--5778, Nov. 2015.

\bibitem{Chen2007}
M.~{Chen}, C.~{Huang}, and J.~{Li}, ``On the maximally recoverable property for
  multi-protection group codes,'' in \emph{2007 IEEE International Symposium on
  Information Theory}, June 2007, pp. 486--490.

\bibitem{Gopalan2014}
P.~{Gopalan}, C.~{Huang}, B.~{Jenkins}, and S.~{Yekhanin}, ``Explicit maximally
  recoverable codes with locality,'' \emph{IEEE Transactions on Information
  Theory}, vol.~60, no.~9, pp. 5245--5256, Sep. 2014.

\bibitem{gopalan2017}
P.~Gopalan, G.~Hu, S.~Kopparty, S.~Saraf, C.~Wang, and S.~Yekhanin, ``Maximally
  recoverable codes for grid-like topologies,'' in \emph{Proceedings of the
  Twenty-Eighth Annual ACM-SIAM Symposium on Discrete Algorithms}.\hskip 1em
  plus 0.5em minus 0.4em\relax Society for Industrial and Applied Mathematics,
  2017, pp. 2092--2108.

\bibitem{martinez2019universal}
U.~Mart{\'\i}nez-Pe{\~n}as and F.~R. Kschischang, ``Universal and dynamic
  locally repairable codes with maximal recoverability via sum-rank codes,''
  \emph{IEEE Transactions on Information Theory}, 2019.

\bibitem{Blaum2013}
M.~Blaum, J.~L. Hafner, and S.~Hetzler, ``{Partial-MDS Codes and Their
  Application to RAID Type of Architectures},'' \emph{IEEE Transactions on
  Information Theory}, vol.~59, no.~7, pp. 4510--4519, Jul. 2013.

\bibitem{gabrys2018constructions}
R.~Gabrys, E.~Yaakobi, M.~Blaum, and P.~H. Siegel, ``Constructions of partial
  mds codes over small fields,'' \emph{IEEE Transactions on Information
  Theory}, 2018.

\bibitem{Blaum2016}
M.~{Blaum}, J.~S. {Plank}, M.~{Schwartz}, and E.~{Yaakobi}, ``Construction of
  partial mds and sector-disk codes with two global parity symbols,''
  \emph{IEEE Transactions on Information Theory}, vol.~62, no.~5, pp.
  2673--2681, May 2016.

\bibitem{calis2016general}
G.~Calis and O.~O. Koyluoglu, ``A general construction for pmds codes,''
  \emph{IEEE Communications Letters}, vol.~21, no.~3, pp. 452--455, 2016.

\bibitem{Horlemann-Trautmann2017}
\BIBentryALTinterwordspacing
A.-L. Horlemann-Trautmann and A.~Neri, ``{A Complete Classification of
  Partial-MDS (Maximally Recoverable) Codes with One Global Parity},'' Jul.
  2017. [Online]. Available: \url{http://arxiv.org/abs/1707.00847}
\BIBentrySTDinterwordspacing

\bibitem{roth2014coding}
R.~M. Roth and P.~O. Vontobel, ``Coding for combined block--symbol error
  correction,'' \emph{IEEE Transactions on Information Theory}, vol.~60, no.~5,
  pp. 2697--2713, 2014.

\bibitem{Pawar2011}
S.~Pawar, S.~{El Rouayheb}, and K.~Ramchandran, ``{Securing dynamic distributed
  storage systems against eavesdropping and adversarial attacks},'' \emph{IEEE
  Transactions on Information Theory}, vol.~57, no.~10, pp. 6734--6753, oct
  2011.

\bibitem{Han2012}
Y.~S. {Han} and R.~Z. and, ``Exact regenerating codes for byzantine fault
  tolerance in distributed storage,'' in \emph{2012 Proceedings IEEE INFOCOM},
  March 2012, pp. 2498--2506.

\bibitem{Dikaliotis2010}
T.~K. {Dikaliotis}, A.~G. {Dimakis}, and T.~{Ho}, ``Security in distributed
  storage systems by communicating a logarithmic number of bits,'' in
  \emph{2010 IEEE International Symposium on Information Theory}, June 2010,
  pp. 1948--1952.

\bibitem{Johnson1962}
S.~Johnson, ``{A new upper bound for error-correcting codes},'' \emph{IEEE
  Trans. Inf. Theory}, vol.~8, no.~3, pp. 203--207, Apr. 1962.

\bibitem{Guruswami1999}
V.~Guruswami and M.~Sudan, ``{Improved decoding of {R}eed-{S}olomon and
  algebraic-geometry codes},'' \emph{IEEE Trans. Inf. Theory}, vol.~45, no.~6,
  pp. 1757--1767, Sep. 1999.

\bibitem{Rudra2013}
A.~Rudra and M.~Wootters, ``Every list-decodable code for high noise has
  abundant near-optimal rate puncturings,'' in \emph{Proc. Forty-sixth Annual
  ACM Symp. on Theory of Computing}, New York, NY, USA, 2014, pp. 764--773.

\bibitem{bassalygo1965}
L.~A. Bassalygo, ``New upper bounds for error correcting codes,''
  \emph{Problemy Peredachi Informatsii}, vol.~1, no.~4, pp. 41--44, 1965.

\bibitem{sudan1997}
M.~Sudan, ``{Decoding of Reed Solomon codes beyond the error-correction
  bound},'' \emph{Journal of complexity}, vol.~13, no.~1, pp. 180--193, 1997.

\bibitem{kotter1996}
R.~Kötter, ``Fast generalized minimum-distance decoding of algebraic-geometry
  and {R}eed-{S}olomon codes,'' \emph{IEEE Transactions on Information Theory},
  vol.~42, no.~3, pp. 721--737, 1996.

\bibitem{roth2000}
R.~M. Roth and G.~Ruckenstein, ``{Efficient decoding of {R}eed-{S}olomon codes
  beyond half the minimum distance},'' \emph{IEEE Transactions on Information
  Theory}, vol.~46, no.~1, pp. 246--257, 2000.

\bibitem{Guruswami2012}
V.~Guruswami and C.~Xing, ``{List decoding {R}eed-{S}olomon,
  algebraic-geometric, and Gabidulin subcodes up to the Singleton bound},''
  \emph{Electronic Colloq. on Computational Complexity}, no. 146, 2012.

\bibitem{Cheung1988}
K.-M. Cheung and R.~J. McEliece, ``The undetected error probability for
  {R}eed-{S}olomon codes,'' in \emph{Military Comm. Conf.}, vol.~1, Oct. 1988,
  pp. 163--167.

\bibitem{cheung1989}
K.-M. Cheung, ``{More on the decoder error probability for {R}eed-{S}olomon
  codes},'' \emph{IEEE Transactions on Information Theory}, vol.~35, no.~4, pp.
  895--900, 1989.

\bibitem{McEliece1986}
{McEliece, R. J. and Swanson, L.}, ``{On the decoder error probability for Reed
  - Solomon codes},'' \emph{IEEE Trans. Inf. Theory}, vol.~32, no.~5, pp.
  701--703, Sep. 1986.

\bibitem{McEliece2003}
R.~J. McEliece, ``{The Guruswami–Sudan decoding algorithm for Reed–Solomon
  codes},'' \emph{IPN Progress Report}, vol. 42-153, 2003.

\bibitem{guruswami2008}
V.~Guruswami and A.~Rudra, ``{Explicit codes achieving list decoding capacity:
  Error-correction with optimal redundancy},'' \emph{IEEE Transactions on
  Information Theory}, vol.~54, no.~1, pp. 135--150, 2008.

\bibitem{parvaresh2005}
F.~Parvaresh and A.~Vardy, ``{Correcting errors beyond the Guruswami-Sudan
  radius in polynomial time},'' in \emph{46th Annual IEEE Symposium on
  Foundations of Computer Science (FOCS'05)}.\hskip 1em plus 0.5em minus
  0.4em\relax IEEE, 2005, pp. 285--294.

\bibitem{metzner1990general}
J.~J. Metzner and E.~J. Kapturowski, ``{A General Decoding Technique Applicable
  to Replicated File Disagreement Location and Concatenated Code Decoding},''
  \emph{IEEE Trans. Inform. Theory}, vol.~36, no.~4, pp. 911--917, 1990.

\bibitem{tamo2016optimal}
I.~Tamo, D.~S. Papailiopoulos, and A.~G. Dimakis, ``Optimal locally repairable
  des and connections to matroid theory,'' \emph{IEEE Transactions on
  Information Theory}, vol.~62, no.~12, pp. 6661--6671, 2016.

\bibitem{Silberstein2013}
N.~Silberstein, A.~S. Rawat, O.~O. Koyluoglu, and S.~Vishwanath, ``{Optimal
  locally repairable codes via rank-metric codes},'' \emph{IEEE International
  Symposium on Information Theory - Proceedings}, pp. 1819--1823, 2013.

\bibitem{krachkovsky1997decoding}
V.~Y. Krachkovsky and Y.~X. Lee, ``{Decoding for Iterative Reed--Solomon Coding
  Schemes},'' \emph{IEEE Transactions on Magnetics}, vol.~33, no.~5, pp.
  2740--2742, 1997.

\bibitem{krachkovsky1998decoding}
------, ``{Decoding of Parallel Reed--Solomon Codes with Applications to
  Product and Concatenated Codes},'' in \emph{IEEE International Symposium on
  Information Theory}, 1998, p.~55.

\bibitem{haslach1999decoding}
C.~Haslach and A.~H. Vinck, ``{A Decoding Algorithm With Restrictions for Array
  Codes},'' \emph{IEEE Transactions on Information Theory}, vol.~45, no.~7, pp.
  2339--2344, 1999.

\bibitem{justesen2004decoding}
J.~Justesen, C.~Thommesen, and T.~H{\o}holdt, ``{Decoding of Concatenated Codes
  with Interleaved Outer Codes},'' in \emph{IEEE International Symposium on
  Information Theory}, 2004, pp. 328--328.

\bibitem{schmidt2005interleaved}
G.~Schmidt, V.~R. Sidorenko, and M.~Bossert, ``{Interleaved Reed--Solomon Codes
  in Concatenated Code Designs},'' in \emph{IEEE Information Theory Workshop},
  2005, pp. 5--pp.

\bibitem{schmidt2009collaborative}
------, ``{Collaborative Decoding of Interleaved Reed--Solomon Codes and
  Concatenated Code Designs},'' \emph{IEEE Transactions on Information Theory},
  vol.~55, no.~7, pp. 2991--3012, 2009.

\bibitem{schmidt2010syndrome}
------, ``{Syndrome Decoding of Reed--Solomon Codes Beyond Half the Minimum
  Distance Based on Shift-Register Synthesis},'' \emph{IEEE Transactions on
  Information Theory}, vol.~56, no.~10, pp. 5245--5252, 2010.

\bibitem{kampf2014bounds}
S.~Kampf, ``{Bounds on Collaborative Decoding of Interleaved Hermitian Codes
  and Virtual Extension},'' \emph{Designs, Codes and Cryptography}, vol.~70,
  no. 1-2, pp. 9--25, 2014.

\bibitem{rosenkilde2018power}
J.~Rosenkilde, ``{Power Decoding Reed--Solomon Codes up to the Johnson
  Radius},'' \emph{Advances in Mathematics of Communications}, vol.~12, no.~1,
  pp. 81--106, 2018.

\bibitem{puchinger2019improved}
S.~Puchinger, J.~Rosenkilde, and I.~Bouw, ``{Improved Power Decoding of
  Interleaved One-Point Hermitian Codes},'' \emph{Designs, Codes and
  Cryptography}, vol.~87, no. 2-3, pp. 589--607, 2019.

\bibitem{elleuch2018interleaved}
\BIBentryALTinterwordspacing
M.~Elleuch, A.~{Wachter-Zeh}, and A.~Zeh, ``{A Public-Key Cryptosystem from
  Interleaved Goppa Codes},'' 2018. [Online]. Available:
  \url{http://arxiv.org/abs/1809.03024}
\BIBentrySTDinterwordspacing

\bibitem{Holzbaur2019crypto}
L.~Holzbaur, H.~Liu, S.~Puchinger, and A.~Wachter-Zeh, ``On decoding and
  applications of interleaved goppa codes,'' in \emph{2019 IEEE International
  Symposium on Information Theory (ISIT)}.\hskip 1em plus 0.5em minus
  0.4em\relax IEEE, 2019.

\bibitem{sidorenko2008decoding}
V.~Sidorenko, G.~Schmidt, and M.~Bossert, ``{Decoding Punctured Reed--Solomon
  Codes up to the Singleton Bound},'' in \emph{International ITG Conference on
  Source and Channel Coding}.\hskip 1em plus 0.5em minus 0.4em\relax VDE, 2008.

\bibitem{bleichenbacher2003decoding}
D.~Bleichenbacher, A.~Kiayias, and M.~Yung, ``{Decoding of Interleaved Reed
  Solomon Codes Over Noisy Data},'' in \emph{International Colloquium on
  Automata, Languages, and Programming}.\hskip 1em plus 0.5em minus 0.4em\relax
  Springer, 2003, pp. 97--108.

\bibitem{coppersmith2003reconstructing}
D.~Coppersmith and M.~Sudan, ``{Reconstructing Curves in Three (and Higher)
  Dimensional Space from Noisy Data},'' in \emph{ACM Symposium on the Theory of
  Computing}, 2003.

\bibitem{parvaresh2004multivariate}
F.~Parvaresh and A.~Vardy, ``{Multivariate Interpolation Decoding Beyond the
  Guruswami--Sudan Radius},'' in \emph{Allerton Conference on Communication,
  Control and Computing}, 2004.

\bibitem{brown2004probabilistic}
A.~Brown, L.~Minder, and A.~Shokrollahi, ``{Probabilistic Decoding of
  Interleaved RS-Codes on the q-Ary Symmetric Channel},'' in \emph{IEEE
  International Symposium on Information Theory}, 2004, pp. 326--326.

\bibitem{parvaresh2007algebraic}
F.~Parvaresh, ``{Algebraic List-Decoding of Error-Correcting Codes},'' Ph.D.
  dissertation, University of California, San Diego, 2007.

\bibitem{schmidt2007enhancing}
G.~Schmidt, V.~Sidorenko, and M.~Bossert, ``{Enhancing the Correcting Radius of
  Interleaved Reed--Solomon Decoding Using Syndrome Extension Techniques},'' in
  \emph{IEEE International Symposium on Information Theory}, 2007, pp.
  1341--1345.

\bibitem{cohn2013approximate}
H.~Cohn and N.~Heninger, ``{Approximate Common Divisors via Lattices},''
  \emph{The Open Book Series}, vol.~1, no.~1, pp. 271--293, 2013.

\bibitem{nielsen2013generalised}
J.~S. Nielsen, ``{Generalised Multi-Sequence Shift-Register Synthesis Using
  Module Minimisation},'' in \emph{2013 IEEE International Symposium on
  Information Theory}.\hskip 1em plus 0.5em minus 0.4em\relax IEEE, 2013, pp.
  882--886.

\bibitem{wachterzeh2014decoding}
A.~{Wachter-Zeh}, A.~Zeh, and M.~Bossert, ``{Decoding Interleaved Reed--Solomon
  Codes Beyond Their Joint Error-Correcting Capability},'' \emph{Designs, Codes
  and Cryptography}, vol.~71, no.~2, pp. 261--281, 2014.

\bibitem{puchinger2017irs}
S.~Puchinger and J.~{Rosenkilde n\'e Nielsen}, ``{Decoding of Interleaved
  Reed--Solomon Codes Using Improved Power Decoding},'' in \emph{IEEE
  International Symposium on Information Theory}, 2017.

\bibitem{yu2018simultaneous}
J.-H. Yu and H.-A. Loeliger, ``{Simultaneous Partial Inverses and Decoding
  Interleaved Reed--Solomon Codes},'' \emph{IEEE Transactions on Information
  Theory}, vol.~64, no.~12, pp. 7511--7528, 2018.

\bibitem{Guruswami2006}
V.~Guruswami, ``{Algorithmic results in list decoding},'' \emph{Found. Trends
  Theor. Comput. Sci.}, vol.~2, no.~2, pp. 107--195, 2006.

\bibitem{goparaju2014binary}
S.~Goparaju and R.~Calderbank, ``Binary cyclic codes that are locally
  repairable,'' in \emph{2014 IEEE International Symposium on Information
  Theory}.\hskip 1em plus 0.5em minus 0.4em\relax IEEE, 2014, pp. 676--680.

\bibitem{augot2011list}
D.~Augot, M.~Barbier, and A.~Couvreur, ``List-decoding of binary goppa codes up
  to the binary johnson bound,'' in \emph{2011 IEEE Information Theory
  Workshop}.\hskip 1em plus 0.5em minus 0.4em\relax IEEE, 2011, pp. 229--233.

\bibitem{Overbeck_Diss_InterleveadGab}
R.~Overbeck, ``{Public Key Cryptography based on Coding Theory},'' Ph.D.
  dissertation, TU Darmstadt, Darmstadt, Germany, 2007.

\bibitem{couveignes2009elliptic}
J.-M. Couveignes and R.~Lercier, ``{Elliptic Periods for Finite Fields},''
  \emph{{Finite Fields and Their Applications}}, vol.~15, no.~1, pp. 1--22,
  2009.

\bibitem{migler2004}
T.~Migler, K.~E. Morrison, and M.~Ogle, ``Weight and rank of matrices over
  finite fields,'' \emph{arXiv preprint math/0403314}, 2004.

\bibitem{harvey2019faster}
D.~Harvey and J.~van~der Hoeven, ``{Faster Integer Multiplication Using Short
  Lattice Vectors},'' \emph{The Open Book Series}, vol.~2, no.~1, pp. 293--310,
  2019.

\bibitem{zeh2016improved}
A.~Zeh and A.~Wachter-Zeh, ``Improved erasure list decoding locally repairable
  codes using alphabet-dependent list recovery,'' in \emph{2016 IEEE
  International Symposium on Information Theory (ISIT)}.\hskip 1em plus 0.5em
  minus 0.4em\relax IEEE, 2016, pp. 1581--1585.

\bibitem{wei1991generalized}
V.~K. Wei, ``Generalized hamming weights for linear codes,'' \emph{IEEE
  Transactions on information theory}, vol.~37, no.~5, pp. 1412--1418, 1991.

\bibitem{guruswami2003list}
V.~Guruswami, ``List decoding from erasures: Bounds and code constructions,''
  \emph{IEEE Transactions on Information Theory}, vol.~49, no.~11, pp.
  2826--2833, 2003.

\bibitem{cohen1994upper}
G.~Cohen, S.~Litsyn, and G.~Z{\'e}mor, ``Upper bounds on generalized
  distances,'' \emph{IEEE Transactions on Information Theory}, vol.~40, no.~6,
  pp. 2090--2092, 1994.

\bibitem{tsfasman1995geometric}
M.~A. Tsfasman and S.~G. Vladut, ``Geometric approach to higher weights,''
  \emph{IEEE Transactions on Information Theory}, vol.~41, no.~6, pp.
  1564--1588, 1995.

\bibitem{ashikhmin1999new}
A.~Ashikhmin, A.~Barg, and S.~Litsyn, ``New upper bounds on generalized
  weights,'' \emph{IEEE Transactions on Information Theory}, vol.~45, no.~4,
  pp. 1258--1263, 1999.

\bibitem{prakash2012optimal}
N.~Prakash, G.~M. Kamath, V.~Lalitha, and P.~V. Kumar, ``Optimal linear codes
  with a local-error-correction property,'' in \emph{2012 IEEE International
  Symposium on Information Theory Proceedings}.\hskip 1em plus 0.5em minus
  0.4em\relax IEEE, 2012, pp. 2776--2780.

\bibitem{prakash2014codes}
N.~Prakash, V.~Lalitha, and P.~V. Kumar, ``Codes with locality for two
  erasures,'' in \emph{2014 IEEE International Symposium on Information
  Theory}.\hskip 1em plus 0.5em minus 0.4em\relax IEEE, 2014, pp. 1962--1966.

\bibitem{hao2017weight}
J.~Hao, S.-T. Xia, B.~Chen, and F.-W. Fu, ``On the weight hierarchy of locally
  repairable codes,'' in \emph{2017 IEEE Information Theory Workshop
  (ITW)}.\hskip 1em plus 0.5em minus 0.4em\relax IEEE, 2017, pp. 31--35.

\bibitem{lalitha2015weight}
V.~Lalitha and S.~V. Lokam, ``Weight enumerators and higher support weights of
  maximally recoverable codes,'' in \emph{2015 53rd Annual Allerton Conference
  on Communication, Control, and Computing (Allerton)}.\hskip 1em plus 0.5em
  minus 0.4em\relax IEEE, 2015, pp. 835--842.

\end{thebibliography}

\appendix
\subsection{List Decoding of Erasures in Locally Repairable Codes}
\label{app:erasure_list_decoding}

The focus of Section~\ref{sec:listDecoding} is list decoding of \emph{errors} in LRCs, but for the sake of completeness we include a short discussion on \emph{erasure} list decoding in LRCs. Similar to list decoding of errors, the objective of list decoding erasures is to output a list containing all codewords that could result in the received word for the given erasure pattern, i.e., that agree with the received word in all non-erased positions. However, while it is a major challenge to find efficient algorithms for list decoding of errors, the output of a list decoder for erasures is simply the solution space of an (underdetermined) linear system of equations. The relation between the maximum size of this solution space, i.e., the maximum list size of the decoder, and the $i$-th \emph{generalized Hamming weight} \cite{wei1991generalized}
\begin{equation*}
  d_i = \min\{|\supp(\mathcal{D})| : \mathcal{D} \subseteq \code \ \text{and} \ \dim(\mathcal{D}) = i\}
\end{equation*}
of a code $\code$ was made explicit in \cite{guruswami2003list}.
These values have been studied intensely for several code classes in the context of different applications, see, e.g., \cite{cohen1994upper,tsfasman1995geometric,ashikhmin1999new}. In particular, they were studied for codes with locality \cite{prakash2012optimal,prakash2014codes,hao2017weight} and the subclass of PMDS codes \cite{lalitha2015weight,Gopalan2014}.

A Singleton-like bound on the generalized Hamming weights of an $[n,k,r,2]$ LRC is given by \cite{hao2017weight}
\begin{equation}\label{eq:LRCgeneralizedHW}
  d_i\leq n-k -\left\lceil \frac{k-i+1}{r} \right\rceil +i+1 \ , \ 1\leq i \leq k \ ,
\end{equation}
and this bound is achieved with equality for optimal $(r,2)$-LRCs.
Applying this result to \cite[Lemma~1]{guruswami2003list} gives an upper bound on the list size for all optimal $[n,k,r,2]$ LRCs.
\begin{lemma}[Erasure list decoding of optimal LRCs \cite{hao2017weight,guruswami2003list}]
  An optimal $q$-ary $[n,k,r,2]$ LRC can correct $t$ erasures with maximal list size $L$ if and only if $d_{1+\floor{\log_q L}} > t$, where $d_i$ is given by (\ref{eq:LRCgeneralizedHW}) with equality.
\end{lemma}
Erasure list decoding of LRCs has also been considered in \cite{zeh2016improved}, where erasure list decoding in the local repair sets is combined with alphabet-dependent list recovery.

Note that the probability of obtaining a list of maximum size depends on the explicit code used. For example, for a $q$-ary MDS codes it is easy to see that list decoding of $n-k+\delta$ erasures \emph{always} gives a list of size equal to the upper bound of $q^\delta$. For erasure list decoding the probability of a unique output, i.e., a list of size one, can easily be characterized as the probability of the non-erased positions containing an information set of the code, i.e., the probability that the generator matrix restricted to the non-erased positions is of full rank. While it is not possible to determine this probability for optimal LRCs in general, we provide a detailed analysis for the subclass of PMDS codes in Section~\ref{sec:PMDS} by first relating this property to specific subsets of code positions in Lemma~\ref{lem:information_sets_Sk} and then bounding the number of such sets\footnote{As we consider a different problem in Section~\ref{sec:PMDS}, the bound is stated for $\Sset_{k+1}$ instead of $\Sset_{k}$.} in Lemma~\ref{lem:Partial_MDS_bound_S_k+1}.

\subsection{Proof of Lemma~\ref{lem:increasingInN}}\label{app:proof_of_lemma_increasingInN}

\begin{IEEEproof}[Proof of Lemma~\ref{lem:increasingInN}]
For ease of notation define $a:=\frac{d}{\theta_q n}$. Then
\begin{align*}
  h(a) = d a^{-1} \left(1-(1-a)^{\frac{\ell}{\ell+1}}\right)
\end{align*}
and observe the equivalence of the conditions
\begin{equation*}
  n\geq \frac{d}{\theta_q} \quad \Longleftrightarrow \quad 0< a \leq 1 \ .
\end{equation*}
As $a$ is a decreasing function in $n$, any function that is increasing in $a$, is decreasing in $n$. The partial derivative is given by
\begin{align*}
  \frac{\partial}{\partial a} h(a) = \frac{(1-a)^{-\frac{1}{\ell+1}} \frac{\ell}{\ell+1} a - (1-(1-a)^{\frac{\ell}{\ell+1}})}{a^2} \ .
\end{align*}
As $a>0$ the partial derivative $\frac{\partial}{\partial a} h(a) \geq 0$ is positive if for the numerator it holds that
\begin{align}
0 &\leq (1-a)^{-\frac{1}{\ell+1}} \frac{\ell}{\ell+1} a +(1-a)^{\frac{\ell}{\ell+1}}  -1 \nonumber \\
&= (1-a)^{-\frac{1}{\ell+1}}\left(\frac{\ell}{\ell+1} a + (1-a)\right) -1 \nonumber \\
&= (1-a)^{-\frac{1}{\ell+1}}\left(1-\frac{a}{\ell+1}\right) -1 \nonumber \\
&= \frac{1-\frac{a}{\ell+1}-(1-a)^{\frac{1}{\ell+1}}}{\underbrace{(1-a)^{\frac{1}{\ell+1}}}_{>0}}  \ . \label{eq:finalha}
\end{align}
The limit of the enumerator for $a \rightarrow 0^+$ is zero and its derivative is given by
\begin{align*}
 \frac{\partial}{\partial a} \left(1-\frac{a}{\ell+1}-(1-a)^{\frac{1}{\ell+1}}\right) &= -\frac{1}{\ell+1} + \frac{1}{\ell+1} (1-a)^{-\frac{\ell}{\ell+1}} \\
                                                                         &= \frac{1}{\ell+1}\big( \underbrace{(1-a)^{-\frac{\ell}{\ell+1}}}_{>1}-1 \big) > 0 , \ \text{for} \ 0 < a < 1 \ .
\end{align*}
It follows that the enumerator and denominator of (\ref{eq:finalha}) are positive for $0<a<1$, i.e., the first derivative $\frac{\partial}{\partial a} h(a)$ is positive in this range. Further, it is easy to check that $h(a) < h(1) = d, \ \forall 0<a<1$. Hence $h(a)$ is increasing in $a$ for $0<a\leq 1$ and thereby decreasing in $n$ for $n\geq \frac{d}{\theta_q}$.
\end{IEEEproof}

\subsection{Improved List Size Bound}\label{app:improved_list_bound}

The following theorem presents an improved bound on the maximal list size for list decoding LRCs up to the decoding radius in Theorem~\ref{thm:ListDecodingLRCs} in Section~\ref{subsec:newdecodingradius}.
It is a slight (first-order) improvement over the bound \eqref{eq:listSize} in Theorem~\ref{thm:ListDecodingLRCs} and further improvements might be possible (see Remark~\ref{rem:list_size_bound_not_tight} in Section~\ref{subsec:newdecodingradius} for more details).

\begin{theorem}\label{thm:ListDecodingLRCsImproved}
Consider the setting of Theorem~\ref{thm:ListDecodingLRCs} in Section~\ref{subsec:newdecodingradius}.
Let $L_{(n,d,\tau)}$ denote the maximum list size when list decoding an~$[n,k,d]_q$ code with radius~$\tau$. Then an $[n,k,r,\varrho]_q$ LRC is~$(\tau_g,L_g)$-list-decodable, where $\tau_g$ is the decoding radius in \eqref{eq:jblrc} (see Theorem~\ref{thm:ListDecodingLRCs}), and the list size is upper-bounded by
\begin{equation}\label{eq:listSizeImproved}
  L_g \leq \binom{\frac{n}{n_l}}{\ceil{\sigma}} \max_{\xi =0,\dots,\ceil{\sigma}} \left\{L_{(n_l,\varrho,\tau_{J,l})}^{\xi} L_{(n-\ceil{\sigma} n_l,d,\tau_g-\xi(\varrho-\tau_{J,l}))}\right\} \ .
\end{equation}
\end{theorem}

\begin{IEEEproof}
We build on the proof of Theorem~\ref{thm:ListDecodingLRCs} and only mention the differences.
Suppose that we have locally decoded up to radius $\tau_{J,l}$ in each repair set independently.
Recall that for each codeword in distance $\tau_g$ to the received word, there are at least $\ceil{\sigma}$ local repair sets in which the local decoder ouput list contains the correct local codewords in its output list.
Hence, as in the proof of Theorem~\ref{thm:ListDecodingLRCs}, we go through all $\binom{\frac{n}{n_l}}{\ceil{\sigma}}$ subsets of $\ceil{\sigma}$ local repair sets and estimate the sum of the output list size after shortening and global decoding w.r.t.\ all combinations of local decoding results in this subset of local repair sets.

The difference to Theorem~\ref{thm:ListDecodingLRCs} is how we count this sum of output lists.
We distinguish between several cases and finally take the maximum of all these cases.
We use the following \emph{sub-claim}:
Let $\Code_\mathsf{local}[r+\varrho-1, r, \varrho]$ be a local code and let $\ve{r} \in \Fq^{r+\varrho-1}$ be a vector such that the ball $\mathcal{B}_{\tau_{J,l}}(\ve{r})$ around $\ve{r}$ of radius $\tau_{J,l}$ contains more than one codeword of $\Code_\mathsf{local}$. Then we have $\dH(\ve{r},\c) \geq \varrho-\tau_{J,l}$ for all $\c \in \mathcal{B}_{\tau_{J,l}}(\ve{r}) \cap \Code_\mathsf{local}$.

\emph{Proof of the sub-claim:}
Assume the contrary, i.e., that there are distinct codewords $\c,\c' \in \mathcal{B}_{\tau_{J,l}}(\ve{r}) \cap \Code_\mathsf{local}$ such that $\dH(\ve{r},\c) < \varrho-\tau_{J,l}$. By the triangular inequality, we have
\begin{equation*}
\dH(\c,\c') \leq \dH(\c,\ve{r}) + \dH(\ve{r},\c') < \varrho,
\end{equation*}
contradicting the assumption $\c \neq \c'$ since $\varrho$ is the minimum distance of $\Code_\mathsf{local}$. This proves the sub-claim.

Fix a subset of $\ceil{\sigma}$ local repair sets and suppose that exactly $\xi \in \{0,1,\dots,\ceil{\sigma}\}$ repair sets have list size $>1$ (note that exactly one of these cases is fulfilled).
Hence, there are at most $L_{(n_l,\varrho,\tau_{J,l})}^{\xi}$ combinations of local codewords that we need to consider for shortening and global decoding.
Furthermore, by the \emph{sub-claim}, each of these combinations of local codewords has (summed) distance at least $\xi(\varrho-\tau_{J,l})$ to the corresponding local repair sets of the received word. After shortening, we have therefore removed several errors and only at most $\tau_g-\xi(\varrho-\tau_{J,l})$ errors remain to be corrected by the shortened code.
Hence, the list size is reduced from $L_{(n-\ceil{\sigma} n_l,d,\tau_g)}$ to $L_{(n-\ceil{\sigma} n_l,d,\tau_g-\xi(\varrho-\tau_{J,l}))}$ compared to decoding with radius $\tau_g$, and the overall list size for this subset of local repair sets is upper-bounded by
\begin{equation*}
L_{(n_l,\varrho,\tau_{J,l})}^{\xi}L_{(n-\ceil{\sigma} n_l,d,\tau_g-\xi(\varrho-\tau_{J,l}))}.
\end{equation*}
By taking the maximum value over all possible cases for $\xi$, we obtain the claimed overall list size bound.
\end{IEEEproof}

\begin{remark}\label{rem:improved_list_size_bound_no_significant_improvement}
Note that
\begin{align*}
\max_{\xi =0,\dots,\ceil{\sigma}} \left\{L_{(n_l,\varrho,\tau_{J,l})}^{\xi} L_{(n-\ceil{\sigma} n_l,d,\tau_g-\xi(\varrho-\tau_{J,l}))}\right\} \leq L_{(n_l,\varrho,\tau_{J,l})}^{\ceil{\sigma}} L_{(n-\ceil{\sigma} n_l,d,\tau_g)},
\end{align*}
so the list-size bound \eqref{eq:listSizeImproved} in Theorem~\ref{thm:ListDecodingLRCsImproved} is always at least as good as the bound \eqref{eq:listSize} in Theorem~\ref{thm:ListDecodingLRCs} (Section~\ref{subsec:newdecodingradius}).

On the other hand, we have
\begin{align*}
\max_{\xi =0,\dots,\ceil{\sigma}} \left\{L_{(n_l,\varrho,\tau_{J,l})}^{\xi} L_{(n-\ceil{\sigma} n_l,d,\tau_g-\xi(\varrho-\tau_{J,l}))}\right\} \geq L_{(n_l,\varrho,\tau_{J,l})}^{\ceil{\sigma}},
\end{align*}
so we save at most a factor $L_{(n-\ceil{\sigma} n_l,d,\tau_g)}$ compared to the bound \eqref{eq:listSize} in Theorem~\ref{thm:ListDecodingLRCs}. In particular, the two factors of \eqref{eq:listSize} that are exponential in $\ceil{\sigma}$ (and thus in the number of local repair sets $\frac{n}{n_l}$) are not eliminated by the improvement.
\end{remark}

\end{document}